\documentclass[prx,aps,%
 reprint,%
 twocolumn,
superscriptaddress,
nofootinbib,
longbibliography]{revtex4-2}

\usepackage{graphicx}
\usepackage{tikz}

\usepackage{amsfonts}
\usepackage{amssymb}
\usepackage{amsmath}
\usepackage{amsthm}
\usepackage{color}
\usepackage{xcolor}
\usepackage[colorlinks=true,linkcolor=dblue,citecolor=dblue,urlcolor=dblue,plainpages=false,pdfpagelabels]{hyperref}

\usepackage{mathtools}
\usepackage{bbm}

\usepackage{float}
\usepackage{array}

\usepackage{verbatim}

\usepackage[utf8]{inputenc}
\usepackage[T1]{fontenc}
\usepackage{etoolbox}

\usepackage{enumitem}

\usepackage{algorithm}
\usepackage{algpseudocode}
\floatname{algorithm}{Protocol}

\algrenewcommand{\algorithmiccomment}[1]{\hspace{-1.2cm} #1}

\usepackage{anyfontsize}  
\usepackage{microtype} 

\usepackage{newtxtext}
\usepackage{newtxmath}

\makeatletter
\g@addto@macro\bfseries{\boldmath}
\makeatother

\newcommand{\bra}[1]{\langle #1|}
\newcommand{\ket}[1]{|#1\rangle}
\newcommand{\braket}[2]{\langle #1|#2\rangle}
\newcommand{\ketbra}[2]{\ket{#1}\!\bra{#2}}

\newcommand{\Abs}[1]{\left|#1\right|}
\newcommand{\e}{\mathrm{e}}
\newcommand{\I}{\mathrm{i}}

\renewcommand{\t}{{\scriptscriptstyle\mathsf{T}}}
\newcommand{\id}{\operatorname{id}}

\DeclarePairedDelimiter{\ceil}{\lceil}{\rceil}

\theoremstyle{definition}
\newtheorem{theorem}{Theorem}

\newtheorem{proposition}[theorem]{Proposition}
\newtheorem{remark}{Remark}

\renewcommand{\qedsymbol}{$\blacksquare$}

\renewcommand{\qedsymbol}{\unskip\nobreak\quad\qedsymbol}
\renewcommand{\qedsymbol}{$\blacksquare$}
\newcommand{\qedgen}{$\blacktriangleleft$}

\definecolor{dblue}{RGB}{14, 34, 102}

\makeatletter
\def\maketitle{
\@author@finish
\title@column\titleblock@produce
\suppressfloats[t]}
\makeatother

\begin{document}

\title{A resource- and computationally-efficient protocol for\\multipartite entanglement distribution in Bell-pair networks}

\author{S.~Siddardha Chelluri}\email{schellur@uni-mainz.de}
\affiliation{Institute of Physics, Johannes-Gutenberg University of Mainz, Staudingerweg 7, 55128 Mainz, Germany}

\author{Sumeet Khatri}\email{skhatri@vt.edu}
\affiliation{Dahlem Center for Complex Quantum Systems, Freie Universit\"{a}t Berlin, 14195 Berlin, Germany}
\affiliation{Department of Computer Science, Virginia Tech, Blacksburg, VA 24061, USA}
\affiliation{Virginia Tech Center for Quantum Information Science and Engineering, Blacksburg, VA 24061, USA}

\author{Peter van Loock}\email{loock@uni-mainz.de}
\affiliation{Institute of Physics, Johannes-Gutenberg University of Mainz, Staudingerweg 7, 55128 Mainz, Germany}

\date{\today}

\let\oldaddcontentsline\addcontentsline 
\renewcommand{\addcontentsline}[3]{\oldaddcontentsline{#1}{lot}{#3}} 

\begin{abstract}

Multipartite entangled states, such as Greenberger--Horne--Zeilinger (GHZ) states, are important resources in multiparty quantum networking tasks. We consider protocols for generating such states from networks of Bell pairs and local operations and classical communication. We present a computationally-efficient protocol for generating GHZ states that is also efficient with respect to the number of consumed Bell pairs, (local) gates, and Bell-pair sources. Our protocol: (1) requires $O(N)$ gates in a network with $N$ nodes, independent of the network topology; (2) has time complexity $O(N^2)$, avoiding the Steiner tree and any other computationally-hard problem; (3) maintains a near-optimal number of consumed Bell pairs. Numerically, our protocol outperforms those based on (approximate) Steiner trees with respect to number of gates and Bell-pair sources. We prove that the minimal Bell-pair source cost is given by solving the graph-theoretic dominating set problem, and we demonstrate numerically that our protocol is nearly optimal for this quantity. Finally, we analytically characterize the impact of noisy Bell pairs and gates on the fidelity of the distributed GHZ states.

\end{abstract}

\maketitle


\section*{Introduction}

Recent years have seen a rapid rise in the development and realization of small-scale quantum devices, such as small-scale quantum computers and quantum sensors~\cite{pirandola2018sensingadvances,tse2019quantumenhancedLIGO,gambetta2020ibm,luo2023progressphotonicschips}. Individually, these small-scale devices are not particularly powerful nor useful. However, by \textit{networking} these small-scale devices, it may become possible to execute large-scale applications in a distributed fashion~\cite{CCT+20,awschalom2021interconnects,humble2021quantumcomputersHPC}, and thus make strides towards the goal of quantum advantage.


Entanglement is the resource that enables distributed quantum information processing tasks, such as quantum communication, computation, and sensing, over long distances. Both bipartite and multipartite entanglement can be used to achieve these tasks. Indeed, bipartite entanglement, along with local operations and classical communication (LOCC) between spatially separated nodes, enables the fundamental protocol of teleportation~\cite{BBC+93,BBP96}, and multipartite entanglement enables various distributed and measurement-based quantum computing strategies~\cite{gottesman1999gateteleportation,eisert2000nonlocalgates,RB01,leung2001MBQC,nielsen2003MBQC,leung2004MBQC,josza2005introMBQC,danos2007distributedMBQC,piveteau2024circuitknitting,broadbent2009blindQC}, and can enhance quantum sensing~\cite{toth2012multipartitemetrology,hyllus2012multipartitemetrology,ZZS18,XZCZ19,guo2020distributedsensing}.

\begin{figure}
    \centering
    \includegraphics[width=0.97\columnwidth]{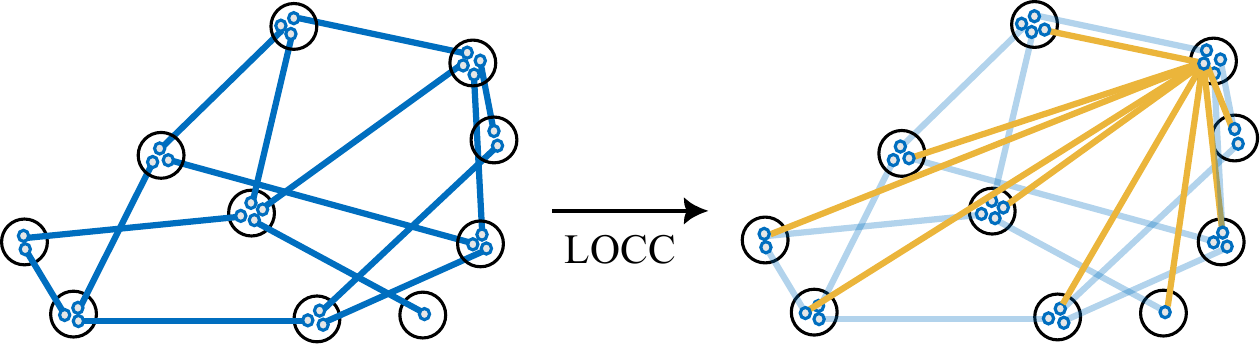}
    \caption{\textbf{Creation of a GHZ state from a network of Bell pairs and LOCC.} (Left) A network of Bell pairs, consisting of nodes (black circles) containing qubits (blue circles). Every edge connecting two qubits represents a Bell pair. Our goal is to transform the Bell-pair network, via local operations and classical communication (LOCC), into a graph state. (Right) A star graph state (which is local-unitary-equivalent to a GHZ state), indicated in yellow, shared by all nodes in the network.}
    \label{fig:LOCC_Bell_pairs_to_multipartite}
\end{figure}

The creation of bipartite entanglement (e.g., Bell pairs) between nearest neighbors in a network has been realized experimentally in recent years~\cite{DSG+16,SSH+17,Hump+18,PHB+21,hermans2022teleportation,pompili2022demonstrationstack,exp_1,knaut2024entanglementnetwork,exp_3,exp_4,exp_5}, and will arguably form the backbone of large-scale quantum networks. Therefore, when considering the distribution of multipartite entanglement in quantum networks, it is meaningful to consider protocols in which the aim is to transform a network of Bell pairs, via LOCC, into a multipartite entangled state shared by a given set of nodes in the network; see Fig.~\ref{fig:LOCC_Bell_pairs_to_multipartite}. Some prior works assume that these multipartite entangled states are already distributed in advance in the network~\cite{PWD18,PD19,HPE19,freund2024quantumdatabus}, without consideration of how this might be done in the first place. Other works consider first creating the multipartite states locally, and then distributing them over quantum channels and/or via teleportation~\cite{CC12,EKB16,fischer2021graph,avis2023multipartitecentral} using the network of nearest-neighbor Bell pairs. The latter strategy requires quantum repeater protocols~\cite{BDC98,DBC99,SSR+11,ABC21,azuma2023repeatersRMP} for creating Bell pairs between one node and all of the others, which is in general costly in terms of the number of Bell pairs and local qubits.

In this work, we consider LOCC-based entanglement distribution of multipartite entangled states directly from a network of Bell pairs, as illustrated in Fig.~\ref{fig:LOCC_Bell_pairs_to_multipartite}. Previously, Ref.~\cite{CC12} has provided a protocol in this setting that creates a graph state shared among all nodes in the network, such that the topology of the corresponding graph mimics that of the Bell-pair network. They have shown that their protocol can outperform the repeater-based protocol that first creates bipartite entanglement between one node and all of the others and then teleports the locally-created graph state. Similarly, Ref.~\cite{MMG19} provides a protocol for distributing an arbitrary graph state using a network of Bell pairs, even among a chosen subset of nodes, but this protocol involves (in the most general case) solving the NP-hard Steiner tree problem~\cite{hwang1992Steinertree_book,Brazil2015} as an initial step. The protocol in Ref.~\cite{fischer2021graph} involves generating the graph state locally at a node and then distributing the state to the other nodes via appropriately selected paths using the Bell-pair network. We refer to the Supplementary Information for a more detailed summary of prior work.

\begin{table}[t]
    \includegraphics[width=\columnwidth]{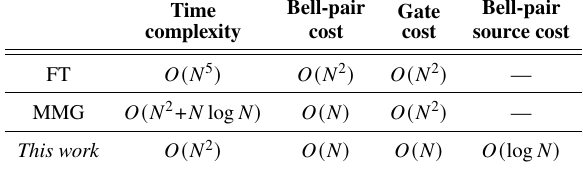}
    \caption{\textbf{Summary of results and comparison to prior work.} The number of nodes in the network is $N$. These results are specific to GHZ states. Both FT~\cite{fischer2021graph} and MMG~\cite{MMG19} consider the creation of arbitrary graph states.
    The gate cost refers to the cost of local two-qubit gates, as these outweigh single-qubit gates in terms of noise, and therefore cost. The Bell-pair cost is for the creation of a GHZ state in the entire network.
    The Bell-pair source cost follows the asymptotic behavior of the Erdős–Rényi model, as established in Refs.~\cite{dom_set_bound_2001,dom_set_bound_2015}. 
    }
    \label{tab:summary_comparisons}
\end{table}

The main contribution of our work is an alternative protocol for distributing graph states, specifically Greenberger--Horne--Zeilinger (GHZ) states~\cite{GHZ89}, in arbitrary network topologies directly from the Bell pairs that connect the nearest-neighbor nodes of the network. GHZ states are useful, for instance, for conference key agreement~\cite{murta2020CKAreview}, and are necessary for achieving optimal precision in certain quantum metrological tasks~\cite{toth2012multipartitemetrology,hyllus2012multipartitemetrology}. Unlike Ref.~\cite{MMG19}, our protocol is not based on solving the Steiner tree problem, nor any other computationally hard problem. Also, unlike Ref.~\cite{fischer2021graph}, our protocol does not require first creating Bell pairs between one node in the network and all of the others.
Nevertheless, we can provably achieve either comparable or better performance compared to Refs.~\cite{MMG19,fischer2021graph}; see Table~\ref{tab:summary_comparisons} for a summary.

The usual figures of merit for assessing the performance of entanglement distribution protocols include the number of consumed Bell pairs, waiting time, and fidelity. In this work, we consider also the computational efficiency, the total number of (local) gates, and the number of Bell-pair sources. These additional figures of merit are motivated by the fact that gates are currently prone to errors, which means that applying more gates comes at the cost of reduced fidelity. The amount of classical communication, particularly between distant nodes for the purpose of communicating measurement outcomes, is also constrained in practice due to limitations on the coherence times of quantum memories. From these facts, we can conclude that LOCC operations are not free in practice, unlike the usual information-theoretic setting of entanglement distribution in which LOCC operations are considered free. (We note that LOCC operations are considered free in Refs.~\cite{MMG19,fischer2021graph}.) Furthermore, when we envision the development of large-scale quantum networks, the question of how many Bell-pair sources are needed in the network naturally arises, and it is meaningful to consider keeping the number of Bell-pair sources to a minimum. With all of these points in mind, we are motivated to seek protocols that are efficient with respect to the following figures of merit: (1) computational efficiency; (2) number of gates applied; and (3) number of Bell-pair sources.

Our protocol 
is efficient with respect to all three of the above figures of merit. It is based on successively creating small GHZ states, starting from the highest-degree nodes, and then merging them to obtain larger GHZ states that eventually span either the entire network or some given subset of nodes. 
We assess the performance of our protocol on models of real networks, which exhibit the small-world property, along with high density and central hubs. We consider the Erd\H{o}s--R\'{e}nyi, Barab\'{a}si--Albert, and Waxman models~\cite{erdosrenyi1959,erdosrenyi1960,waxman1988,wattsstrogatz1998,albert1999diameterWWW,barabasi1999scalefree,Barabasi16_book,Newman2018_book}, which capture these realistic network features.

We prove that our protocol has time complexity $O(N^2)$ and 
requires $O(N)$ gates to produce a GHZ state shared by all nodes in the network. On the other hand, we prove that the protocol of Ref.~\cite{MMG19} requires at most $O(N^2)$ gates, and
we provide numerical evidence that our protocol indeed improves the gate count compared to Ref.~\cite{MMG19}. 
It is also remarkable that the gate count for our protocol depends \textit{only} on the number $N$ of nodes, and not on any other property of the network, e.g., topology and density of nodes.

We also prove
that our protocol consumes the same optimal number of Bell pairs as in Ref.~\cite{MMG19} when entangling all nodes in the network. 
When entangling only a subset of nodes, we provide numerical evidence that our protocol is nearly optimal with respect to consumed Bell pairs, i.e., the number of consumed Bell pairs in our protocol is close to that of the Steiner tree approach of Ref.~\cite{MMG19}. As such, our protocol represents a heuristic algorithm for the Steiner tree problem.

Regarding the number of Bell-pair sources, our protocol has the notable feature that it automatically prescribes the location of the Bell-pair sources. Indeed, by creating small-scale GHZ states at the highest-degree nodes within the network, these sources only need to be located at the highest-degree nodes, rather than at every individual node or at the midpoint of every edge of the network. It is then natural to ask whether this placement of Bell-pair sources is optimal. 
We prove that minimizing the number of Bell-pair sources is equivalent to the \textit{minimum dominating set problem} from graph theory~\cite{dom_set_alg_book}, which is an NP-hard problem in general~\cite{dom_set_np}. This result represents an interesting connection to graph theory, analogous to the well-known connection between the minimal Bell-pair cost and the Steiner tree problem. With this result, we are able to show numerically 
that our protocol is nearly optimal with respect to the number of Bell-pair sources. In other words, our protocol provides a heuristic algorithm for both the Steiner tree and minimum dominating set problems. Furthermore, by using an approximate algorithm for the minimum dominating set problem~\cite{dom_set_alg_book}, we evaluate the Bell-pair source counts of algorithms based on the Steiner tree, and we find 
that our protocol outperforms these existing algorithms for both the Erd\H{o}s--R\'{e}nyi and Barab\'{a}si--Albert network models.

Finally, 
we consider our protocol in the presence of noise, both noisy gates and noisy Bell pairs, and we derive exact analytical expressions for the fidelity of the output state with respect to the desired GHZ state. These analytical expressions hold for arbitrary noise models for both the Bell pairs and the gates.

\section*{Results}

\subsection*{Our protocol}

We model a network with a graph. The vertices of the graph correspond to the nodes in the network. Every node has a number of qubits equal to the degree of the corresponding vertex. The edges in the graph dictate how these qubits within the nodes are connected to the qubits in other nodes via quantum channels, which are used to generate Bell pairs $\ket{\Phi}\coloneqq\frac{1}{\sqrt{2}}(\ket{0,0}+\ket{1,1})$ between the qubits in different nodes. In particular, if the graph has a vertex with degree $k$, then its corresponding network has a node with $k$ qubits, and each one of these qubits can be connected via a Bell pair to a qubit in the neighboring node; see Fig.~\ref{fig:LOCC_Bell_pairs_to_multipartite} for an example. Starting from the network of Bell pairs, we consider the task of distributing a GHZ state among either all nodes of the network (the \textit{complete case}) or a given subset of nodes in the network (the \textit{subset case}). Every qubit of the GHZ state should belong to one of the desired nodes.

The $N$-party GHZ state~\cite{GHZ89} is defined as $\ket{\text{GHZ}_N}\coloneqq\frac{1}{\sqrt{2}}(\ket{0}^{\otimes N}+\ket{1}^{\otimes N})$.
Graphically, GHZ states can be depicted using a star graph; see Fig.~\ref{fig:GHZ_Protocols}. Notably, because the GHZ state is permutation invariant, the center of the star can be placed at any of the qubits, a fact that we make use of extensively in our protocol. (This fact is in contrast to the star graph state, which has a distinguished central qubit, and to change the central qubit requires local operations.) Throughout the rest of this work, we use the term ``star'' to refer to both a GHZ state and its graphical depiction.

\begin{figure}
    \centering
    \includegraphics[width=0.85\columnwidth]{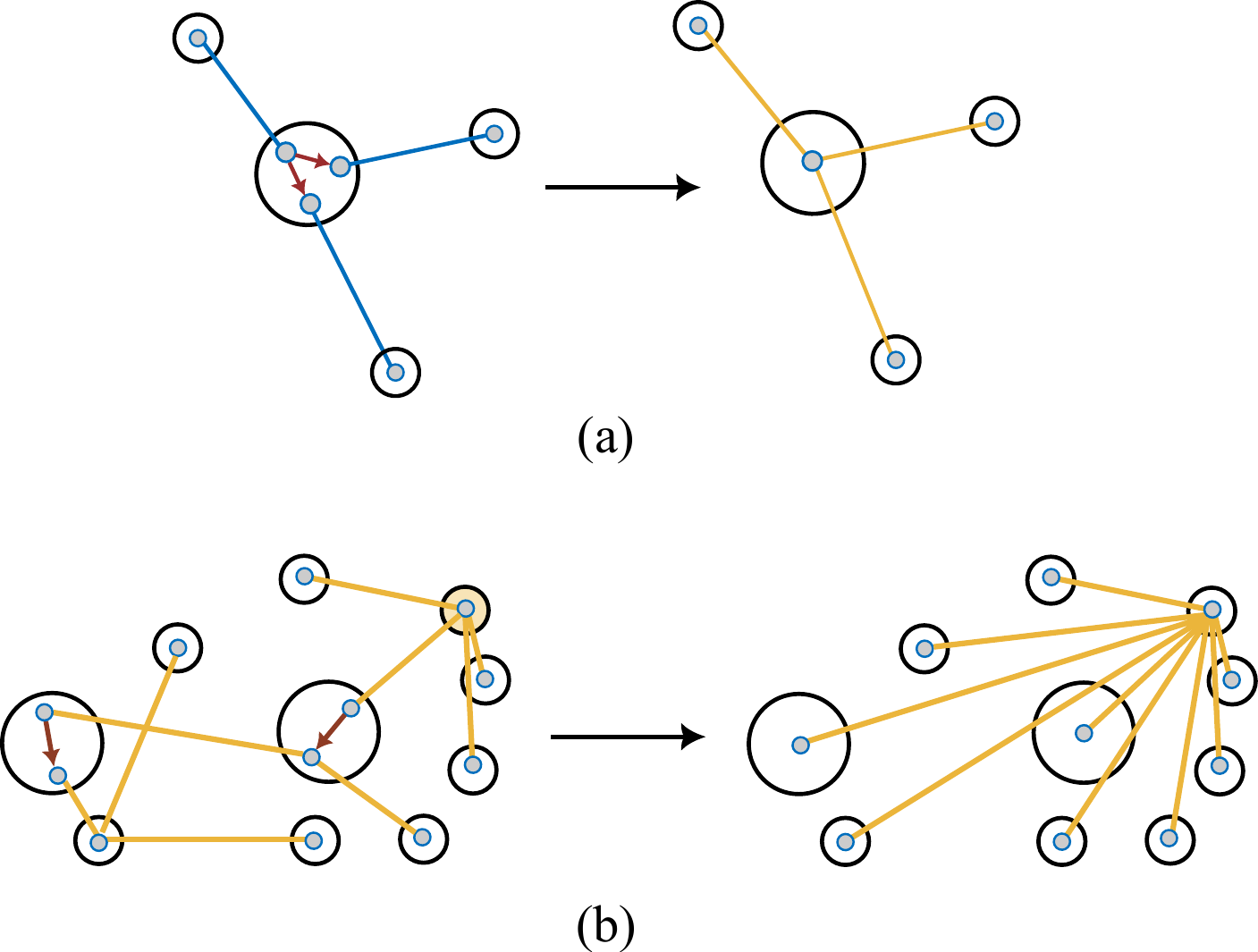}
    \caption{\textbf{Fusion protocols~\cite{kruszynska2006purificationgraphstates,PWD18}.} (a) Fusion of Bell pairs connected to a central node to create a GHZ state (Protocol~\ref{alg:Bell_star_to_GHZ}). Arrows indicate CNOT gates, pointing from control qubits to target qubits. (b) Fusion of GHZ states in a tree topology (Protocols~\ref{alg:GHZ_merge_protocol} and \ref{alg:multiple_GHZ_merge_tree_protocol_main}), with the root node highlighted. 
    In both protocols, after the CNOT gates, the target qubits are measured in the Pauli-$Z$ basis.}
    \label{fig:GHZ_Protocols}
\end{figure}

\begin{figure*}
    \centering
    \includegraphics[width=0.95\textwidth]{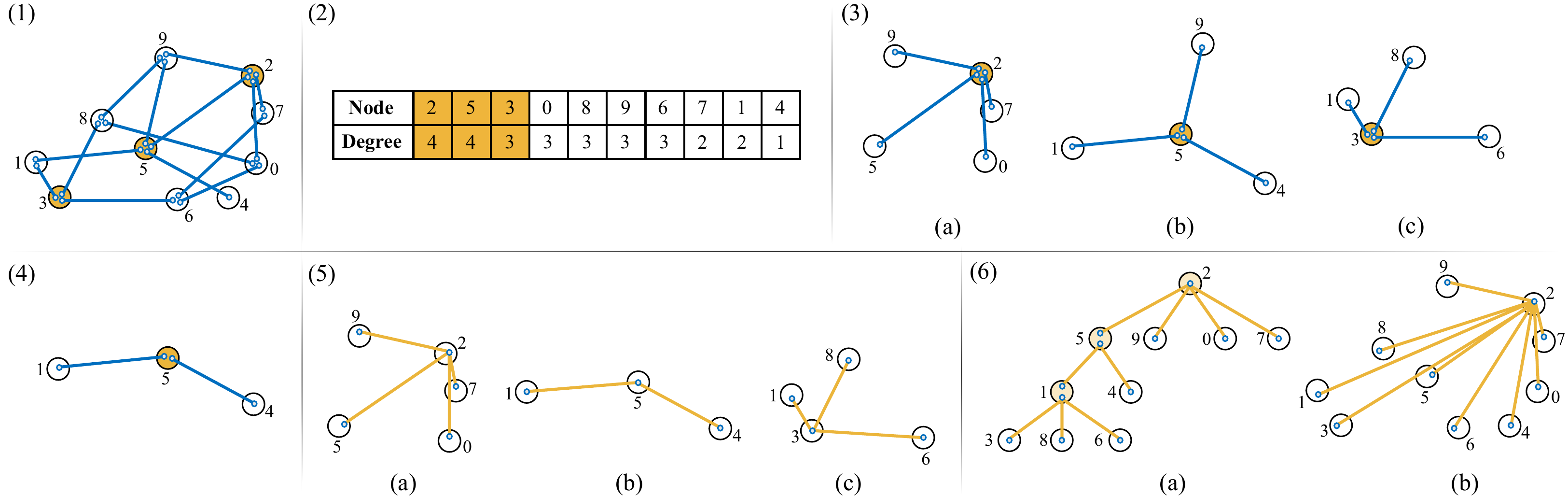}
    \caption{
    \textbf{Example execution of Protocol~\ref{alg:GHZ_network_complete} for distributing a GHZ state among all network nodes.} (1) The initial network of Bell pairs. (2) Table of nodes and their degrees in descending order, as per Step~2. (3) The selection of star subgraphs according to Steps~3–6, with (a), (b), and (c) representing the consecutive stars obtained from the first three entries of the table. At this point, all nodes are included in at least one star. (4) Duplicate nodes and edges are removed, according to Steps~7–11, where stars (3a) and (3c) are unaltered, but the edge between nodes 5 and 9 in star (3b) is deleted. (5) The locally created GHZ states as in Step~12. (6a) GHZ states from (5) arranged as a tree, with the center of (5c) shifted to node 1. (6b) The final distributed GHZ state after Step~13.}
    \label{fig:GHZ_state_dist_example}
\end{figure*}

In Protocol~\ref{alg:GHZ_network_complete}, we present our protocol for creating a GHZ state in an arbitrary network of Bell pairs, and in Fig.~\ref{fig:GHZ_state_dist_example} we provide an example of how it works. 
The protocol makes use of two fusion protocols, illustrated in Fig.~\ref{fig:GHZ_Protocols}: Protocol~\ref{alg:Bell_star_to_GHZ}, which creates local, small-scale GHZ states centered around nodes with the highest degree; and Protocol~\ref{alg:multiple_GHZ_merge_tree_protocol_main}, which merges small-scale GHZ states (using Protocol~\ref{alg:GHZ_merge_protocol}) in a tree topology to achieve the desired GHZ state among all of the desired nodes. 
We note that Protocol~\ref{alg:Bell_star_to_GHZ} is a well-known method for creating a GHZ state from Bell pairs in a star topology, and Protocol~\ref{alg:multiple_GHZ_merge_tree_protocol_main} is based on a well-known method for merging GHZ states; see, e.g., Refs.~\cite{kruszynska2006purificationgraphstates,PWD18}.

\begin{algorithm}[H]
\caption{GHZ creation in a Bell-pair network}\label{alg:GHZ_network_complete}
\begin{algorithmic}[1]
\Require Network of Bell pairs, represented by a graph $G$ with $N$ nodes. 
\Ensure GHZ state shared among all $N$ nodes.
\State Create an empty list, \texttt{SG}.
\State Create a table listing the degree of each node in $G$, sorted in descending order. If multiple nodes have the same degree, then order them randomly.
\State Select the top-most entry of the table. Save the star subgraph, consisting of the corresponding node and its neighboring nodes, as an element in \texttt{SG}.
\State Delete the top most entry of the table, and remove the corresponding node and its associated edges from $G$.
\State Repeat steps 3 and 4 until all nodes in $G$ are included in at least one star.
\State If any edges remain in $G$, add them along with their two incident nodes to \texttt{SG}.
\State Select the largest star in \texttt{SG}, label it \texttt{S0}, and remove it from \texttt{SG}. Create a new list called \texttt{MSG} and add \texttt{S0} to it.
\State Find another star in \texttt{SG} that shares at least one node in common with any star in \texttt{MSG} and call it \texttt{S1}. 
\State Check if all nodes of \texttt{S1} are already in any stars of \texttt{MSG}. If so, remove \texttt{S1} from \texttt{SG} and return to the previous step.
\State If the central node of \texttt{S1} is in common with any star in \texttt{MSG}, remove all common nodes (and the associated edges) from \texttt{S1} except its central node. If not, remove all common nodes (and the associated edges) except one randomly-chosen common node. Remove \texttt{S1} from \texttt{SG} and add the modified \texttt{S1} to \texttt{MSG}.
\State Repeat steps 8--10, naming stars as \texttt{S2}, \texttt{S3}, etc., until \texttt{SG} is empty.
\State For every star in \texttt{MSG}, apply Protocol~\ref{alg:Bell_star_to_GHZ} to create a GHZ state among the nodes in the star.
\State Using Protocol~\ref{alg:multiple_GHZ_merge_tree_protocol_main}, merge all of the small GHZ states into one.

\end{algorithmic}
\end{algorithm}

In the subset case, we can first generate a GHZ state for the entire graph using our Protocol~\ref{alg:GHZ_network_complete} and then use $X$-measurements to remove the qubits outside of the target subset, resulting in a GHZ state shared among the desired nodes. However, this method consumes significantly more gates and Bell pairs than necessary. Instead, we first apply Protocol~\ref{alg:GHZ_network_subset}, which constructs a connected subgraph that includes the desired nodes, and then using this subgraph as the input to Protocol~\ref{alg:GHZ_network_complete} we generate the GHZ state within this subgraph. Finally, we $X$-measure out any qubits in the subgraph that are not part of the target subset. (If the set of desired nodes already forms a connected subgraph, then Protocol~\ref{alg:GHZ_network_subset} is not required.) We provide an example of how Protocol~\ref{alg:GHZ_network_subset} works in Fig.~\ref{fig:BFS_travel}. Note that this initial step of constructing a connected subgraph is in contrast to the protocol in Ref.~\cite{MMG19}---which we refer to as the ``MMG protocol'' from now on---in which the subset case involves first identifying a Steiner tree (a process known to be NP-hard~\cite{garey1977SteinertreeNP}) and then executing the star expansion protocol of the same work to obtain the GHZ state among the selected nodes. 


\begin{algorithm}[H]
\caption{GHZ state from Bell pairs in a star topology~\cite{kruszynska2006purificationgraphstates,PWD18}}\label{alg:Bell_star_to_GHZ}
\begin{algorithmic}[1]
\Require $k\in\{2,3,\dotsc\}$ Bell pairs $\ket{\Phi}_{A_iB_i}$, $k+1$ nodes in a star topology. (The central node shares a Bell pair with each of the $k$ outer nodes.)
\Ensure GHZ state $\ket{\text{GHZ}_{k+1}}$ shared by all $k+1$ nodes.
\State Apply the CNOT gates $\text{CNOT}_{A_1A_2}$, $\text{CNOT}_{A_1A_3}$, \ldots ,$\text{CNOT}_{A_1A_k}$ on the qubits in the central node.
\State Measure the qubits $A_2,A_3,\dotsc,A_k$ in the $Z$-basis $\{\ket{0},\ket{1}\}$.
\State Communicate the outcome of the measurement on qubit $A_i$ to the qubit $B_i$, $i\in\{2,3,\dotsc,k\}$. For every $i\in\{2,3,\dotsc,k\}$: if the outcome communicated by $A_i$ is 0, then nothing is done on qubit $B_i$; if the outcome communicated by $A_i$ is 1, then apply the Pauli-$X$ gate to qubit $B_i$.
\end{algorithmic}
\end{algorithm}

\begin{algorithm}[H]
\caption{Fusion of two GHZ states~\cite{kruszynska2006purificationgraphstates,PWD18}}\label{alg:GHZ_merge_protocol}
\begin{algorithmic}[1]
\Require Two GHZ states, $\ket{\text{GHZ}_{n+1}}_{A_1B_{1:n}}$ and $\ket{\text{GHZ}_{m+1}}_{A_2C_{1:m}}$, $n,m\in\{1,2,\dotsc\}$, such that qubits $A_1$ and $A_2$ are located at the same node.
\Ensure GHZ state $\ket{\text{GHZ}_{n+m+1}}_{A_1B_{1:n}C_{1:m}}$.
\State Perform the gate $\text{CNOT}_{A_1A_2}$ between the qubits $A_1$ and $A_2$.
\State Measure the qubit $A_2$ in the $Z$-basis (i.e., the $\{\ket{0},\ket{1}\}$ basis).
\State Communicate the outcome $x\in\{0,1\}$ of the measurement to the nodes $C_1,C_2,\dotsc,C_m$. If the outcome is $x=0$, then $C_1,C_2,\dotsc,C_m$ do nothing; if the outcome is $x=1$, then $C_1,C_2,\dotsc,C_m$ apply the Pauli-$X$ gate to their qubits.
\end{algorithmic}
\end{algorithm}

\begin{algorithm}[H]
\caption{Fusion of GHZ states in a tree topology~\cite{kruszynska2006purificationgraphstates,PWD18}}\label{alg:multiple_GHZ_merge_tree_protocol_main}
\begin{algorithmic}[1]
\Require GHZ states in a tree configuration, with a root node identified.
\Ensure GHZ state shared by all of the nodes.
\State Apply a CNOT gate at every node except for the root node and the leaf nodes. The control qubit is in a GHZ state with qubits in the next highest level of the tree, and the target qubits is in a GHZ with qubits in the next lowest level of the tree. 
\State Measure the target qubits Pauli-$Z$ basis $\{\ket{0},\ket{1}\}$. The outcomes are communicated to all nodes.
\State For every leaf node of the tree, apply the Pauli-$X$ correction $X^{z_1\oplus z_2\oplus\dotsb\oplus z_{\ell-1}}$, where $z_1,z_2,\dotsc,z_{\ell}\in\{0,1\}$ are the measurement outcomes obtained at the nodes along the path going from the leaf node to the root node.
\end{algorithmic}
\end{algorithm}

\begin{algorithm}[H]
\caption{Connected subgraph generation}\label{alg:GHZ_network_subset}
\begin{algorithmic}[1]
\Require Network of Bell pairs, represented by a graph $G$ with $N$ nodes, and a list of desired nodes, \texttt{DN}.
\Ensure A connected subgraph of $G$ containing all of the nodes in~\texttt{DN}.
\State Generate the subgraph induced by \texttt{DN}, call it $G'$.
This subgraph may not be connected; if it is connected, output $G'$.
\State Randomly choose a node from \texttt{DN} to serve as the root node. Using a breadth-first search (BFS) algorithm~\cite{bfs}, create a tree-like structure starting from the root node by visiting its first-level (i.e., nearest) neighbors, then the second-nearest neighbors, and so on, without repeating any nodes (see Fig.~\ref{fig:BFS_travel}).  
\State Iteratively delete leaf nodes (red nodes in Fig.~\ref{fig:BFS_travel}) from the BFS tree unless they are part of the set of desired nodes (\texttt{DN}). Continue this process until all remaining leaf nodes belong to \texttt{DN}, and no further removals are possible.
\State Construct a subgraph using nodes from BFS tree (yellow and blue nodes from Fig.~\ref{fig:BFS_travel}) and corresponding edges from $G$.
\end{algorithmic}
\end{algorithm}

\begin{figure}
   \centering
   \includegraphics[width=0.85\columnwidth]{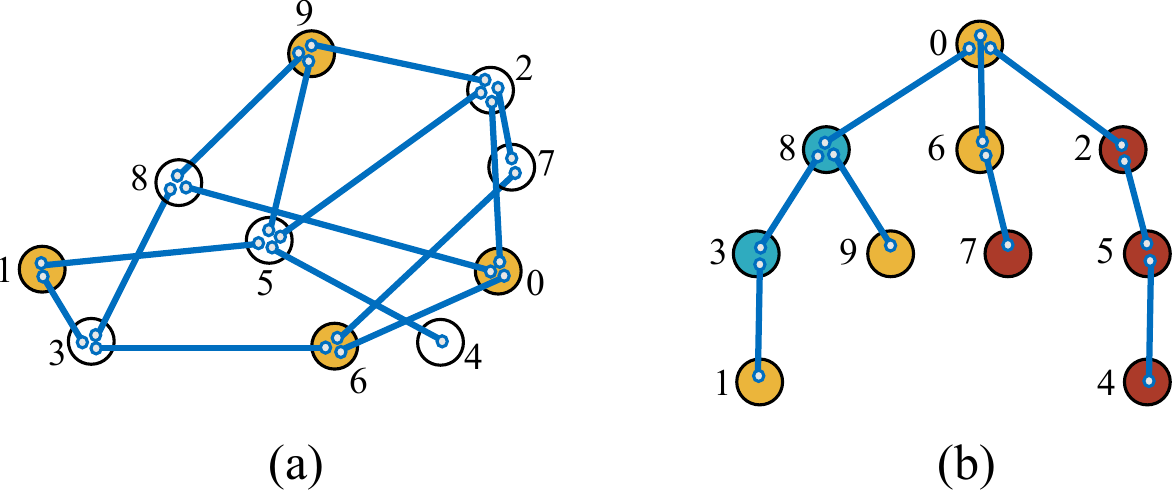}
   \caption{\textbf{Using Protocol~\ref{alg:GHZ_network_subset} to determine a connected subgraph containing a desired subset of nodes.} (a) The original network, with the desired nodes highlighted in yellow, forming the set $\texttt{DN}=\{0,1,6,9\}$. (b) The nodes highlighted in red are not necessary and are ignored, while the nodes highlighted in blue are necessary to have a connected subgraph. The resulting subgraph $G$ consists of the desired yellow nodes and the blue nodes.}\label{fig:BFS_travel}
\end{figure}

\subsubsection*{Remarks}

\begin{enumerate}[leftmargin=2ex]
    \item The graph $G$ in Protocol~\ref{alg:GHZ_network_complete} can be arbitrary, as long as it is connected, i.e., if there exists a path connecting every pairs of vertices in $G$. Note that because we are in the LOCC setting, an LOCC protocol can create a GHZ among all nodes of $G$ only if $G$ is connected.
    
    \item  In our protocols, the locations of the star creations are not unique, in part because the highest-degree nodes (selected in Step~2 of the protocol) need not be unique and could be chosen randomly. 

    \item Protocol~\ref{alg:GHZ_network_complete} requires a number of Bell pairs equal to the total sum of Bell pairs across all the stars in the list \texttt{MSG} from Step~12. The fusion protocols do not require any additional Bell pairs. It is also sufficient to use a number of Bell pair sources that is equal to the number of elements in the list \texttt{MSG}. 

    \item For a given network, Protocols~\ref{alg:GHZ_network_complete} and \ref{alg:GHZ_network_subset} need not be executed in real time; they can be performed offline, as they depend solely on the topology of the initial Bell-pair network. Indeed, up to Step~11, Protocol~\ref{alg:GHZ_network_complete} is purely classical and is intended to determine the list \texttt{MSG} that determines what fusions the nodes must do. This list can be determined in advance of the actual implementation in a network and distributed to the nodes. When it is time for the actual implementation, in real time, the nodes simply take the ``instructions'' represented by \texttt{MSG} and execute Protocols~\ref{alg:Bell_star_to_GHZ} and \ref{alg:multiple_GHZ_merge_tree_protocol_main}, coordinating with each other as needed. Similarly, Protocol~\ref{alg:GHZ_network_complete} is purely classical and can be executed in advance to determine the subgraph on which Protocol~\ref{alg:GHZ_network_complete} would be performed.
    
    \item Protocol~\ref{alg:GHZ_network_complete} will produce a tree. If we are in the complete case, then this is a minimum spanning tree for $G$. In the subset case, it will be a minimum spanning tree corresponding to the subgraph containing the desired nodes, as generated by Protocol~\ref{alg:GHZ_network_subset}. However, we emphasize that while our algorithm produces a tree, much like prior algorithms, ours is not just any tree --- it features a specific structure aimed at minimizing the number of Bell pair sources required for creation of the GHZ state. 
    
    \item If the given network has one node with a degree equal to \( N-1 \)—that is, a node connected to all other nodes—then the protocol finishes in one step.

\end{enumerate}

\subsubsection*{Classical communication cost}

Protocol~\ref{alg:Bell_star_to_GHZ} requires $k-1$ bits of classical communication, where $k$ is the number of nodes in the star graph. In Protocol~\ref{alg:multiple_GHZ_merge_tree_protocol_main}, the number of bits is determined as follows. We first note that all correction operations in Protocol~\ref{alg:multiple_GHZ_merge_tree_protocol_main} are done by the leaf nodes only --- the intermediate internal nodes (i.e., non-leaf and non-root nodes) do not have to apply correction operations. Every internal node performs a single fusion operation, consisting of the CNOT gate followed by a Pauli-$Z$ measurement, and the measurement outcome is communicated downward in the tree to its corresponding leaf nodes. Every leaf node, $v_{\ell}$, therefore receives $L_{v_{\ell}}$ bits, where $L_{v_{\ell}}$ is the number of internal nodes along the path from $v_{\ell}$ to the root node. This means that the total classical communication cost of Protocol~\ref{alg:multiple_GHZ_merge_tree_protocol_main} is $\sum_{\ell}L_{v_{\ell}}$, where we sum over all leaf nodes. In total, the classical communication cost of Protocol~\ref{alg:GHZ_network_complete} is the total number of bits needed to create the local GHZ states from the set $\texttt{MSG}$ using Protocol~\ref{alg:Bell_star_to_GHZ} combined with the total number of bits needed in Protocol~\ref{alg:multiple_GHZ_merge_tree_protocol_main} to fuse the local GHZ states into the overall GHZ state shared by the desired nodes.

\subsubsection*{Creation of arbitrary graph states}\label{sec:arbitrary_graph_states}

While our primary focus is on creating GHZ states, it is possible to use our protocol to create other graph states. A first strategy is one based on Ref.~\cite{MMG19}. Specifically, from a GHZ state shared by all of the nodes, we can create the complete graph (all-to-all connected graph) via local complementation. Then, if a desired graph state can be created from the complete graph by deleting vertices and their corresponding edges, implemented via $Z$-basis measurements, or via other local operations on the complete graph state, then such states can be generated by first using Protocol~\ref{alg:GHZ_network_complete} to create a GHZ state in the network and then performing the appropriate local operations. This strategy can be enhanced with additional qubits, as described in Ref.~\cite{MMG19}, in order to create the so-called edge-decorated complete graph using GHZ states in order to generate an arbitrary graph state. Alternatively, instead of proceeding via the GHZ state, we can modify Protocol~\ref{alg:GHZ_network_complete} by replacing Protocol~\ref{alg:Bell_star_to_GHZ} therein with the Bipartite~B protocol of Ref.~\cite{CC12}, and then perform the appropriate fusion operations similar to those in Protocol~\ref{alg:multiple_GHZ_merge_tree_protocol_main}. In this case, determining the appropriate fusion operations in order to achieve the desired graph state is non-trivial, depending in general on both the desired graph state and the network topology, and thus we leave this as an interesting direction for future work.



\subsection*{Time complexity}\label{sec-complexity-proofs}

We now present results on the time complexity of our protocols in the worst-case scenario, focusing solely on the classical steps and neglecting the complexity associated with implementing quantum operations. We present all proofs in Methods.

\begin{theorem}[Time complexity: complete case]\label{complexity_complete_proof}
    The time complexity of Protocol~\ref{alg:GHZ_network_complete} is $O(N^2)$, where $N$ is the total number of nodes in the graph.
\end{theorem}

\begin{theorem}[Time complexity: subset case]\label{complexity_subset_proof}
    The time complexity of Protocol \ref{alg:GHZ_network_subset} is $O(N+E)=O(N^2)$, where $N$ is the total number of nodes in the graph.
\end{theorem}

To create a GHZ state on a subset of nodes, Protocol~\ref{alg:GHZ_network_subset} should be executed first, followed by Protocol~\ref{alg:GHZ_network_complete}. Since each protocol is applied exactly once, regardless of the size of initial network or subset, the overall complexity is $O(N^2)$. It is worth pointing out that while polynomial-time approximation algorithms for Steiner trees exist~(see, e.g., Refs.~\cite{kou1981steiner,mehlhorn1988steiner,robins2005steinertreeapproximation}), the algorithms in Refs.~\cite{kou1981steiner,robins2005steinertreeapproximation} have time complexity $O(S'\cdot N^2)$, where $S'$ is the number of desired nodes, and the widely-implemented algorithm of Ref.~\cite{mehlhorn1988steiner} has time complexity $O(E+N\log N)$. This then implies that the time complexity of the MMG protocol~\cite{MMG19}, which requires a Steiner tree algorithm as a subroutine, is $O(E+N\log N)$. Our Protocol~\ref{alg:GHZ_network_subset} has a faster time complexity of $O(E+N)$ compared to Refs.~\cite{kou1981steiner,mehlhorn1988steiner,robins2005steinertreeapproximation}. At the same time, let us remark that if we are only in the complete case, then we need Protocol~\ref{alg:GHZ_network_complete} only, which produces a minimum spanning tree, and its time complexity of $O(N^2)$ is comparable to known algorithms for minimum spanning trees; see, e.g., Ref.~\cite{fredman1994spanningtrees}. However, while the minimum spanning tree would be enough if we were concerned only with the number of Bell pairs, 
the advantage of Protocol~\ref{alg:GHZ_network_complete} becomes clear when considering the number of Bell-pair sources as a figure of merit in addition to the number of Bell pairs.

Let us also consider the time complexity of the protocol of Ref.~\cite{fischer2021graph}. The algorithm of Ref.~\cite{fischer2021graph} uses a network flow algorithm that has time complexity $O(N\cdot E^2)$~\cite[Chapter~26]{bfs}.

\subsection*{Gate cost}\label{sec-analytical_results_gates}

To determine the gate cost of our protocol, we start by counting the number of gates needed in Protocols~\ref{alg:Bell_star_to_GHZ} and \ref{alg:GHZ_merge_protocol}.
\begin{enumerate}
    \item Protocol~\ref{alg:Bell_star_to_GHZ} requires $n-1=N-2$ gates, where $n$ is the number of Bell pairs input to the protocol and $N=n+1$ is the number of nodes involved in the protocol. 
    \item Protocol~\ref{alg:GHZ_merge_protocol} requires one gate, independent of the sizes of the GHZ states being merged.
\end{enumerate}
Note that we do not consider the Pauli correction gates in the gate count because in near-term hardware, experiments show that two-qubit gates are typically one or two orders of magnitude noisier than single-qubit gates~\cite{Moses2023}. The proofs of the following results are presented in Methods.

\begin{theorem}[Number of gates: complete case]\label{thm:gates_complete_case}
   The gate cost of Protocol~\ref{alg:GHZ_network_complete} in the complete case is $N-2$ gates, where $N$ is the total number of nodes in the network.
\end{theorem}

\begin{theorem}[Number of gates: subset case]\label{thm:gate_cost_subset}
   In the subset case, the gate cost of our protocol is $S-2$ gates where $S$ is the total number of nodes in the subgraph (the subgraph obtained using Protocol~\ref{alg:GHZ_network_subset}).
\end{theorem}

The proof of Theorem~\ref{thm:gate_cost_subset} is analogous to the proof of Theorem~\ref{thm:gates_complete_case}.

It is notable here that the number of gates depends \textit{only} on the number of nodes in the subgraph---no specific properties of the structure/topology/connectivity of the graph matters, except for the fact that it should be connected. In particular, this means that if we take any fraction $f\in(0,1]$ of the nodes $N$ of our network to comprise our subset, then the size $S$ of the subgraph containing these nodes satisfies $S\geq Nf$, which means that the number of gates is bounded from below as $S-2\geq Nf-2$.

\paragraph*{Comparison to Ref.~\texorpdfstring{\cite{MMG19}}{MMG}.}

The gate cost of the star expansion protocol of Ref.~\cite{MMG19} is as follows. If $k$ is the degree of the node on which the star expansion protocol is performed, then:
\begin{itemize}
    \item If the central node is to be kept, then the number of gates is $g_1(k)\coloneqq\binom{k}{2}+2k-1=\frac{1}{2}k^2+\frac{3}{2}k-1$.
    \item If the central node is to be discarded, then the number of gates is $g_2(k)\coloneqq\binom{k}{2}+k=\frac{1}{2}k^2+\frac{1}{2}k$.
\end{itemize}
Using this, we can establish the following fact about the number of gates required to create a GHZ state according to the protocol in Ref.~\cite{MMG19}.

\begin{proposition}[Upper bound on gate cost of the MMG protocol~{\cite{MMG19}}]\label{prop:gate_cost_MMG}
    For any graph with $N$ nodes, the number of gates required to create a GHZ state among all $N$ nodes, according to the MMG protocol, scales as $O(N^2)$.
\end{proposition}

Let us now also observe that we can significantly reduce the gate cost of the star expansion Protocol by noticing that the end result of the protocol is a GHZ state, either with or without the central node. If the central node is to be excluded, then by the Bipartite~B protocol of Ref.~\cite{CC12}, the gate cost is $k-1$. If the central node is to be included, then the gate cost is also $k-1$, using Protocol~\ref{alg:Bell_star_to_GHZ}. This represents a quadratic reduction in gate cost, from $O(k^2)$ (as established in Ref.~\cite{MMG19}) to $O(k)$.

\paragraph*{Comparison to Ref.~\texorpdfstring{\cite{fischer2021graph}}{Fischer}.} The lower bound on gate count in Ref.~\cite{fischer2021graph} is reached when a single node is connected to all other nodes in the graph, resulting in \( N - 2 \) gates—consistent with our protocol's gate count. However, as we prove below, the upper bound in Ref.~\cite{fischer2021graph} scales as \( O(N^2) \). Therefore, our protocol generally achieves a lower gate count, except in the specific scenario where one node connects to all others, in which case our protocol matches the gate count in Ref.~\cite{fischer2021graph}.

\begin{proposition}[Upper bound on gate cost in Ref.~{\cite{fischer2021graph}}]\label{prop:gate_cost_fisher}
    For a graph with \( N \) nodes, the upper bound on the number of gates required to create a GHZ state across all \( N \) nodes, as specified by the protocol in Ref.~\cite{fischer2021graph}, scales as \( O(N^2) \).
\end{proposition}

\subsection*{Bell-pair cost}\label{sec-analytical_results_Bell_pairs}

\begin{theorem}[Number of Bell pairs used: complete case]\label{full_edges_proof}
     For the complete case of our protocol, the number of Bell pairs used is equal to $N-1$. This matches the lower bound of $N-1$, which is the minimum number of Bell pairs needed to create an arbitrary graph state.
\end{theorem}

Let us now examine the number of Bell pairs in the subset case by introducing the concept of the graph diameter. The diameter of an unweighted graph is defined as the length of the longest shortest path between any two nodes. In simpler terms, it is the greatest distance, measured by the number of edges, between any pair of nodes. Here, the shortest path refers to the path connecting two nodes with the fewest edges.

\begin{theorem}[Number of Bell pairs used: subset case]\label{subset_edges_proof}
    Let $S$ be the size of the subgraph containing the desired subset of nodes, as generated from Protocol~\ref{alg:GHZ_network_subset}. The number of Bell pairs needed to generate a GHZ state is bounded from below by $S-1$ and bounded from above by $d_S(S-1)$, where $d_S$ is the diameter of the connected subgraph containing the desired subset of nodes, as generated from Protocol~\ref{alg:GHZ_network_subset}.
\end{theorem}

\paragraph*{Comparison to Refs.~\texorpdfstring{\cite{MMG19,fischer2021graph}}{MMG, Fischer}.} In the complete case, for the MMG protocol~\cite{MMG19}, the number of Bell pairs required is $N - 1$, as it corresponds to the size of the minimum spanning tree. This is the same number of Bell pairs used in our protocol (Theorem~\ref{full_edges_proof}). However, the approach in Ref.~\cite{fischer2021graph} requires additional Bell pairs to connect the desired node to every other node, unlike our protocol and the MMG protocol. In the subset case, the number of Bell pairs needed in the MMG protocol is equal to the size of the tree, which will generally exceed $N_s - 1$, where $N_s$ is the number of nodes in the subset. This is because the Steiner tree will typically include more nodes than just the terminal nodes. In Ref.~\cite{fischer2021graph}, the same logic applies as in the complete case.

\subsection*{Bell-pair source cost}\label{sec-Bell-pair-sources}

It is commonly assumed that a Bell pair source must be placed at the midpoint of each utilized edge. Ideally, this allows all Bell pairs to be distributed simultaneously in a single time step. However, sources such as quantum dots can generate Bell pairs at very high rates (GHz range)~\cite{vanLoock2020}. This opens up an alternative approach to designing the placement of Bell pair sources. Instead of placing a source on every edge, one can place sources at selected nodes, and then have each source generate and distribute Bell pairs to its neighbors. While this introduces a small overhead due to sequential pair generation, it does not significantly affect timing thanks to the high repetition rates. Furthermore, and importantly, the number of edges in a network scales quadratically with the number of nodes, while the number of nodes grows linearly. Therefore, placing sources at the nodes rather than at edges can offer a significant advantage in terms of reducing the total number of sources required. With careful placement, a small number of node-based sources can efficiently cover the entire network. Now the key question becomes: given a network topology, what is the minimal number of Bell pair sources required? As we now show, this optimization problem is equivalent to the well-known minimum dominating set problem.

\begin{theorem}[Bell-pair source cost]\label{thm:bell_pair_dominating_set}
The problem of determining the minimum number of Bell pair sources required to cover all nodes in a quantum network is equivalent to the minimum dominating set problem on the network graph.
\end{theorem}

Since the minimum dominating set problem is known to be NP-hard \cite{dom_set_np}, the Bell pair source placement problem inherits this computational complexity. However, a simple structural observation provides a useful upper bound on the number of sources required. In any connected graph, each leaf node (a node of degree one) is connected to exactly one neighbor, which must be an internal node (i.e., node with degree greater than one). We begin by placing a source at every internal node.  By placing sources on all internal nodes, we ensure that every leaf node is adjacent to a source. Since a source can distribute Bell pairs to all its neighbors, each leaf node can receive a Bell pair from its connected internal node. Therefore, placing a source at a leaf node is unnecessary. Thus, placing sources only at all internal nodes suffices to cover the graph, giving a simple upper bound on the minimum number of Bell pair sources required. We emphasize that this is generally a better upper bound compared to placing the sources at the midpoints of the edges, based on the reasoning presented at the beginning of this subsection.

\subsection*{Bell pairs vs. GHZ states as building blocks}

Throughout the protocols presented so far, Bell pairs serve as the fundamental building blocks, and we assume the availability of sources capable of producing them. Importantly, Bell-pair sources are not required at every node; instead, their placement is chosen strategically, as specified in Protocol~\ref{alg:GHZ_network_complete} (Step 12, center of each star in MSG). A natural alternative is to consider sources that directly produce GHZ states. One key distinction is that Bell pairs are universal building blocks for entangled graph states: in principle, any graph state can be generated from Bell pairs~\cite{CC12}. By contrast, when GHZ states are used as the primitive resource, it is not clear whether other classes of graph states can be generated. Nevertheless, if the sole objective is the creation of a GHZ state, then using GHZ states as building blocks introduces nontrivial trade-offs in terms of resources and performance.

The main difference between distributing Bell pairs and distributing a GHZ state lies in the role of quantum memories. When Bell pairs are used, they do not need to be established simultaneously: thanks to quantum memories, each Bell pair can be generated independently and stored until all required links are available. Assuming, for simplicity, equal segment lengths, if $\eta$ denotes the transmission efficiency of the optical fiber for a single segment, the effective success probability for all Bell pair remains $\eta$. In contrast, if a GHZ state is produced directly as a building block, all qubits must be transmitted successfully in the same attempt. In this case, the overall transmission efficiency scales as $\eta^d$, where $d$ is the degree of the node. Consequently, from a rate perspective, Bell pairs offer a clear advantage. However, with respect to fidelity, some Bell pairs may be stored in memory for longer periods than others, potentially leading to nonuniform noise depending on memory quality and distribution rates. In contrast, when a GHZ state is distributed directly, all nodes receive their qubits simultaneously, which can lead to a higher overall state fidelity. Therefore, we perform a Monte Carlo simulation to compare the distribution of Bell pairs and GHZ states. We consider finite memory lifetimes, measured in discrete time steps, where one time step corresponds to the duration of a single entanglement-generation attempt. Bell pairs are assigned a fixed pre-determined lifetime, after which they are discarded and generated again. A lifetime of one time step becomes equivalent to the GHZ-state distribution case, as all Bell-pair generation attempts must succeed simultaneously. The results are shown in Fig.~\ref{fig:Bellpair-vs-GHZ}.

An additional degree of freedom comes from the choice of the GHZ state size used as a building block. Rather than generating a single GHZ state whose size matches the full degree of the node, one may instead distribute multiple smaller GHZ states and combine them locally. For instance, in the example above, two three-qubit GHZ states could be generated and then merged at the central node. This introduces a trade-off between transmission success probability and the number of local operations required. As a result, the relative performance of Bell pairs versus GHZ states depends on how this trade-off is optimized, and a fair comparison therefore requires a more detailed, quantitative analysis, which we leave for future work.

When GHZ states are used as the building blocks, two-qubit gates are required only to merge smaller GHZ states. Consequently, the total number of two-qubit gates needed is equal to the number of Bell-pair sources minus one, which is also equal to the size of the minimum dominating set, as shown in Theorem~\ref{thm:bell_pair_dominating_set}. By comparison, when Bell pairs are used as building blocks, the number of required two-qubit gates is $n-2$, where $n$ is the number of qubits in the final GHZ state, as proven in Theorem~\ref{thm:gates_complete_case}. In all cases $n-2$ is higher than the size of the dominating set. 

\begin{figure}
    \centering
    \includegraphics[width=0.90\columnwidth]{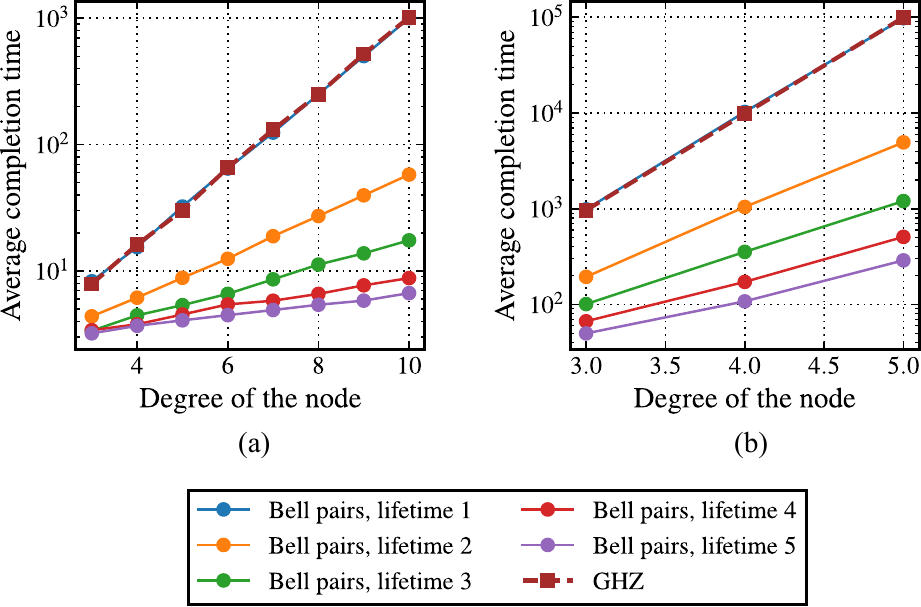}
    \caption{\textbf{Comparing Bell pairs and GHZ states as building blocks in a star network.} The Bell pair lifetime is defined as the number of time steps after which the pair is discarded and freshly generated again. (a) Transmission success probability of 0.5, corresponding to a segment length of 15~km given $L_{\text{att}} = 22$~km according to $p=\exp(-L/L_{\text{att}})$. (b) Transmission success probability of 0.1, corresponding to a segment length of 50~km.}
    \label{fig:Bellpair-vs-GHZ}
\end{figure}


\subsection*{Performance on random networks}\label{sec:numerical_results}

To evaluate our protocol, i.e., calculate the gate and Bell pairs costs, we implemented our algorithms and performed Monte Carlo simulations on various types of random networks~\cite{Bollobas2001_book,newman2001randomgraphs}. We made use of the Python package NetworkX~\cite{networkx}, and for finding Steiner trees, we applied Mehlhorn's approximate algorithm~\cite{mehlhorn1988steiner}. Details on the random network models are provided in Methods.


For our Monte Carlo simulations, we considered network sizes ranging from 100 to 500 nodes for each model, with 500 samples used for each network size. 

In Fig.~\ref{fig:BA_ER_gates_vs_nodes}, we plot the gate cost as a function of the size $S$ of the subgraph. We can see that the gate cost of our protocol agrees with the theoretical bounds discussed above, specifically Theorem~\ref{thm:gate_cost_subset}. Similarly, the gate cost of the MMG protocol~\cite{MMG19} agrees with the upper bound in Proposition~\ref{prop:gate_cost_MMG}. 

\begin{figure}
    \centering
    \includegraphics[width=0.70\columnwidth]{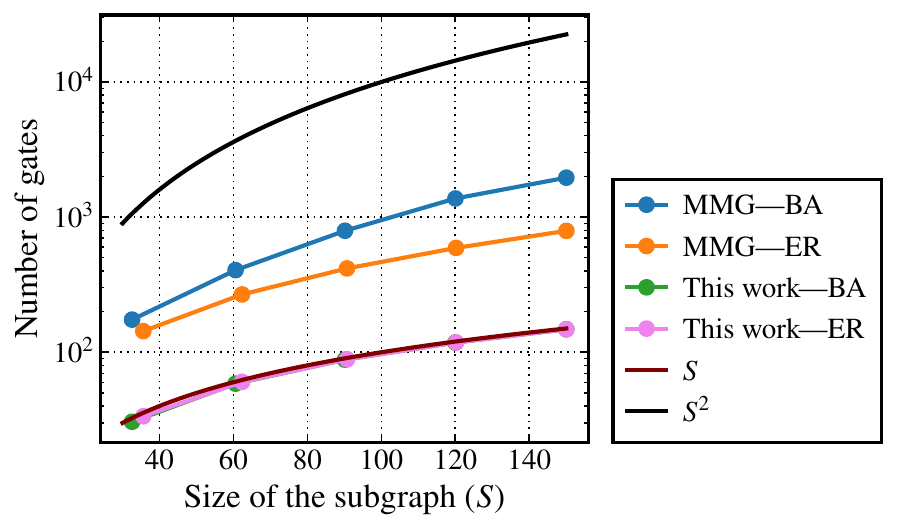}
    \caption{\textbf{Gate-cost comparisons for random networks, part~I.} Number of gates versus subgraph size for Erd\H{o}s--R\'{e}nyi and Barab\'{a}si--Albert networks with up to 500 nodes. We consider a fraction $f=0.3$ of nodes to be entangled, and we set $p=0.05$.} 
    \label{fig:BA_ER_gates_vs_nodes}
\end{figure}

In addition to understanding how the gate cost scales with the number of nodes (as presented in Fig.~\ref{fig:BA_ER_gates_vs_nodes}), we are interested in understanding how the gate cost depends on graphical properties of the networks. In Fig.~\ref{fig:BA_ER_gates_vs_p}, we plot the gate cost as a function of the parameter $p$ characterizing the ER and BA network models. In both models, the parameter $p$ gives an indication of the density of the network, with larger values of $p$ indicating higher density. Intuitively, as the density increases, we might expect the gate cost to decrease; however, as we can see in Fig.~\ref{fig:BA_ER_gates_vs_p}~(left), for the MMG protocol the gate cost actually \textit{increases} with $p$. In contrast, our protocol maintains a nearly constant number of gates because the size of the subgraph does not change beyond a certain value of $p$, as we can see in Fig.~\ref{fig:BA_ER_gates_vs_p}~(right). In particular, beyond that value of $p$, the connected subgraph obtained from Protocol~\ref{alg:GHZ_network_subset} contains only those nodes that are to be entangled, with no additional nodes. 

\begin{figure}
    \centering
    \includegraphics[width=0.95\columnwidth]{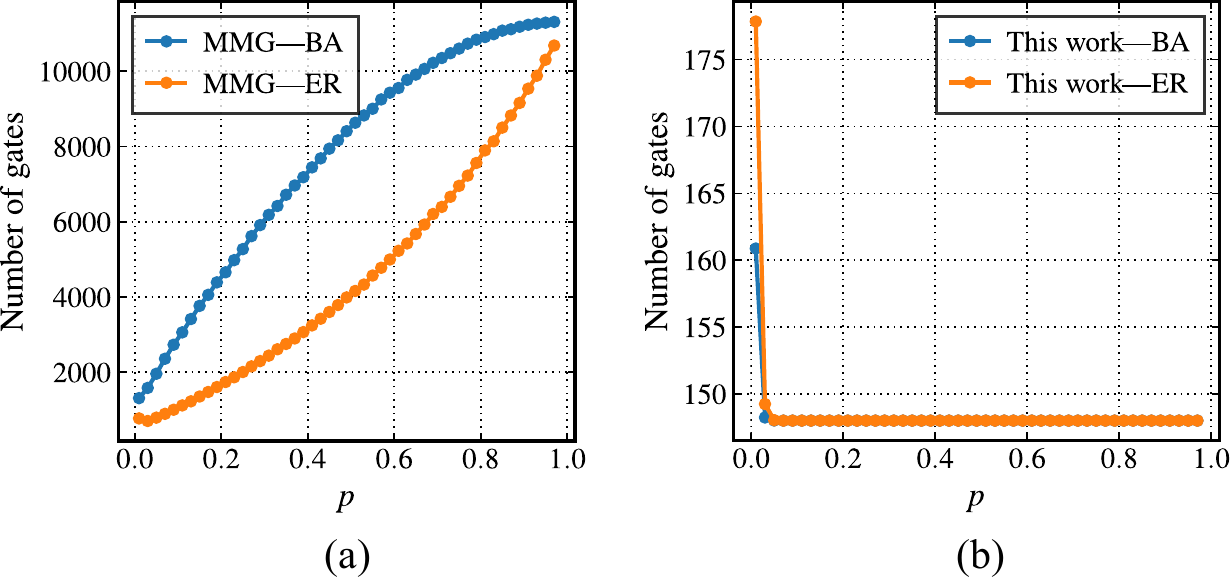}
    \caption{\textbf{Gate-cost comparisons for random networks, part~II.} Gate cost versus $p$ for Erd\H{o}s--R\'{e}nyi and  Barab\'{a}si--Albert networks with 500 nodes and a fraction $f=0.3$ of nodes to be entangled. We compare the MMG protocol (a) with this work (b).}
    \label{fig:BA_ER_gates_vs_p}
\end{figure}

In Fig.~\ref{fig:photonic_1}, we consider the Waxman model, which incorporates the distance between nodes to better reflect realistic networks. We set $\beta$ to unity, and we set the attenuation length between any two nodes to be $L_0 = 22~\text{km}$. Because the Waxman model does not always produce a connected graph, we focus on the largest connected component when analyzing this model. In this model, the largest connected component undergoes a phase transition with respect to diameter; see Fig.~\ref{fig:photonic_1}(a). This phase transition is also reflected in the change of density of the largest connected component, as shown in Fig.~\ref{fig:photonic_1}(b). In Fig.~\ref{fig:photonic_1}(c), for networks with 500 nodes, we can also see the effect of the phase transition on the gate cost of our protocol. Finally, in Fig.~\ref{fig:photonic_1}(d), we show how the gate cost varies with the size of the largest connected component. As with Fig.~\ref{fig:BA_ER_gates_vs_nodes}, the results are in agreement with the analytical results presented above. We observe multiple values for the gate cost when $S=500$, which reflects the dependence of the gate cost of the MMG protocol~\cite{MMG19} on the density of the graph.

\begin{figure}
    \centering
    \includegraphics[width=0.95\columnwidth]{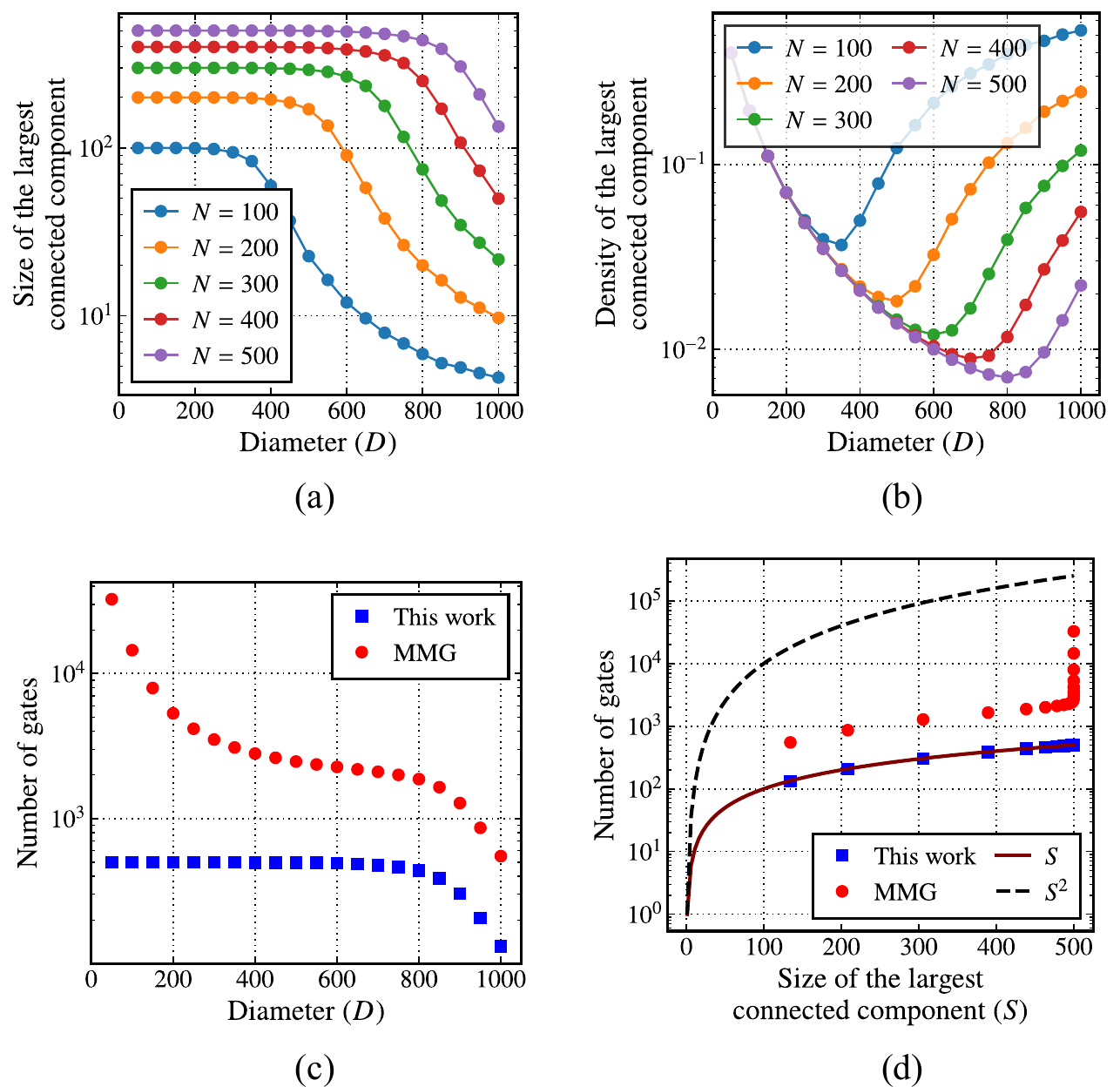}
    \caption{\textbf{Results for the Waxman model with $\beta =1$ and $L_0 = 22~\text{km}$.} (a) and (b): size and density of the largest connected component as a function of diameter. (c): Gate cost as a function of diameter for networks with 500 nodes. (d): Gate cost as a function of the size of the largest connected component.}
    \label{fig:photonic_1}
\end{figure}

Let us now consider the number of consumed Bell pairs. In particular, it is of interest to know, particularly in the subset case, how close to optimal is the number of Bell pairs consumed using Protocol~\ref{alg:GHZ_network_subset} compared to using the Steiner tree. In Fig.~\ref{fig:steiner_vs_ours}, we plot the ratio of the size of the subgraph to the size of the subset (i.e., the number of nodes to be entangled). This ratio tells us how many additional nodes are needed in the subgraph generated by our Protocol~\ref{alg:GHZ_network_subset} compared to the Steiner tree subgraph (generated using the protocol of Ref.~\cite{mehlhorn1988steiner}). We can see that for both the BA and ER models, the ratio of our subgraph is close to that of the Steiner tree, implying the number of consumed Bell pairs is also close. In particular, the ratios appear to be getting closer as the total number $N$ of nodes increases, indicating that our protocol may be optimal asymptotically.

\begin{figure}
    \centering
    \includegraphics[width=0.95\columnwidth]{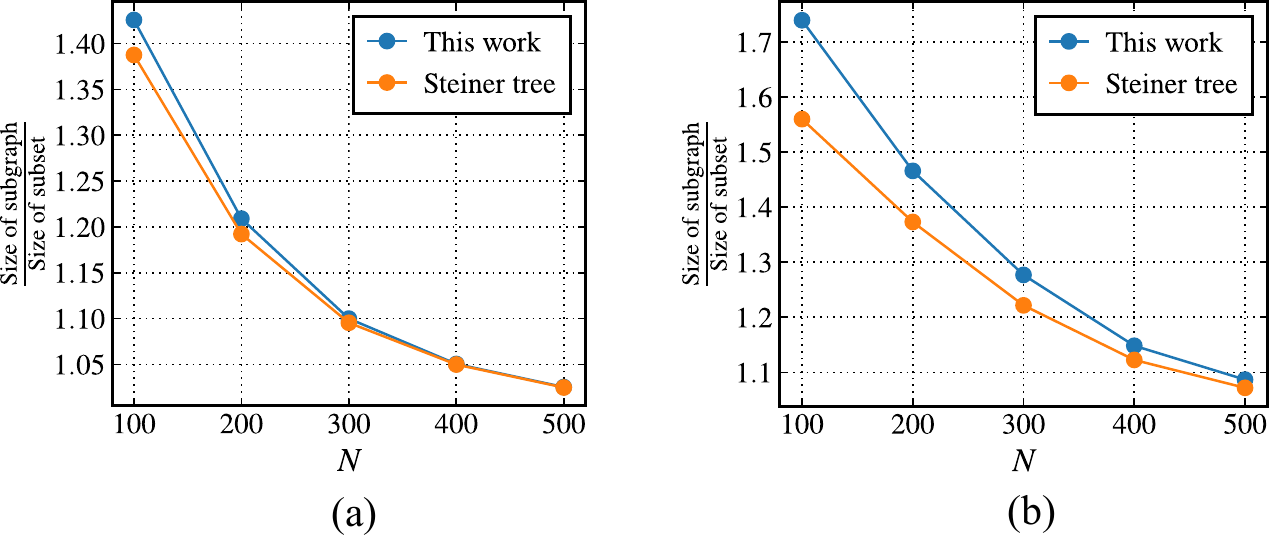}
    \caption{\textbf{Comparison of subgraph generation protocols.} Ratio of the subgraph size to the number of nodes to be entangled, as a function of the total number of nodes, for both the Steiner tree algorithm~\cite{mehlhorn1988steiner} and our Protocol~\ref{alg:GHZ_network_subset}. We consider the Barab\'{a}si--Albert (a) and Erd\H{o}s--R\'{e}nyi (b) models with $p=0.05$ and a fraction $f=0.1$ of nodes to be entangled.}
    \label{fig:steiner_vs_ours}
\end{figure}

\begin{figure}
    \centering
    \includegraphics[width=0.60\columnwidth]{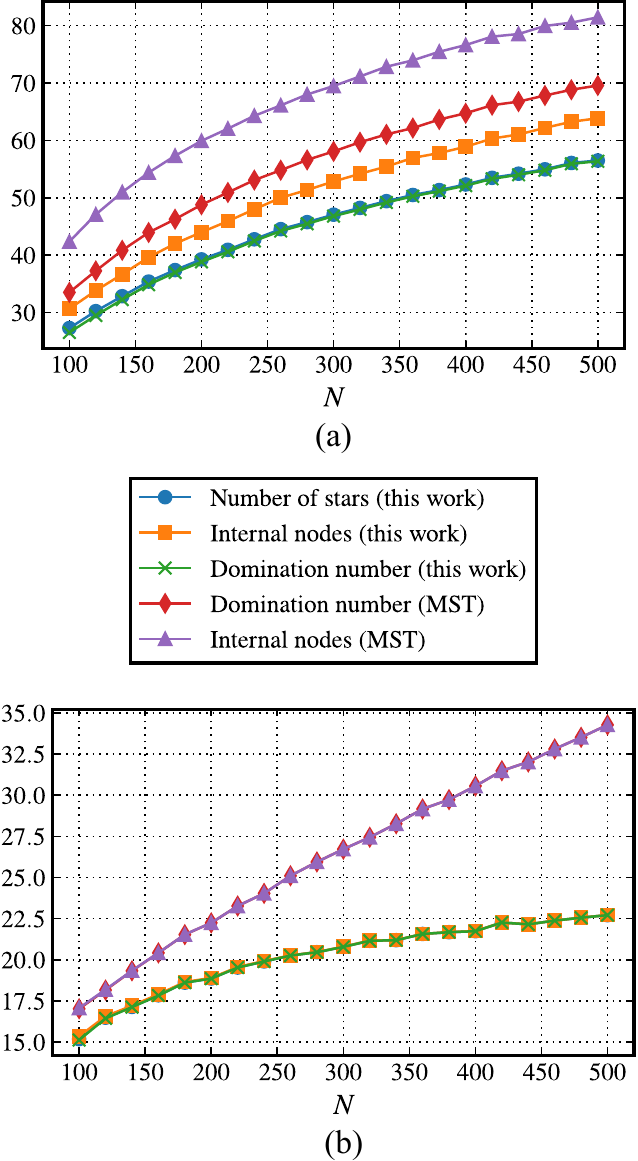}
    \caption{\textbf{Bell-pair source cost comparison for random networks.} The number of Bell pair sources needed as a function of the number of nodes ($N$) for the Erd\H{o}s--R\'{e}nyi model (a) and Barab\'{a}si--Albert model (b) with $p=0.05$. We consider the complete case (GHZ among all the nodes) using Protocol~\ref{alg:GHZ_network_complete} of this work, and we compare it to counting the internal nodes and the (approximate) domination number of the tree produced by both Protocol~\ref{alg:GHZ_network_complete} and the minimum spanning tree (MST).}
    \label{fig:Bell_pair_sources_ER_BA}
\end{figure}

Finally, we assess the number of Bell‐pair sources required by our protocol; see Fig.~\ref{fig:Bell_pair_sources_ER_BA}. 
We know that the size of a minimum dominating set (i.e., the domination number) provides a lower bound on the number of Bell‐pair sources, while the number of internal nodes provides an upper bound. In Fig.~\ref{fig:Bell_pair_sources_ER_BA}, we compare these bounds for the trees generated by our Protocol~\ref{alg:GHZ_network_complete} and those produced by a Minimum Spanning Tree (MST) of the same network. We also include the number of stars, that is, the number of elements in \texttt{MSG}, from Protocol~\ref{alg:GHZ_network_complete}, since \texttt{MSG} captures the list of stars that form the entire network (see Fig.~\ref{fig:GHZ_state_dist_example}, panel~(5)), and one source at the center of each star is sufficient. Importantly, computing the number of elements in \texttt{MSG} requires no additional effort; it is a direct byproduct of executing our Protocol~\ref{alg:GHZ_network_complete}, unlike computing a minimum dominating set or identifying internal nodes. Since finding a Minimum Weight Dominating Set (MWDS) is NP-hard, we employ the approximation algorithm from Ref.~\cite{dom_set_alg_book} in our simulations. In our setting, all nodes are assigned equal weight (set to 1), so the MWDS problem reduces to the classic Minimum Dominating Set (MDS) problem. Notably, we find that the lowest number of Bell-pair sources is achieved only when the MWDS algorithm is applied to the graph generated by Protocol~\ref{alg:GHZ_network_complete}; applying MWDS directly to the MST does not yield better results (see the red and purple curves in Fig.~\ref{fig:Bell_pair_sources_ER_BA}). Finally, we observe that Barabási--Albert networks require fewer Bell‐pair sources than Erdős--Rényi networks. Interestingly, as the network size increases in the Barabási--Albert model, the gap between our protocol and the MST approach also widens.

\subsection*{Performance with noise}\label{sec:noise_case}

In this section, we consider the case that the Bell pairs, locally-created GHZ states, and local gates are noisy. We note that our analysis is more general than the one in Ref.~\cite{CC12}, which only considered Pauli noise models, particularly depolarizing noise, while we consider arbitrary (generally non-Pauli and non-unital) noise models. All proofs are presented in Methods. 

\begin{theorem}[Fidelity after Protocol~\ref{alg:Bell_star_to_GHZ} with noisy Bell pairs]\label{thm:Bell_star_to_GHZ_post_fid}
    If in Protocol~\ref{alg:Bell_star_to_GHZ} the Bell pairs are noisy and given by arbitrary two-qubit density operators $\rho_{A_1B_1}^{(1)}, \rho_{A_2B_2}^{(2)},\dotsc,\rho_{A_nB_n}^{(n)}$, then the fidelity of the state after Protocol~\ref{alg:Bell_star_to_GHZ} with respect to the target GHZ state is given by
    \begin{equation}\label{eq-Bell_star_to_GHZ_post_fid_main}
        \sum_{\vec{z}\in\{0,1\}^{n-1}} \mu_{z'}^{(1)}\mu_{z_1}^{(2)}\mu_{z_2}^{(3)}\dotsb\mu_{z_{n-1}}^{(n)},
    \end{equation}
    where for every $\vec{z}=(z_1,z_2,\dotsc,z_{n-1})\in\{0,1\}^{n-1}$, we have let $z'=z_1\oplus z_2\oplus\dotsb\oplus z_{n-1}$, and $\mu_z^{(i)}=\bra{\Phi^{z,0}}\rho^{(i)}\ket{\Phi^{z,0}}$ for $z\in\{0,1\}$. We have also made use of the Bell state vectors
    \begin{equation}\label{eq-qubit_Bell_states_main}
        \ket{\Phi^{z,x}}\coloneqq(Z^zX^x\otimes\mathbb{I})\ket{\Phi},
    \end{equation}
    for $z,x\in\{0,1\}$. In particular, if the Bell pairs are subjected to the same noise, meaning that $\mu^{(i)}=\mu_z=\bra{\Phi^{z,0}}\rho_{A_iB_i}^{(i)}\ket{\Phi^{z,0}}$ for all $i\in\{1,2,\dotsc,n\}$, $z\in\{0,1\}$, then the fidelity is equal to
    \begin{equation}\label{eq-Bell_star_to_GHZ_post_fid_main_2}
        \frac{1}{2}\big((\mu_0-\mu_1)^n+(\mu_0+\mu_1)^n\big).
    \end{equation}
\end{theorem}


We emphasize that the expression in \eqref{eq-Bell_star_to_GHZ_post_fid_main} holds for noisy Bell pairs with an arbitrary noise model for each Bell pair. Also, the expression in \eqref{eq-Bell_star_to_GHZ_post_fid_main} involves not only the fidelity of the states (with respect to $\ket{\Phi^+}$), given by $\bra{\Phi^{0,0}}\rho_{A_iB_i}^i\ket{\Phi^{0,0}}$, but also the fidelity of the states with respect to $\ket{\Phi^{1,0}}=\frac{1}{\sqrt{2}}(\ket{0,0}-\ket{1,1})$. Also, while a noise analysis was performed in Ref.~\cite{CC12}, the model considered therein was depolarizing noise only. Our results hold for arbitrary noise models, which allows for a more precise and meaningful analysis for specific qubit architectures. 

Now, given a desired fidelity $F$ for the final GHZ state and a number $n$ of noisy Bell pairs impacted by the same type of noise, how much noise can be tolerated while achieving a fidelity of at least $F$ after Protocol~\ref{alg:Bell_star_to_GHZ}? Using \eqref{eq-Bell_star_to_GHZ_post_fid_main_2}, we can simply evaluate the condition
\begin{equation}\label{eq:fidelity_condition}
    \frac{1}{2}\big((\mu_0-\mu_1)^n+(\mu_0+\mu_1)^n\big)\geq F
\end{equation}
for any given noise model. Letting the noisy Bell pair be given by $\rho_{AB}=(\id\otimes\mathcal{N})(\Phi_{AB}^{0,0})$, where $\mathcal{N}$ is the quantum channel describing the noise, we obtain the following results for different choices of $\mathcal{N}$.
\begin{itemize}[left=0cm]
    \item \textit{Depolarizing noise:} $\mathcal{N}(\cdot)=(1-p)(\cdot)+\frac{p}{3}(X(\cdot)X+Y(\cdot)Y+Z(\cdot)Z)$, where $p\in[0,1]$ is the depolarizing parameter and $X,Y,Z$ are the single-qubit Pauli operators. The corresponding noisy Bell pair is given by $\rho_{AB}=(1-p)\Phi_{AB}^{0,0}+\frac{p}{3}(\Phi_{AB}^{1,0}+\Phi_{AB}^{0,1}+\Phi_{AB}^{1,1})$, where $\Phi_{AB}^{z,x}=\ketbra{\Phi^{z,x}}{\Phi^{z,x}}$. We have that $\mu_0=1-p$ and $\mu_1=\frac{p}{3}$, so the condition \eqref{eq:fidelity_condition} evaluates to
        \begin{equation}\label{eq-Bell_star_to_GHZ_noise_tolerance_Dep}
            \left(1-\frac{4p}{3}\right)^n+\left(1-\frac{2p}{3}\right)^n\geq 2F,
        \end{equation}
        which can be solved numerically to get an upper bound on $p$; see Fig.~\ref{fig:Bell_star_to_GHZ_region}.

    \item \textit{Dephasing noise:} $\mathcal{N}(\cdot)=(1-q)(\cdot)+qZ(\cdot)Z$, where $q\in[0,1]$ is the dephasing parameter. The corresponding noisy Bell pair is given by $\rho_{AB}=(1-q)\Phi_{AB}^{0,0}+q\Phi_{AB}^{1,0}$. In this case, $\mu_0=1-q$ and $\mu_1=q$, so the condition \eqref{eq:fidelity_condition} evaluates to
        \begin{align}
            &(1-2q)^n+1\geq 2F\\
            \Rightarrow & \, q\leq\frac{1}{2}\big(1-(2F-1)^{1/n}\big).
        \end{align}

    \item \textit{Pauli noise:} Both depolarizing and dephasing are examples of Pauli noise. A general Pauli-noise channel has the form $\mathcal{N}(\cdot)=p_I(\cdot)+p_XX(\cdot)X+p_Y Y(\cdot)Y+p_ZZ(\cdot)Z$, where $p_I,p_X,p_Y,p_Z\geq 0$ are the Pauli error rates and $p_I+p_X+p_Y+p_Z=1$. Bit-flip is another special case of this noise model, given by setting $p_X\equiv p$, $p_I=1-p$, and $p_Y=p_Z=0$. The bit-phase-flip noise model corresponds to setting $p_Y\equiv p$, $p_I=1-p$, and $p_X=p_Z=0$.
    
    For arbitrary Pauli noise, the corresponding noisy Bell pair is given by $\rho_{AB}=p_I\Phi_{AB}^{0,0}+p_X\Phi_{AB}^{0,1}+p_Z\Phi_{AB}^{1,0}+p_Y\Phi_{AB}^{1,1}$. Then, $\mu_0=p_I$ and $\mu_1=p_Z$, and the condition \eqref{eq:fidelity_condition} evaluates to $\frac{1}{2}((p_I-p_Z)^n+(p_I+p_Z)^n\geq F$. Evaluating this condition specifically for bit-flip noise leads to $(1-p)^n\geq F\Rightarrow p\leq1-F^{1/n}$. An analogous condition holds for the bit-phase-flip noise model.

    \item \textit{Amplitude damping noise:} $\mathcal{N}(\cdot)=K_1(\cdot)K_1^{\dagger}+K_2(\cdot)K_2^{\dagger}$, where $K_1=\begin{pmatrix} 1 & 0 \\ 0 & \sqrt{1-\gamma}\end{pmatrix}$, $K_2=\begin{pmatrix} 0 & \sqrt{\gamma} \\ 0 & 0 \end{pmatrix}$, and $\gamma\in[0,1]$ is the damping parameter. The corresponding noisy Bell pair is given by $\rho_{AB}=\frac{\gamma}{2}\ketbra{1,0}{1,0}+\left(1-\frac{\gamma}{2}\right)\ketbra{v}{v}$, where $\ket{v}=\frac{1}{\sqrt{2-\gamma}}(\ket{0,0}+\sqrt{1-\gamma}\ket{1,1})$. In this case, we have $\mu_0=\frac{1}{4}(1+\sqrt{1-\gamma})^2$ and $\mu_1=\frac{1}{4}(1-\sqrt{1-\gamma})^2$, so the condition \eqref{eq:fidelity_condition} evaluates to
        \begin{equation}\label{eq-Bell_star_to_GHZ_noise_tolerance_AD}
            (1-\gamma)^{n/2}+\left(1-\frac{\gamma}{2}\right)^n\geq 2F.
        \end{equation}

    \item \textit{$(T_1,T_2)$ noise:} $\mathcal{N}(\cdot)=K_1(\cdot)K_1^{\dagger}+K_2(\cdot)K_2^{\dagger}+K_3(\cdot)K_3^{\dagger}$, where $K_1=\begin{pmatrix} 1 & 0 \\ 0 & \sqrt{1-\gamma-\lambda}  \end{pmatrix}$, $K_2=\begin{pmatrix} 0 & \sqrt{\gamma} \\ 0 & 0 \end{pmatrix}$, $K_3=\begin{pmatrix} 0 & 0 \\ 0 & \sqrt{\lambda}  \end{pmatrix}$, and $1-\gamma=\e^{-t/T_1}$, $\lambda=\e^{-t/T_1}-\e^{-2t/T_2}$~\cite{ghosh2012surfacecodedecoherence}. (For $T_2=2T_1$, this reduces to the amplitude-damping noise model.) The corresponding noisy Bell pair has the form
    \begin{equation}
        \rho_{AB}=\frac{1}{2}\begin{pmatrix} 1 & 0 & 0 & \sqrt{1-\gamma-\lambda} \\ 0 & 0 & 0 & 0 \\ 0 & 0 & \gamma & 0 \\ \sqrt{1-\gamma-\lambda} & 0 & 0 & 1-\gamma  \end{pmatrix},
    \end{equation}
    so the condition \eqref{eq:fidelity_condition} evaluates to
    \begin{equation}
        \e^{-nt/T_2}+\frac{1}{2^n}\big(1+\e^{-t/T_1}\big)^n\geq 2F.
    \end{equation}
    We solve this numerically for $F$ and plot it in Fig.~\ref{fig:Bell_star_to_GHZ_region}. Specifically, we let $r_1=t/T_1$, $r_2=t/T_2$ be the fractions of the $T_1$ and $T_2$ times, respectively, that the Bell pairs have evolved, and we set $T_2=2T_1$, so that $r_2=\frac{1}{2}r_1\equiv\frac{1}{2}r$. The regions indicate how long the Bell pairs can be kept in memory, as a fraction of the $T_1$ time, while achieving the desired fidelity.
        
\end{itemize}

\begin{figure}
    \centering
    \includegraphics[width=0.95\columnwidth]{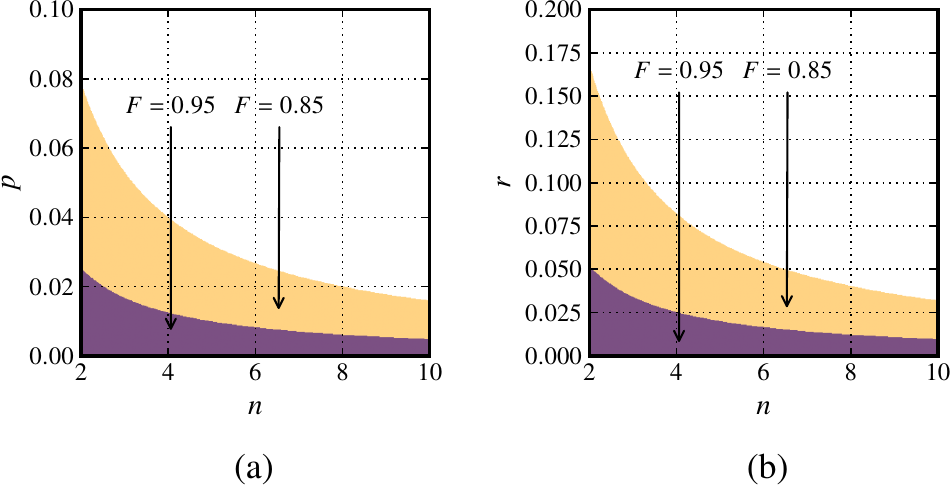}
    \caption{\textbf{Noise tolerance of Protocol~\ref{alg:Bell_star_to_GHZ}.} The desired output fidelity of the protocol is set to be at least $F=0.85,0.95$. The number of Bell pairs is $n$. Shaded regions indicate the values of $n$ and the noise parameter for which the desired fidelity can be achieved. (a) Depolarizing noise, with the tolerance regions given by the condition in \eqref{eq-Bell_star_to_GHZ_noise_tolerance_Dep}. (b) $(T_1,T_2)$ noise, assuming $T_2=2T_1$, and $r=t/T_1$ is the time spent in memory as a fraction of the $T_1$ time.
    }
    \label{fig:Bell_star_to_GHZ_region}
\end{figure}

\begin{theorem}[Fidelity after Protocol~\ref{alg:multiple_GHZ_merge_tree_protocol_main}]\label{thm:GHZ_merge_tree_protocol_post_fid}
    Consider Protocol~\ref{alg:multiple_GHZ_merge_tree_protocol_main} with $m\in\{2,3,\dotsc\}$ noisy GHZ states (of arbitrary size). Then, the fidelity of the state after the protocol with respect to the target GHZ state is given by
    \begin{equation}
        \sum_{\vec{z}\in\{0,1\}^{m-1}}\lambda_{z_1}^{(1)}\lambda_{z_2}^{(2)}\dotsb\lambda_{z_{m-1}}^{(m-1)}\lambda_{z'}^{(m)},
    \end{equation}
    where $z'=z_1\oplus z_2\oplus\dotsb\oplus z_{m-1}$, $\lambda_{z}^{(i)}=\bra{\text{GHZ}^{(z,\vec{0})}}\rho^{(i)}\ket{\text{GHZ}^{(z,\vec{0})}}$ is the fidelity of the $i^{\text{th}}$ state with respect to $\ket{\text{GHZ}^{(z,\vec{0})}}$, $i\in\{1,2,\dotsc,m\}$, and
    \begin{equation}\label{eq-GHZ_basis_main}
        \ket{\text{GHZ}_n^{(z,\vec{x})}}\coloneqq (Z^z\otimes X^{\vec{x}})\ket{\text{GHZ}_n},
    \end{equation}
    with $z\in\{0,1\}$ and $\vec{x}\in\{0,1\}^{n-1}$, is the $n$-qubit GHZ basis.
\end{theorem}


We note that the fidelity after Protocol~\ref{alg:multiple_GHZ_merge_tree_protocol_main} depends only on the number of input GHZ states, and not on the star topology of the network. Also, similarly to the fidelity after Protocol~\ref{alg:Bell_star_to_GHZ}, the fidelity after Protocol~\ref{alg:multiple_GHZ_merge_tree_protocol_main} depends only on the diagonal elements of the noisy GHZ states with respect to the GHZ basis. In particular, only the fidelities of the noisy GHZ states with respect to $\ket{\text{GHZ}_n^{(0,\vec{0})}}\equiv\ket{\text{GHZ}_n}=\frac{1}{\sqrt{2}}(\ket{0}^{\otimes n}+\ket{1}^{\otimes n})$ and $\ket{\text{GHZ}_n^{(1,\vec{0})}}=\frac{1}{\sqrt{2}}(\ket{0}^{\otimes n}-\ket{1}^{\otimes n})$ are relevant.

Consider fusing $m\in\{2,3,\dotsc\}$ noisy GHZ states according to Protocol~\ref{alg:multiple_GHZ_merge_tree_protocol_main}, where each of the noisy GHZ states has been created using Protocol~\ref{alg:Bell_star_to_GHZ}. Assuming (for simplicity and illustrative purposes) that all of the (non-leaf) nodes in the tree have the same number $n\in\{2,3,\dotsc\}$ of children nodes, and that all of the Bell pairs have the same noise acted upon them, we know from Theorem~\ref{thm:Bell_star_to_GHZ_post_fid} that the fidelity with respect to $\ket{\text{GHZ}^{(0,\vec{0})}}$ is given by
\begin{equation}
    \lambda_0^{(i)}\equiv\lambda_0=\frac{1}{2}\big((\mu_0+\mu_1)^n+(\mu_0-\mu_1)^n\big),
\end{equation}
for all $i\in\{1,2,\dotsc,m\}$. By reasoning similar to the proof of Theorem~\ref{thm:Bell_star_to_GHZ_post_fid}, it is possible to show that the fidelity of the state after Protocol~\ref{alg:Bell_star_to_GHZ} with respect to $\ket{\text{GHZ}^{(1,\vec{0})}}$ is
\begin{equation}
    \lambda_1^{(i)}\equiv\lambda_1=\frac{1}{2}\big((\mu_0+\mu_1)^n-(\mu_0-\mu_1)^n\big),
\end{equation}
for all $i\in\{1,2,\dotsc,m\}$. Therefore, by Theorem~\ref{thm:GHZ_merge_tree_protocol_post_fid}, the final fidelity with respect to the GHZ state is
\begin{align}
    &\sum_{\vec{z}\in\{0,1\}^{m-1}}\lambda_{z_1}\lambda_{z_2}\dotsb\lambda_{z_{m-1}}\lambda_{z'}\nonumber\\
    &\qquad=\sum_{\vec{z}\in\{0,1\}^{m-1}}\lambda_1^{\Abs{\vec{z}}}\lambda_0^{m-1-\Abs{\vec{z}}}\lambda_{z'}\\
    &\qquad=\frac{1}{2}\big((\lambda_0+\lambda_1)^m+(\lambda_0-\lambda_1)^m\big)\\
    &\qquad=\frac{1}{2}\big((\mu_0+\mu_1)^{nm}+(\mu_0-\mu_1)^{nm}\big).
\end{align}
This tells us that, for every noise model, the final fidelity depends essentially on the number of CNOT gates performed among all fusion operations, which is $nm$.


\paragraph*{Evaluation for an example hardware platform.}

We now evaluate the performance of Protocol~\ref{alg:Bell_star_to_GHZ} for a particular hardware platform. Our example is based on the recent experiment~\cite{knaut2024entanglementnetwork} for distributing bipartite entanglement in a network using Silicon-vacancy (SiV) centers. The electronic spins are used as the communication qubits and the nuclear spins are used as the storage/memory qubits. Over a 20~m optical-fiber link, the resulting fidelities of the entangled electronic spins with respect to the states $\Phi^{\pm}$ (taking all realistic elements into account) are~\cite{knaut2024entanglementnetwork} $\mu_0=0.74$ and $\mu_1=0.11$, where we recall the definitions of $\mu_0$ and $\mu_1$ from Theorem~\ref{thm:Bell_star_to_GHZ_post_fid}. Assuming entangled links with these fidelities in a star topology, it immediately follows from \eqref{eq-Bell_star_to_GHZ_post_fid_main_2} that the fidelity after Protocol~\ref{alg:Bell_star_to_GHZ} is $\frac{1}{2}(0.63^n+0.85^n)$ for all $n\in\{2,3,\dotsc\}$, where $n$ is the number of Bell pairs used to make the  $(n+1)$-party GHZ state. Already for $n=2$, the resulting fidelity of the three-party GHZ state is approximately 0.56, even without taking gate noise into account.


\paragraph*{Noise in the gates.}
Now, in addition to noise on the input states of Protocols~\ref{alg:Bell_star_to_GHZ} and \ref{alg:multiple_GHZ_merge_tree_protocol_main}, we want to consider noise in the local CNOT gates. Recall that the two-qubit CNOT gate is $\text{CNOT}=\ketbra{0}{0}\otimes\mathbb{I}+\ketbra{1}{1}\otimes X$, where $X$ is the Pauli-$X$ gate. Let us denote the sequence of CNOT gates in Protocol~\ref{alg:Bell_star_to_GHZ} as
\begin{multline}\label{eq-CNOT_star}
    \text{CNOT}_n\equiv\text{CNOT}_{A_1A_n}\text{CNOT}_{A_1A_{n-1}}\dotsb\\\dotsb\text{CNOT}_{A_1A_3}\text{CNOT}_{A_1A_2},
\end{multline}
where $n\in\{2,3,\dotsc\}$ is the number of qubits. It is straightforward to show that
\begin{equation}
    \text{CNOT}_n=\ketbra{0}{0}\otimes\mathbb{I}^{\otimes (n-1)}+\ketbra{1}{1}\otimes X^{\otimes (n-1)}.
\end{equation}
Now, as our model of the noisy $\text{CNOT}_n$ gate, let us assume that an arbitrary $n$-qubit quantum channel $\mathcal{M}$ is applied immediately after the $\text{CNOT}_n$ gate, which means that the noisy $\text{CNOT}_n$ gate is the quantum channel $(\cdot)\mapsto\mathcal{M}(\text{CNOT}_n(\cdot)\text{CNOT}_n)$.

In order to apply Theorems~\ref{thm:Bell_star_to_GHZ_post_fid} and \ref{thm:GHZ_merge_tree_protocol_post_fid} while taking gate noise into account, we can attempt to absorb the gate noise into the input states. This would mean finding a channel $\mathcal{M}^{\prime}$ such that
\begin{equation}
    \mathcal{M}(\text{CNOT}_n(\cdot)\text{CNOT}_n)=\text{CNOT}_n\mathcal{M}^{\prime}(\cdot)\text{CNOT}_n.
\end{equation}
Because $\text{CNOT}_n$ is a unitary, we immediately see that
\begin{equation}
    \mathcal{M}^{\prime}(\cdot)=\text{CNOT}_n\mathcal{M}(\text{CNOT}_n(\cdot)\text{CNOT}_n)\text{CNOT}_n.
\end{equation}
Letting $\mathcal{M}(\cdot)=\sum_{j=1}^r K_j(\cdot)K_j^{\dagger}$ be a Kraus representation of $\mathcal{M}$, it follows that
\begin{align}
    \mathcal{M}^{\prime}(\cdot)&=\sum_{j=1}^r \text{CNOT}_nK_j\text{CNOT}_n(\cdot)\text{CNOT}_nK_j^{\dagger}\text{CNOT}_n\\
    &=\sum_{j=1}^r K_j^{\prime}(\cdot)K_j^{\prime\dagger},
\end{align}
where $K_j^{\prime}=\text{CNOT}_nK_j\text{CNOT}_n$.

If $\mathcal{M}$ is a Pauli channel, then it can be written as $\mathcal{M}(\cdot)=\sum_{\{\mathbb{I},X,Y,Z\}^{\otimes n}}p_P\,P(\cdot)P$, where the quantities $p_P$, $P\in\{\mathbb{I},X,Y,Z\}^{\otimes n}$, form a probability distribution. This implies that the transformed Kraus operators are given by $P^{\prime}=\text{CNOT}_nP\text{CNOT}_n$. Now, because $\text{CNOT}_n$ is a Clifford gate, it holds that $P^{\prime}\in\{\mathbb{I},X,Y,Z\}^{\otimes n}$ is some other $n$-qubit Pauli operator. In other words, when the noise is any Pauli channel, the $\text{CNOT}_n$ gate can be brought outside of the action of the noise, and can therefore be absorbed into the input states of Protocol~\ref{alg:Bell_star_to_GHZ}, while still retaining the Pauli-covariant structure of the noise. This means that Theorem~\ref{thm:Bell_star_to_GHZ_post_fid} holds even when the $\text{CNOT}_n$ gate is impacted by arbitrary Pauli noise. 

It is straightforward to determine the form of the transformed Pauli operator $P'$. Let $P\equiv \I^{\vec{x}\cdot\vec{z}}X^{\vec{x}}Z^{\vec{z}}$, where $\vec{x},\vec{z}\in\{0,1\}^{n}$ and $X^{\vec{x}}\equiv X^{x_1}\otimes X^{x_2}\otimes\dotsb\otimes X^{x_n}$, $Z^{\vec{z}}=Z^{z_1}\otimes Z^{z_2}\otimes\dotsb\otimes Z^{z_n}$. Then, it holds that $P'=\text{CNOT}_n(\I^{\vec{x}\cdot\vec{z}}X^{\vec{x}}Z^{\vec{z}})\text{CNOT}_n=\I^{\vec{x}\cdot\vec{z}}X^{\vec{x}'}Z^{\vec{z}'}$, where $\vec{x}'=(x_1,x_1\oplus x_2,x_1\oplus x_3,\dotsc,x_1\oplus x_n)$ and $\vec{z}'=(z_1\oplus z_2\oplus\dotsb\oplus z_n,z_2,\dotsc,z_n)$. As an example, let us consider $n=2$, and assume a tensor-product Pauli noise channel, $\mathcal{M}=\mathcal{P}^{\otimes 2}$, where $\mathcal{P}(\cdot)=\sum_{x,z\in\{0,1\}} p_{x,z}X^x Z^z(\cdot)Z^zX^x$, with the coefficients $p_{x,z}$ forming an arbitrary probability distribution. This means that $\mathcal{M}(\cdot)=\sum_{x_1,z_1,x_2,z_2\in\{0,1\}}p_{x_1,z_1}p_{x_2,z_2}(X^{x_1}Z^{z_1}\otimes X^{x_2}Z^{z_2})(\cdot)(Z^{z_1}X^{x_1}\otimes Z^{z_2}X^{x_2})$. Then, the transformed channel is $\mathcal{M}^{\prime}(\cdot)=\sum_{x_1,z_1,x_2,z_2\in\{0,1\}}p_{x_1,z_1\oplus z_2}p_{x_1\oplus x_2,z_2}(X^{x_1}Z^{z_1}\otimes X^{x_2}Z^{z_2})(\cdot)(Z^{z_1}X^{x_1}\otimes Z^{z_2}X^{x_2})$. Observe that this channel is not necessarily a tensor-product channel, due to the correlations in the transformed probability distributions. As such, this example illustrates the important point that transforming an uncorrelated noise channel will generally result in a correlated noise channel. In particular, after absorbing the noise into the input states via this transformation, the initial product state of noisy Bell pairs may not retain the product state form, i.e., the initial input state to Protocol~\ref{alg:Bell_star_to_GHZ} may not be a product state, as assumed in Theorem~\ref{thm:Bell_star_to_GHZ_post_fid}. However, Theorem~\ref{thm:Bell_star_to_GHZ_post_fid} holds even for arbitrarily correlated input states, and we provide the precise statement of this general case in Methods. 

For non-Pauli noise, the form of $\mathcal{M}^{\prime}$ depends heavily on the specific form of $\mathcal{M}$. Taking $\mathcal{M}=\mathcal{A}_{\gamma}^{\otimes 2}$ as an example, where $\mathcal{A}_{\gamma}$ is the amplitude-damping channel defined above, it is straightforward to determine the transformed Kraus operators. Let $K_{i,j}^{\prime}=\text{CNOT}_n(K_i\otimes K_j)\text{CNOT}_n$ be the transformed Kraus operators, for $i,j\in\{0,1\}$. Then,
\begin{align}
    K_{1,1}^{\prime}&=\begin{pmatrix} 1 & 0 & 0 & 0 \\ 0 & \sqrt{1-\gamma} & 0 & 0 \\ 0 & 0 & 1-\gamma & 0 \\ 0 & 0 & 0 & \sqrt{1-\gamma} \end{pmatrix},\\
   K_{1,2}^{\prime}&=\begin{pmatrix} 0 & \sqrt{\gamma} & 0 & 0 \\ 0 & 0 & 0 & 0 \\ 0 & 0 & 0 & 0 \\ 0 & 0 & \sqrt{\gamma(1-\gamma)} & 0\end{pmatrix},\\
   K_{2,1}^{\prime}&=\begin{pmatrix} 0 & 0 & 0 & \sqrt{\gamma} \\ 0 & 0 & \sqrt{\gamma(1-\gamma)} & 0 \\ 0 & 0 & 0 & 0 \\ 0 & 0 & 0 & 0\end{pmatrix},\,
   K_{2,2}^{\prime}=\begin{pmatrix} 0 & 0 & \gamma & 0 \\ 0 & 0 & 0 & 0 \\ 0 & 0 & 0 & 0 \\ 0 & 0 & 0 & 0\end{pmatrix}.
\end{align}
As in the case of Pauli channels, the transformed channel is not a tensor-product channel.

Analogous reasoning can be applied to deal with gate noise in Protocol~\ref{alg:multiple_GHZ_merge_tree_protocol_main}, such that Theorem~\ref{thm:GHZ_merge_tree_protocol_post_fid} can be applied by absorbing the gate noise into the input states.

\section*{Discussion}

In this work, we consider the distribution of GHZ states in arbitrary Bell-pair networks. While this task is not new, we adopt a novel perspective on evaluating protocols to achieve this task by considering the gate cost and the Bell-pair source cost of the protocol, in addition to the number of Bell pairs used. Notably, the gate count imposes a more restricted form of local operations and classical communications (LOCC). Unlike the traditional LOCC framework, in which all local operations are considered free, our approach emphasizes the constraints associated with gate usage. Furthermore, to minimize the number of consumed Bell pairs when distributing a GHZ state among a subset of nodes in a network, a minimum Steiner tree must generally be identified; however, the Steiner tree problem is NP-hard~\cite{garey1977SteinertreeNP}. It is therefore interesting to ask whether other, polynomial-time strategies can be employed to distribute a GHZ state, while still maintaining near-optimal Bell pair consumption, gate usage, source cost. To the best of our knowledge, the number of Bell-pair sources has not been considered in previous works.

We present protocols for GHZ entanglement distribution that use a number of gates that is linear in the size of the subgraph corresponding to the desired nodes to be entangled. Notably, this gate cost is independent of the structure and topology of the network. We also propose an alternative protocol that identifies connected subgraphs using the breadth-first search (BFS) algorithm. While this approach does not yield the optimal Bell pair count, it avoids the need to find Steiner trees, it runs in polynomial time, and we provide numerical evidence that it achieves near-optimal Bell pair consumption. Furthermore, by relaxing optimality in Bell-pair consumption for the sake of computational efficiency, our protocol has better complexity compared to polynomial-time Steiner tree approximation algorithms~\cite{kou1981steiner,mehlhorn1988steiner,robins2005steinertreeapproximation}.


Unlike prior work, which focuses primarily on minimizing the number of Bell pairs and therefore involves solving the Steiner tree problem (either exactly or with an approximate algorithm), our objective has been stronger: our goal was to minimize not only the number of Bell pairs, but also the number of Bell pair sources. What we have shown is that doing this necessitates a different perspective. Specifically, we have shown that we must find not only an approximation of a minimum Steiner tree, in order to achieve a minimal number of Bell pairs, but also an approximation of the minimum dominating set size. In other words, we have shown that focusing on Steiner trees alone will not achieve the minimum number of Bell-pair sources, and this is where our algorithm has its advantage, because it can simultaneously achieve near-optimal Bell-pair consumption and near-optimal Bell-pair source count.

To evaluate our protocols, we performed Monte Carlo simulations on Erd\H{o}s--R\'{e}nyi, Barab\'{a}si--Albert, and Waxman random network models, which represent increasingly realistic network models. For these models, our protocol shows enhanced performance compared to the MMG protocol in dense networks, achieving a nearly quadratic advantage compared to the MMG protocol by minimizing the number of gates required while keeping the number of Bell pairs and sources of Bell pairs as low as possible. Additionally, we analyzed our protocol in the context of noisy Bell pairs and noisy gates, and we provided explicit analytic expressions for the fidelity of the states after our protocols for arbitrary noise models.

\paragraph*{Future directions.} While our current focus has been on creating GHZ states, exploring the distribution of other graph states, such as cluster states, would be relevant for distributed and measurement-based quantum computing. Additionally, analyzing in detail the robustness of our approach against link and node failures is of particular interest, especially in light of classical and quantum research on network resilience, as demonstrated by Refs.~\cite{albert2000attacktolerancenetworks,DKD18,coutinho2022robustness,sadhu2023practical}. Intuitively, by creating local, small-scale GHZ states independently, our protocol is less susceptible to the failure of any single GHZ state creation, which might result from the failure of nearest-neighbor entanglement generation in networks of lossy optical fiber channels. In cases of failure, the other GHZ states can be stored in memory, potentially reducing overall waiting times in the presence of link and node failures. Studying the waiting time of our protocol would generalize studies on the waiting time of linear repeater chains for bipartite entanglement, e.g., Refs.~\cite{KMSD19,waiting_time}, to the multipartite setting.

Currently, our protocol requires full, global knowledge of the network. Therefore, developing variants of our protocol that rely primarily on local knowledge and require only a few rounds of classical communication to establish global awareness, leveraging the ideas of, e.g., Ref.~\cite{haldar2024multiplexing}, would be beneficial.


Beyond theoretical studies, implementing our algorithm on various experimental platforms is also of great interest. Another intriguing approach is to add an optimization layer to our protocol enhance the way the stars are selected and merged, particularly in the context of lossy Bell pairs and lossy memory stations. We anticipate that this can be accomplished using reinforcement learning, by expressing our problem as a Markov decision process (MDP) and leveraging the tools from Refs.~\cite{Kha21b,khatri2022networkMDP,haldar2023fastreliable,simon_RL,evgeny_swap_asap,waiting_time}. Moreover, employing further tools such as (probability) generating functions, the protocols could be optimized in the presence of noise~\cite{tim_coopmans_2024,lars}.

\section*{Methods}

\subsection*{Proof of Theorem~\ref{complexity_complete_proof}}
Protocol~\ref{alg:GHZ_network_complete} consists of three main components:
\begin{enumerate}
    \item Picking stars, which contains two elements: (i) Computing the degree distribution table, which takes $O(N+E)$ time steps~\cite{bfs}; (ii) Sorting the degree distribution in decreasing order takes $O(N \log N)$ time steps, using standard comparison-based sorting algorithms such as mergesort or heapsort~\cite{bfs}.
    \item Modifying stars, which takes $O(E + K)$ time steps, where $K$ is the total number of stars, because we visit all of the nodes for each star to remove extra common nodes. The number of nodes visited in each star corresponds to its degree plus one. Consequently, the total number of nodes visited across all stars is $E+K$. 
    \item Fusions, which take $O(N \cdot K)$ time steps, because $N$ nodes are visited for each of the $K-1$ fusion operations.
\end{enumerate}
Putting everything together, the total time complexity is $O(N+E +N\log N + E + K + N \cdot K)$. Since a graph can have at most $E=O(N^2)$ edges (as in the case of a complete graph), and because $K = O(N)$, we conclude that the total time complexity is $O(E) = O(N^2)$.

\subsection*{Proof of Theorem~\ref{complexity_subset_proof}}

Protocol~\ref{alg:GHZ_network_subset} consists of two main components: (1) BFS traversal, which has a complexity of $O(N + E)$~\cite{bfs}, because it visits all $N$ nodes and $E$ edges; (2) leaf node removal, with a complexity of $O(N)$, as it involves going through every node in the BFS tree one time. Combining these, the total time complexity is $O(N + E + N) = O(N + E)$. In a general graph, the number $E$ of edges is at most $O(N^2)$ (as in the case of a complete graph), so the overall time complexity can be simplified to $O(N^2)$.

\subsection*{Proof of Theorem~\ref{thm:gates_complete_case}}

Let us assume there are a total of $N$ nodes and these are distributed in $k$ GHZ states of different sizes. The number of nodes in each GHZ state is given by $n_1, n_2, n_3,\dotsc,n_k$. The number of gates required to create these GHZ states is given by (according to point 1 above) $(n_1-2), (n_2-2), (n_3-2),\dotsc,(n_k-2)$. Adding up all of these gives us  $n_1 + n_2 +n_3 +\dotsb+n_k  -2k$. As there are $k$ GHZ states, there will be $k-1$ nodes that are repeated one time. Therefore, the total number of nodes in all of the small GHZ states should be $N+k-1$, i.e., $n_1+n_2+n_3+\dotsb+n_k=N+k-1$, which means that the gate cost is $N+k-1-2k=N-k-1$. Now, let us add the gates for the merging of the GHZ states (explained in point 2 above). As there are $k$ GHZ states, there will be $k-1$ merges in total, and each merge requires one gate (see point 2 above), which means that the total number of gates required is $N-k-1+k-1=N-2$.

\subsection*{Proof of Proposition~\ref{prop:gate_cost_MMG}}

The MMG protocol of Ref.~\cite{MMG19} proceeds by starting with a minimum spanning tree for the graph. Then, an arbitrary leaf node of the tree is chosen, and the star expansion protocol is applied successively on non-leaf neighbors of the chosen leaf node until all of the nodes are included in the GHZ state. If the number of non-leaf neighbors is $L$, then the gate cost is $\sum_{i=1}^L g_1(k_i)$. We can bound this from above by summing over all nodes, i.e.,
\begin{align}
    \sum_{i=1}^L g_1(k_i)&\leq \sum_{i=1}^Ng_1(k_i)\\
    &=\frac{1}{2}\sum_{i=1}^N k_i^2+\frac{3}{2}\sum_{i=1}^N k_i-N\\
    &\leq \frac{1}{2}\left(\sum_{i=1}^N k_i\right)^2+\frac{3}{2}\cdot 2(N-1)-N\\
    &=\frac{1}{2}(2(N-1))^2+2N-3\\
    &=2N^2-2N-1\\
    &=O(N^2),
\end{align}
where we have made use of the fact that $\sum_{i=1}^N k_i=2(N-1)$, i.e., the sum of the degrees is equal to twice the number of edges, and the number of edges in the minimum spanning tree is necessarily $N-1$. 

\subsection*{Proof of Proposition~\ref{prop:gate_cost_fisher}}

Consider a chain topology for the network, with nodes labeled from 1 to $N$ from left to right. Suppose that the root node is node 1. In this setup, the protocol of Ref.~\cite{fischer2021graph} requires ``connection transfers'' from each node to node 1, with each transfer utilizing one gate. The farthest node, node $N$, requires $N - 2$ connection transfers to reach node 1. The next farthest node, node $N - 1$, requires $N - 3$ connection transfers, and so on. Consequently, the total number of gates used is given by $(N-2) + (N-3) + (N-4) +\dotsb+ (N-N)=\sum_{j=2}^N(N-j)=\frac{N(N-3)}{2} +1 =O(N^2)$.

\subsection*{Proof of Theorem~\ref{full_edges_proof}}

The proof consists of two parts. First, we establish that the lower bound is $N-1$; second, we show that our protocol achieves this bound.

To begin, the network must be connected, as it is otherwise impossible to generate entanglement among all nodes in the LOCC setting. The minimum number of edges required to connect a graph of $N$ nodes is $N-1$, which implies that at least $N-1$ Bell pairs are needed. This gives us the lower bound of $N-1$ Bell pairs.

To minimize the number of Bell pairs, a single common node between any two stars is necessary and sufficient for merging. If more than one common node is present, additional Bell pairs are consumed. Therefore, in step 10 of Protocol~\ref{alg:GHZ_network_complete}, we modify the stars to ensure that only one common node connects with the rest. This ensures that we require only $N-1$ Bell pairs to create the GHZ state.

\subsection*{Proof of Theorem~\ref{subset_edges_proof}}

For the lower bound, the minimum number of Bell pairs used is determined analogous to the complete case, resulting in $S-1$.

For the upper bound, the maximum number of Bell pairs will be used when all the nodes in the subset are maximally distant from each other. By definition, the diameter gives the maximum distance in terms of the number of Bell pairs. Since there are $S$ maximally distant nodes, the maximum number of Bell pairs used is $d_S(S-1)$.

\subsection*{Proof of Theorem~\ref{thm:bell_pair_dominating_set}}

Let the quantum network be represented by an undirected graph $G = (V, E)$, where we recall that each node $v \in V$ is a quantum node (e.g., a repeater or memory station), and each edge $\{u, v\} \in E$ represents a quantum channel capable of distributing a Bell pair between adjacent nodes.

A Bell pair source placed at a node $s \in V$ can generate a Bell pair locally and distribute one qubit to itself and the other to any one of its neighbors via the quantum channel. Thus, a node $v \in V$ is said to be \emph{covered} if it either hosts a source or is adjacent to a node that hosts a source.

We define a set $D \subseteq V$ of nodes that host Bell pair sources. Then, the set of covered nodes is given by
\begin{equation}
\bigcup_{v \in D} \left( \{v\} \cup N(v) \right),
\end{equation}
where $N(v)$ denotes the set of neighbors of node $v$. The objective is to find the smallest such set $D$ such that this union equals $V$, i.e.,
\begin{equation}
    \bigcup_{v \in D} \left( \{v\} \cup N(v) \right) = V.
\end{equation}
This is precisely the definition of a \emph{dominating set} in graph theory: a set $D \subseteq V$ such that every node $u \in V$ is either in $D$ or adjacent to a node in $D$. Therefore, finding the minimum number of Bell pair sources required to cover the network is equivalent to finding a minimum dominating set in $G$. Specifically, the \textit{domination number} $\gamma(G)\coloneqq\min\big\{|D|:D\text{ is a dominating set of }G\big\}$ is a lower bound on the minimum number of Bell-pair sources in the network.

\subsection*{Random network models}

Let us first briefly describe the three random network models that we consider.

\paragraph*{Erd\H{o}s--R\'{e}nyi model.} The Erdős–Rényi (ER) random network ensemble, commonly referred to as the ``$G(N, p)$ model'', is defined as follows~\cite{erdosrenyi1959,erdosrenyi1960}:
\begin{itemize}
    \item $N\in\{1,2,\dotsc\}$: The number of nodes in the graph.
    \item For every pair of (distinct) nodes, add an edge between them with probability $p\in[0,1]$.
\end{itemize}
In general, graphs in the ER ensemble will not be connected, particularly for small values of $p$. However, in order to ensure that all nodes of the network are included in the final GHZ state, we require the graph to be connected. To address this, we introduce a slight modification to ensure the graph remains connected. The modification ensures connectivity by initially adding a random edge for each node, guaranteeing at least one connection, and then adding additional edges based on the probability $p$. 

\paragraph*{Barab\'{a}si--Albert model.} The Barab\'{a}si--Albert (BA) model~\cite{barabasi1999scalefree} generates scale-free networks, which are characterized by a power-law degree distribution. In such scale-free networks, a few nodes (hubs) have many connections, while most nodes have only a few. As such, it represents a realistic model of the internet. 
The construction of a random graph in the BA model is done as follows. We let $c\in\{1,2,\dotsc\}$, and we suppose that the total number $N$ of nodes satisfies $N\geq c+1$. Then:
\begin{enumerate}
    \item Start with the complete graph (i.e., all-to-all connected graph) of $c+1$ nodes. 
    \item Until the total number $N$ of nodes is reached, add a new node with $c$ edges that links the new node to $c$ different nodes already present in the network. The probability that a new node connects to an existing node $i$ with degree $k_i$ is given by the so-called ``preferential attachment'' model, in which nodes of higher degree are more likely to receive a new edge. Specifically, the probability is $k_i/\sum_{j} k_j$, where the sum in the denominator is over all existing nodes $j$ in the network.
\end{enumerate}
Observe that the parameter $c$ is the minimum possible degree of a node. In other words, every node in every graph generated according to this model has degree at least $c$. For the sake of comparison with the ER model, we consider values of $c$ given by some fraction of $N$, i.e., $c=\ceil{Np}$, where $p\in[0,1]$.

\paragraph*{Waxman model.}

In order to analyse our protocol on more realistic networks, namely, photonic networks based on fiber-optic links, we consider the Waxman model~\cite{waxman1988}. The Waxman model 
incorporates aspects of spatial and geographical proximity by using random geometric graphs~\cite{penrose2003_RGG_book} with a particular distance-dependent probability distribution for the edges. Graphs in this model are generated as follows:

\begin{itemize}
\item Nodes are distributed uniformly at random over a circle of diameter $D$. 
\item The probability of creating an edge between any two nodes $u$ and $v$ depends on their Euclidean distance $d(u,v)$. Specifically, the probability $P(u,v)$ of an edge between two nodes $u$ and $v$ is given by $P(u,v)=\beta\exp[-d(u,v)/L_0]$, where $d(u,v)$ is the distance between the nodes, $L_0$ is the attenuation length, and $\beta\in(0,1]$ is a parameter that controls the density of the graphs.
\end{itemize}
As with the ER model, this model does not guarantee a connected graph. To address this, we generate data by varying the diameter 
$D$ and focusing exclusively on creating a GHZ state among the nodes in the largest connected component of the graph. 


\subsection*{Proof of Theorem~\ref{thm:Bell_star_to_GHZ_post_fid}}\label{sec:Bell_star_to_GHZ_post_fid_pf}

Let us first prove that Protocol~\ref{alg:Bell_star_to_GHZ} performs as claimed. We use the abbreviation $A_{1:n}\equiv A_1A_2\dotsb A_n$, similarly for $B_{1:n}$. The initial state is
\begin{align}
    &\ket{\Phi}_{A_1B_1}\otimes\ket{\Phi}_{A_2B_2}\otimes\dotsb\otimes\ket{\Phi}_{A_nB_n}\nonumber\\
    &\quad=\frac{1}{\sqrt{2^n}}\sum_{\vec{\alpha}\in\{0,1\}^n}\ket{\vec{\alpha}}_{A_{1:n}}\ket{\vec{\alpha}}_{B_{1:n}}\\
    &\quad=\frac{1}{\sqrt{2^n}}\sum_{\vec{\beta}\in\{0,1\}^{n-1}}\Big(\ket{0}_{A_1}\ket{\vec{\beta}}_{A_{2:n}}\ket{0}_{B_1}\ket{\vec{\beta}}_{B_{2:n}}\nonumber\\[-1em]
    &\qquad\qquad\qquad\qquad+\ket{1}_{A_1}\ket{\vec{\beta}}_{A_{2:n}}\ket{1}_{B_1}\ket{\vec{\beta}}_{B_{2:n}}\Big),
\end{align}
where we have made use of the abbreviation $\ket{\vec{\beta}}_{A_{2:n}}\equiv \ket{\beta_1}_{A_2}\otimes\ket{\beta_2}_{A_3}\otimes\dotsb\otimes\ket{\beta_{n-1}}_{A_n}$. Then, after applying the CNOT gates, we obtain
\begin{align}
    &\frac{1}{\sqrt{2^n}}\sum_{\vec{\beta}\in\{0,1\}^{n-1}}\Big(\ket{0}_{A_1}\ket{\vec{\beta}}_{A_{2:n}}\ket{0}_{B_1}\ket{\vec{\beta}}_{B_{2:n}}\nonumber\\[-1em]
    &\qquad\qquad\qquad\qquad+\ket{1}_{A_1}\ket{\vec{\beta}\oplus\vec{1}}_{A_{2:n}}\ket{1}_{B_1}\ket{\vec{\beta}}_{B_{2:n}}\Big)\\
    &=\frac{1}{\sqrt{2^n}}\sum_{\vec{\beta}\in\{0,1\}^{n-1}}\Big(\ket{0}_{A_1}\ket{\vec{\beta}}_{A_{2:n}}\ket{0}_{B_1}\ket{\vec{\beta}}_{B_{2:n}}\nonumber\\[-1em]
    &\qquad\qquad\qquad\qquad+\ket{1}_{A_1}\ket{\vec{\beta}}_{A_{2:n}}\ket{1}_{B_1}\ket{\vec{\beta}\oplus\vec{1}}_{B_{2:n}}\Big).
\end{align}
From this, we can see that upon measurement of the qubits $A_2,A_3,\dotsc,A_n$, we obtain every outcome string $\vec{\beta}\in\{0,1\}^{n-1}$ with probability $\frac{1}{2^{n-1}}$, and conditioned on this outcome the state of the remaining qubits is 
\begin{multline}
    \frac{1}{\sqrt{2}}\left(\ket{0}_{A_1}\ket{0}_{B_1}\ket{\vec{\beta}}_{B_{2:n}}+\ket{1}_{A_1}\ket{1}_{B_1}\ket{\vec{\beta}\oplus\vec{1}}_{B_{2:n}}\right)\\=X_{B_2}^{\beta_1}\otimes X_{B_3}^{\beta_2}\otimes\dotsb\otimes X_{B_n}^{\beta_{n-1}}\ket{\text{GHZ}_{n+1}}_{A_1B_{1:n}},
\end{multline}
where $\beta_i\in\{0,1\}$ corresponds to the measurement outcome for the qubit $A_i$, $i\in\{2,3,\dotsc,n\}$. With the convention that $X^0\equiv\mathbb{I}$, we see that by applying the Pauli-$X$ gate to every qubit $B_i$ whose corresponding measured qubit $A_i$ had outcome 1, we obtain the desired GHZ state on the qubits $A_1,B_1,B_2,\dotsc,B_n$.

This protocol has the following representation as an LOCC quantum channel:
\begin{multline}\label{eq-Bell_star_to_GHZ_channel}
    \mathcal{L}(\rho_{A_{1:n}B_{1:n}})\\=\sum_{\vec{\beta}\in\{0,1\}^{n-1}} \left(\bra{\vec{\beta}}_{A_{2:n}}\text{CNOT}_n^{\dagger}\otimes X_{B_{2:n}}^{\vec{\beta}}\right)\left(\rho_{A_{1:n}B_{1:n}}\right)\\\left(\text{CNOT}_n\ket{\vec{\beta}}_{A_{2:n}}\otimes X_{B_{2:n}}^{\vec{\beta}}\right),
\end{multline}
where we have made use of the abbreviations
\begin{align}
    \text{CNOT}_n&\equiv \text{CNOT}_{A_1A_n}\text{CNOT}_{A_1A_{n-1}}\dotsb\nonumber\\
    &\qquad\qquad\qquad\text{CNOT}_{A_1A_3}\text{CNOT}_{A_1A_2},\\
    X_{B_{2:n}}^{\vec{\beta}}&\equiv X_{B_2}^{\beta_1}\otimes X_{B_3}^{\beta_2}\otimes\dotsb\otimes X_{B_n}^{\beta_{n-1}}.
\end{align}
for all $\vec{\beta}=(\beta_1,\beta_2,\dotsc,\beta_{n-1})\in\{0,1\}^{n-1}$.

For the proof of Theorem~\ref{thm:Bell_star_to_GHZ_post_fid}, we start by noting that
\begin{align}
    \ket{v(\vec{\beta})}&\coloneqq\left(\text{CNOT}_n\ket{\vec{\beta}}_{A_{2:n}}\otimes X_{B_{2:n}}^{\vec{\beta}}\right)\ket{\text{GHZ}_{n+1}}_{A_1B_{1:n}}\\
    &=\frac{1}{\sqrt{2}}\Big(\text{CNOT}_n\ket{0}_{A_1}\ket{\vec{\beta}}_{A_{2:n}}\otimes\ket{0}_{B_1}\ket{\vec{\beta}}_{B_{2:n}}\nonumber\\
    &\qquad+\text{CNOT}_n\ket{1}_{A_1}\ket{\vec{\beta}}_{B_{2:n}}\otimes\ket{1}_{B_1}\ket{\vec{\beta}\oplus\vec{1}}_{B_{2:n}}\Big)\\
    &=\frac{1}{\sqrt{2}}\Big(\ket{0}_{A_1}\ket{\vec{\beta}}_{A_{2:n}}\otimes\ket{0}_{B_1}\ket{\vec{\beta}}_{B_{2:n}}\nonumber\\
    &\qquad+\ket{1}_{A_1}\ket{\vec{\beta}\oplus\vec{1}}_{A_{2:n}}\otimes\ket{1}_{B_1}\ket{\vec{\beta}\oplus\vec{1}}_{B_{2:n}}\Big)\\
    &=\frac{1}{\sqrt{2}}\sum_{x\in\{0,1\}}\!\!\ket{x,x}_{A_1B_1}\bigotimes_{i=1}^{n-1}\ket{\beta_i\oplus x,\beta_i\oplus x}_{A_{i+1}B_{i+1}}.
\end{align}
Now, we use the fact that
\begin{equation}\label{eq-comp_diag_to_Bell}
    \ket{x}_A\ket{x}_B=\frac{1}{\sqrt{2}}\sum_{z\in\{0,1\}}(-1)^{xz}\ket{\Phi^{z,0}},
\end{equation}
for all $x\in\{0,1\}$. With this, we obtain
\begin{widetext}
\begin{align}
    \ket{v(\vec{\beta})}&=\frac{1}{\sqrt{2}}\frac{1}{\sqrt{2^n}}\sum_{x\in\{0,1\}}\sum_{z_0,z_1,\dotsc,z_{n-1}\in\{0,1\}}(-1)^{xz_0}\ket{\Phi^{z_0,0}}_{A_1B_1}\bigotimes_{i=1}^{n-1}\sum_{z_i\in\{0,1\}}(-1)^{(\beta_i\oplus x)z_i}\ket{\Phi^{z_i,0}}_{A_{i+1}B_{i+1}}\\
    &=\frac{1}{\sqrt{2}}\frac{1}{\sqrt{2^n}}\sum_{z_0,z_1,\dotsc,z_{n-1}\in\{0,1\}}\sum_{x\in\{0,1\}}(-1)^{x(z_0+z_1+\dotsb+z_{n-1})}(-1)^{\beta_1z_1+\beta_2z_2+\dotsb+\beta_{n-1}z_{n-1}}\ket{\Phi^{z_0,0}}_{A_1B_1}\bigotimes_{i=1}^{n-1}\ket{\Phi^{z_i,0}}_{A_{i+1}B_{i+1}}\\
    &=\frac{1}{\sqrt{2^{n-1}}}\sum_{\vec{z}\in\{0,1\}^{n-1}}(-1)^{\vec{\beta}\cdot\vec{z}}\ket{\Phi^{z',0}}_{A_1B_1}\bigotimes_{i=1}^{n-1}\ket{\Phi^{z_i,0}}_{A_{i+1}B_{i+1}},
\end{align}
where for the final equality we used the fact that
\begin{equation}
    \sum_{x\in\{0,1\}}(-1)^{x(z_0+z_1+\dotsb+z_{n-1})}=2\delta_{z_0+z_1+\dotsb+z_{n-1},0}=2\delta_{z_0,z_1\oplus z_2\oplus\dotsb\oplus z_{n-1}},
\end{equation}
and we let $z'\equiv z_1\oplus z_2\oplus\dotsb\oplus z_{n-1}$. Therefore, for an arbitrary state $\rho_{A_{1:n}B_{1:n}}$, we have
\begin{align}
    &\bra{\text{GHZ}_{n+1}}_{A_1B_{1:n}}\mathcal{L}\left(\rho_{A_{1:n}B_{1:n}}\right)\ket{\text{GHZ}_{n+1}}_{A_1B_{1:n}}\nonumber\\
    &\quad=\sum_{\vec{\beta}\in\{0,1\}^{n-1}}\bra{v(\vec{\beta})}\rho_{A_{1:n}B_{1:n}}\ket{v(\beta)}\\
    &\quad=\frac{1}{2^{n-1}}\sum_{\vec{\beta},\vec{z},\vec{y}\in\{0,1\}^{n-1}}(-1)^{\vec{\beta}\cdot\vec{z}}(-1)^{\vec{\beta}\cdot\vec{y}}\left(  \bra{\Phi^{z',0}}_{A_1B_1}\bigotimes_{i=1}^{n-1}\bra{\Phi^{z_i,0}}_{A_{i+1}B_{i+1}}\right)\left(\rho_{A_{1:n}B_{1:n}}\right)\left(\ket{\Phi^{y',0}}_{A_1B_1}\bigotimes_{i=1}^{n-1}\ket{\Phi^{y_i,0}}_{A_{i+1}B_{i+1}}\right)\\
    &\quad=\sum_{\vec{z}\in\{0,1\}^{n-1}} \left(\bra{\Phi^{z',0}}_{A_1B_1}\bigotimes_{i=1}^{n-1}\bra{\Phi^{z_i,0}}_{A_{i+1}B_{i+1}}\right)\left(\rho_{A_{1:n}B_{1:n}}\right)\left(\ket{\Phi^{z',0}}_{A_1B_1}\bigotimes_{i=1}^{n-1}\ket{\Phi^{z_i,0}}_{A_{i+1}B_{i+1}}\right),\label{eq-Bell_star_to_GHZ_fidelity_pf}
\end{align}
where to obtain the final equality we used the fact that
\begin{equation}
    \sum_{\vec{\beta}\in\{0,1\}^{n-1}}(-1)^{\vec{\beta}\cdot\vec{z}}(-1)^{\vec{\beta}\cdot\vec{y}}=\sum_{\vec{\beta}\in\{0,1\}^{n-1}}(-1)^{\vec{\beta}\cdot(\vec{z}\oplus\vec{y})}=2^{n-1}\delta_{\vec{z},\vec{y}},
\end{equation}
for all $\vec{z},\vec{y}\in\{0,1\}^{n-1}$. As \eqref{eq-Bell_star_to_GHZ_fidelity_pf} holds for all states $\rho_{A_{1:n}B_{1:n}}$, it holds also for the tensor-product state in the statement of the proposition. The proof is therefore complete.

\end{widetext}

\subsection*{Fusion of GHZ states in a star topology}\label{sec-GHZ_multiple_merge}


Consider $N\in\{2,3,\dotsc\}$ GHZ states, $\ket{\text{GHZ}_{n_1+1}}_{A_1B_{1:n_1}^1}, \ket{\text{GHZ}_{n_2+1}}_{A_1B_{1:n_2}^2},\dotsc,\ket{\text{GHZ}_{n_N+1}}_{A_NB_{1:n_N}^N}$, such that the qubits $A_i$, $i\in\{1,2,\dotsc,N\}$, are located at the same node. The protocol for fusing all of these GHZ states into the one large GHZ state $\ket{\text{GHZ}_{1+N_{\text{tot}}}}_{A_1B_{1:n_1}^2\dotsb B_{1:n_N}^N}$, with $N_{\text{tot}}=n_1+n_2+\dotsb+n_N$, is shown in Protocol~\ref{alg:multiple_GHZ_merge_protocol}.

\begin{algorithm}[H]
\caption{Fusion of GHZ states in a star topology~\cite{kruszynska2006purificationgraphstates,PWD18}}\label{alg:multiple_GHZ_merge_protocol}
\begin{algorithmic}[1]
\Require $N\in\{2,3,\dotsc\}$ GHZ states in a star topology, with one qubit of every GHZ state at a central node.
\Ensure GHZ state shared by all nodes. 
\State Apply the gates $\text{CNOT}_{A_1A_2},\text{CNOT}_{A_1A_3},\dotsc,\text{CNOT}_{A_1A_N}$.
\State Measure each of the qubits $A_2,A_3,\dotsc,A_N$ in the basis $\{\ket{0},\ket{1}\}$.
\State The outcome of measurement on $A_2$ is communicated to $B_{1:n_2}^2$, the outcome of measurement on $A_3$ is communicated to $B_{1:n_3}^3$, etc.
\State For every set $B_{1:n_i}^i$ of qubits, $i\in\{2,3,\dotsc,N\}$: if the outcome communicated by $A_i$ is 0, then do nothing; if the outcome communicated by $A_i$ is 1, then apply the Pauli-$X$ gate to all qubits $B_{1:n_i}^i$.
\end{algorithmic}
\end{algorithm}
The number of gates needed in Protocol~\ref{alg:multiple_GHZ_merge_protocol} is the number of $\text{CNOT}$ gates in the first step, which is $N-1$.
    
Let us verify that Protocol~\ref{alg:multiple_GHZ_merge_protocol} works as claimed. The initial state is
\begin{align}
    &\bigotimes_{j=1}^N \ket{\text{GHZ}_{n_j+1}}_{A_jB_{1:n_j}^j}\nonumber\\
    &\quad=\frac{1}{\sqrt{2^N}}\left(\sum_{x_1\in\{0,1\}} \ket{x_1}_{A_1}\otimes\ket{x_1}_{B_{1:n_1}^1}^{\otimes n_1}  \right)\nonumber\\
    &\qquad\qquad\qquad\otimes\left(\sum_{x_2\in\{0,1\}}\ket{x_2}_{A_2}\otimes\ket{x_2}_{B_{1:n_2}^2}^{\otimes n_2}  \right)\otimes\dotsb\nonumber\\
    &\qquad\qquad\qquad\otimes\left(\sum_{x_N\in\{0,1\}}\ket{x_N}_{A_N}\otimes\ket{x_N}_{B_{1:n_N}^N}^{\otimes n_N}  \right).
\end{align}
After the CNOT gates $\text{CNOT}_{A_1A_2},\text{CNOT}_{A_1A_3},\dotsc,\text{CNOT}_{A_1A_N}$, we obtain the state
\begin{multline}
    \frac{1}{\sqrt{2^N}}\sum_{x_1,x_2,\dotsc,x_N\in\{0,1\}}\ket{x_1}_{A_1}\otimes\ket{x_1}_{B_{1:n_1}^1}^{\otimes n_1}\otimes\ket{x_1\oplus x_2}_{A_2}\\\otimes\ket{x_2}_{B_{1:n_2}^2}^{\otimes n_2}\otimes\dotsb\otimes\ket{x_1\oplus x_N}_{A_N}\otimes\ket{x_{n_N}}_{B_{1:n_N}^N}^{\otimes n_N}.
\end{multline}
Then, for every collection $(y_1,y_2,\dotsc,y_{N-1})\in\{0,1\}^{N-1}$ of measurement outcomes corresponding to the $Z$-basis measurement of the qubits $A_2,A_3,\dotsc,A_N$, we obtain the state
\begin{widetext}
\begin{align}
    &\frac{1}{\sqrt{2}}\sum_{x_1,x_2,\dotsc,x_N\in\{0,1\}^N}\ket{x_1}_{A_1}\otimes\ket{x_1}_{B_{1:n_1}^1}^{\otimes n_1}\otimes\braket{y_1}{x_1\oplus x_2}\ket{x_2}_{B_{1:n_2}^2}^{\otimes n_2}\otimes\dotsb\otimes\braket{y_{N-1}}{x_1\oplus x_N}\ket{x_N}_{B_{1:n_N}^N}^{\otimes n_N}\\
    &\quad=\sum_{x_1\in\{0,1\}}\ket{x_1}_{A_1}\otimes\ket{x_1}_{B_{1:n_1}^1}^{\otimes n_1}\otimes\ket{x_1\oplus y_1}_{B_{1:n_2}^2}^{\otimes n_2}\otimes\dotsb\otimes\ket{x_1\oplus y_{N-1}}_{B_{1:n_N}^N}^{\otimes n_N}\\
    &\quad = \left(\mathbb{I}_{A_1B_{1:n_1}^1}\otimes (X^{y_1})_{B_{1:n_2}^2}^{\otimes n_2}\otimes\dotsb\otimes (X^{y_{N-1}})_{B_{1:n_N}^N}^{\otimes n_N}\right)\ket{\text{GHZ}_{1+N_{\text{tot}}}}_{A_1B_{1:n_1}\dotsb B_{1:n_N}},
\end{align}
as required.

Protocol~\ref{alg:multiple_GHZ_merge_protocol} has the following form as an LOCC channel:
\begin{multline}\label{eq-multiple_GHZ_merge_output}
    \mathcal{L}\!\left(\rho_{A_1B_{1:n_1}^1A_2B_{1:n_2}^2\dotsb A_NB_{1:n_N}^N}\right)=\sum_{\vec{y}\in\{0,1\}}\left(\bra{\vec{y}}_{A_{2:N}}\text{CNOT}_N\otimes (X^{y_1})_{B_{1:n_2}^2}^{\otimes n_2}\otimes\dotsb\otimes (X^{y_{N-1}})_{B_{1:n_N}^N}^{\otimes n_N}\right)\rho_{A_1B_{1:n_1}^1A_2B_{1:n_2}^2\dotsb A_NB_{1:n_N}^N}\\\times \left(\text{CNOT}_N\ket{\vec{y}}_{A_{2:N}}\otimes (X^{y_1})_{B_{1:n_2}^2}^{\otimes n_2}\otimes\dotsb\otimes (X^{y_{N-1}})_{B_{1:n_N}^N}^{\otimes n_N}\right),
\end{multline}
where $\rho_{A_1B_{1:n_1}^1A_2B_{1:n_2}^2\dotsb A_NB_{1:n_N}^N}$ is an arbitrary state, and we recall the definition of $\text{CNOT}_N$ from \eqref{eq-CNOT_star}.

\begin{proposition}[Fidelity after Protocol~\ref{alg:multiple_GHZ_merge_protocol}]\label{prop:GHZ_merge_star_fid}
    Consider states $\rho_{A_1B_{1:n_1}^1}^{1},\rho_{A_2B_{1:n_2}^2}^2,\dotsc,\rho_{A_NB_{1:n_N}^N}^{N}$, representing noisy input states to Protocol~\ref{alg:multiple_GHZ_merge_protocol}. With these input states, the fidelity of the state produced by Protocol~\ref{alg:multiple_GHZ_merge_protocol} with respect to the final, ideal GHZ state is
    \begin{align}
        &\bra{\text{GHZ}_{1+N_{\text{tot}}}}\mathcal{L}\!\left(\rho_{A_1B_{1:n_1}^1}^1\otimes\rho_{A_2B_{1:n_2}^2}^2\otimes\dotsb\rho_{ A_NB_{1:n_N}^N}^N\right)\ket{\text{GHZ}_{1+N_{\text{tot}}}}\nonumber\\
        &\quad=\sum_{\vec{z}\in\{0,1\}^{N-1}}\bra{\text{GHZ}_{n_1+1}^{(z',\vec{0})}}\rho_{A_1B_{1:n_1}^1}^1\ket{\text{GHZ}_{n_1+1}^{(z',\vec{0})}}\bra{\text{GHZ}_{n_2+1}^{(z_1,\vec{0})}}\rho_{A_2B_{1:n_2}^2}^2\ket{\text{GHZ}_{n_2+1}^{(z_1,\vec{0})}}\dotsb\bra{\text{GHZ}_{n_N+1}^{(z_{N-1},\vec{0})}}\rho_{A_NB_{1:n_N}^N}^N\ket{\text{GHZ}_{n_N+1}^{(z_{N-1},\vec{0})}},
    \end{align}
    $z'=z_1\oplus z_2\oplus\dotsb\oplus z_{N-1}$, and we have used the definition of the $n$-qubit GHZ basis in \eqref{eq-GHZ_basis_main}.
\end{proposition}

\begin{proof}
    We start by defining
    \begin{align}
        \ket{w(\vec{y})}&\coloneqq \left(\text{CNOT}_N\ket{\vec{y}}_{A_{2:N}}\otimes (X^{y_1})_{B_{1:n_2}^2}^{\otimes n_2}\otimes\dotsb\otimes (X^{y_{N-1}})_{B_{1:n_N}^N}^{\otimes n_N}\right)\ket{\text{GHZ}_{1+N_{\text{tot}}}}_{A_1B_{1:n_1}^1B_{1:n_2}^2\dotsb B_{1:n_N}^N}\\
        &=\frac{1}{\sqrt{2}}\sum_{x\in\{0,1\}}\text{CNOT}_N\ket{x}_{A_1}\ket{\vec{y}}_{A_{2:N}}\otimes\ket{x}_{B_{1:n_1}^1}^{\otimes n_1}\otimes (X^{y_1})_{B_{1:n_2}^2}^{\otimes n_2}\ket{x}_{B_{1:n_2}^2}^{\otimes n_2}\otimes\dotsb\otimes (X^{y_{N-1}})_{B_{1:n_N}^N}^{\otimes n_N}\ket{x}_{B_{1:n_N}^N}^{\otimes n_N}\\
        &=\frac{1}{\sqrt{2}}\sum_{x\in\{0,1\}}\ket{x}_{A_1}\ket{y_1\oplus x}_{A_2}\dotsb\ket{y_{N-1}\oplus x}_{A_N}\otimes\ket{x}_{B_{1:n_1}^1}^{\otimes n_1}\otimes \ket{y_1\oplus x}_{B_{1:n_2}^2}^{\otimes n_2}\otimes\dotsb\otimes\ket{y_{N-1}\oplus x}_{B_{1:n_N}}^{\otimes n_N}\\
        &=\frac{1}{\sqrt{2}}\sum_{x\in\{0,1\}}\ket{x}_{A_1}\ket{x}_{B_{1:n_1}}^{\otimes n_1}\otimes\ket{y_1\oplus x}_{A_2}\ket{y_1\oplus x}_{B_{1:n_2}^2}^{\otimes n_2}\otimes\dotsb\otimes\ket{y_{N-1}\oplus x}_{A_N}\ket{y_{N-1}\oplus x}_{B_{1:n_N}^N}^{\otimes n_N}.
    \end{align}
    Now, it is straightforward to show that $\ket{x}^{\otimes n}=\frac{1}{\sqrt{2}}\sum_{z\in\{0,1\}}(-1)^{zx}\ket{\text{GHZ}_n^{(z,\vec{0})}}$, for $x\in\{0,1\}$. With this, we find that
    \begin{align}
        \ket{w(\vec{y})}&=\frac{1}{\sqrt{2}}\sum_{x\in\{0,1\}}\left(\frac{1}{\sqrt{2}}\sum_{z_1\in\{0,1\}}(-1)^{z_1x}\ket{\text{GHZ}_{n_1+1}^{(z_1,\vec{0})}}_{A_1B_{1:n_1}^1}\right)\otimes\left(\frac{1}{\sqrt{2}}\sum_{z_2\in\{0,1\}}(-1)^{z_2(y_1\oplus x)}\ket{\text{GHZ}_{n_2+1}^{(z_2,\vec{0})}}_{A_2B_{1:n_2}^2}\right)\otimes\dotsb\nonumber\\
        &\qquad\qquad\qquad\qquad\dotsb\otimes\left(\frac{1}{\sqrt{2}}\sum_{z_N\in\{0,1\}}(-1)^{z_N(y_{N-1}\oplus x)}\ket{\text{GHZ}_{n_N+1}^{(z_N,\vec{0})}}_{A_NB_{1:n_N}^N}\right)\\
        &=\frac{1}{\sqrt{2^{N+1}}}\sum_{z_1,z_2,\dotsc,z_N\in\{0,1\}}(-1)^{z_2y_1+z_3y_2+\dotsb+z_Ny_{N-1}}\underbrace{\left(\sum_{x\in\{0,1\}}(-1)^{x(z_1\oplus z_2\oplus\dotsb\oplus z_N)}\right)}_{2\delta_{z_1,z_2\oplus\dotsb\oplus z_N}}\nonumber\\
        &\qquad\qquad\qquad\qquad\times\ket{\text{GHZ}_{n_1+1}^{(z_1,\vec{0})}}_{A_1B_{1:n_1}^1}\otimes\ket{\text{GHZ}_{n_2+1}^{(z_2,\vec{0})}}_{A_2B_{1:n_2}^2}\otimes\dotsb\otimes\ket{\text{GHZ}_{n_N+1}^{(z_N,\vec{0})}}_{A_NB_{1:n_N}^N}\\
        &=\frac{1}{\sqrt{2^{N-1}}}\sum_{z_2,z_3,\dotsc,z_N\in\{0,1\}}(-1)^{z_2y_1+z_3y_2+\dotsb+z_Ny_{N-1}}\ket{\text{GHZ}_{n_1+1}^{(z',\vec{0})}}_{A_1B_{1:n_1}^1}\otimes\ket{\text{GHZ}_{n_2+1}^{(z_2,\vec{0})}}_{A_2B_{1:n_2}^2}\otimes\dotsb\otimes\ket{\text{GHZ}_{n_N+1}^{(z_N,\vec{0})}}_{A_NB_{1:n_N}^N},
    \end{align}
    where in the final line we let $z'=z_2\oplus z_3\oplus\dotsb\oplus z_N$. Therefore, the fidelity of the state in \eqref{eq-multiple_GHZ_merge_output} is given by
    \begin{align}
        &\bra{\text{GHZ}_{1+N_{\text{tot}}}}\mathcal{L}\!\left(\rho_{A_1B_{1:n_1}^1A_2B_{1:n_2}^2\dotsb A_NB_{1:n_N}^N}\right)\ket{\text{GHZ}_{1+N_{\text{tot}}}}\nonumber\\
        &\quad=\sum_{\vec{y}\in\{0,1\}^{N-1}}\bra{w(\vec{y})}\rho_{A_1B_{1:n_1}^1A_2B_{1:n_2}^2\dotsb A_NB_{1:n_N}^N}\ket{w(\vec{y})}\\
        &\quad=\frac{1}{2^{N-1}}\sum_{\vec{y},\vec{z},\vec{w}\in\{0,1\}^{N-1}}(-1)^{\vec{z}\cdot\vec{y}}(-1)^{\vec{w}\cdot\vec{y}}\left(\bra{\text{GHZ}_{n_1+1}^{(z',\vec{0})}}_{A_1B_{1:n_1}^1}\otimes\bra{\text{GHZ}_{n_2+1}^{(z_1,\vec{0})}}_{A_2B_{1:n_2}^2}\otimes\dotsb\otimes\bra{\text{GHZ}_{n_N+1}^{(z_N,\vec{0})}}_{A_NB_{1:n_N}^N}\right)\nonumber\\
        &\qquad\qquad\times\rho_{A_1B_{1:n_1}^1A_2B_{1:n_2}^2\dotsb A_NB_{1:n_N}^N}\left(\ket{\text{GHZ}_{n_1+1}^{(w',\vec{0})}}_{A_1B_{1:n_1}^1}\otimes\ket{\text{GHZ}_{n_2+1}^{(w_1,\vec{0})}}_{A_2B_{1:n_2}^2}\otimes\dotsb\otimes\ket{\text{GHZ}_{n_N+1}^{(w_N,\vec{0})}}_{A_NB_{1:n_N}^N}\right)\\
        &\quad=\sum_{\vec{z}\in\{0,1\}^{N-1}}\left(\bra{\text{GHZ}_{n_1+1}^{(z',\vec{0})}}_{A_1B_{1:n_1}^1}\otimes\bra{\text{GHZ}_{n_2+1}^{(z_1,\vec{0})}}_{A_2B_{1:n_2}^2}\otimes\dotsb\otimes\bra{\text{GHZ}_{n_N+1}^{(z_N,\vec{0})}}_{A_NB_{1:n_N}^N}\right)\nonumber\\
        &\qquad\qquad\times\rho_{A_1B_{1:n_1}^1A_2B_{1:n_2}^2\dotsb A_NB_{1:n_N}^N}\left(\ket{\text{GHZ}_{n_1+1}^{(z',\vec{0})}}_{A_1B_{1:n_1}^1}\otimes\ket{\text{GHZ}_{n_2+1}^{(z_1,\vec{0})}}_{A_2B_{1:n_2}^2}\otimes\dotsb\otimes\ket{\text{GHZ}_{n_N+1}^{(z_N,\vec{0})}}_{A_NB_{1:n_N}^N}\right),
    \end{align}
    as required.
\end{proof}

\end{widetext}

\section*{Data and code availability}

The data and code for this project are available at \url{https://github.com/SidCh1/Graph-State-Generation-Project}.

\begin{acknowledgments}

We thank Frederik Hahn, Ernst Althaus, Luzie Marianczuk, Eleanor Reiffel, Kurt Mehlhorn, Lina Vandr\'{e}, Kenneth Goodenough and Annirudha Sen for useful discussions. We acknowledge funding from the BMBF/BMFTR in Germany (QR.X/QR.N, PhotonQ, QuKuK, QuaPhySI, HYBRID), the EU's Horizon Research and Innovation Actions (CLUSTEC), and from the Deutsche Forschungsgemeinschaft (DFG, German Research Foundation)--Project-ID 429529648 – TRR 306 QuCoLiMa (``Quantum Cooperativity of Light and Matter'').

\end{acknowledgments}

\section*{Author contributions}

SC, SK, and PvL jointly conceived the project. SC and SK conducted the analysis, and SC developed and executed the numerical simulations. All authors worked together to write the manuscript.

\section*{Competing interests}

The authors declare no competing interests.

\newpage


\bibliography{refs_main_npjQI}

@article{KVS+19,
	author = {Kuzmin, V. V. and Vasilyev, D. V. and Sangouard, N. and D\"{u}r, W. and Muschik, C. A.},
	title = {Scalable repeater architectures for multi-party states},
	journal = {npj Quantum Information},
	volume = {5},
	issue = {1},
	pages = {115},
	year = {2019},
	url = {https://doi.org/10.1038/s41534-019-0230-3}
}

@article{HPE19,
	title = {Quantum network routing and local complementation},
	author = {Hahn, F. and Pappa, A. and Eisert, J.},
	journal = {npj Quantum Information},
	volume = {5},
	issue = {1},
	pages = {76},
	year = {2019},
	url = {https://doi.org/10.1038/s41534-019-0191-6}
}

@article{MMG19,
  title = {Distributing graph states over arbitrary quantum networks},
  author = {Meignant, Cl\'ement and Markham, Damian and Grosshans, Fr\'ed\'eric},
  journal = {Physical Review A},
  volume = {100},
  issue = {5},
  pages = {052333},
  numpages = {6},
  year = {2019},
  month = {November},
  publisher = {American Physical Society},
  doi = {10.1103/PhysRevA.100.052333},
  url = {https://link.aps.org/doi/10.1103/PhysRevA.100.052333}
}

@article{CC12,
  title = {Growth of graph states in quantum networks},
  author = {Cuquet, Mart\'i and Calsamiglia, John},
  journal = {Physical Review A},
  volume = {86},
  issue = {4},
  pages = {042304},
  numpages = {15},
  year = {2012},
  month = {October},
  publisher = {American Physical Society},
  doi = {10.1103/PhysRevA.86.042304},
  url = {https://link.aps.org/doi/10.1103/PhysRevA.86.042304},
  eprint = {1208.0710},
  eprinttype = {arXiv}
}

@article{PWD18,
	author = {A. Pirker and J. Walln\"{o}fer and W. D\"{u}r},
	title = {Modular architectures for quantum networks},
	journal = {New Journal of Physics},
	doi = {10.1088/1367-2630/aac2aa},
	url = {https://doi.org/10.1088%2F1367-2630%2Faac2aa},
	year = {2018},
	month = {May},
	publisher = {{IOP} Publishing},
	volume = {20},
	number = {5},
	pages = {053054}
}

@article{PD19,
	author = {A. Pirker and W. D\"{u}r},
	title = {A quantum network stack and protocols for reliable entanglement-based networks},
	journal = {New Journal of Physics},
	doi = {10.1088/1367-2630/ab05f7},
	url = {https://doi.org/10.1088%2F1367-2630%2Fab05f7},
	year = {2019},
	month = {March},
	publisher = {{IOP} Publishing},
	volume = {21},
	number = {3},
	pages = {033003},
}

@article{EKB16,
  author={Michael Epping and Hermann Kampermann and Dagmar Bru\ss},
  title={Large-scale quantum networks based on graphs},
  journal={New Journal of Physics},
  volume={18},
  number={5},
  pages={053036},
  url={https://doi.org/10.1088/1367-2630/18/5/053036},
  year={2016}
}

@article{WZM+16,
  title = {Two-dimensional quantum repeaters},
  author = {Walln\"ofer, J. and Zwerger, M. and Muschik, C. and Sangouard, N. and D\"ur, W.},
  journal = {Physical Review A},
  volume = {94},
  issue = {5},
  pages = {052307},
  numpages = {12},
  year = {2016},
  month = {November},
  publisher = {American Physical Society},
  doi = {10.1103/PhysRevA.94.052307},
  url = {https://link.aps.org/doi/10.1103/PhysRevA.94.052307}
}

@book{Barabasi16_book,
  title={{Network Science}},
  author={Barab\'{a}si, Albert-L\'{a}szl\'{o}},
  year={2016},
  publisher={Cambridge University Press},
  url = {http://networksciencebook.com/}
}

@article{BBP96,
  title = {Purification of Noisy Entanglement and Faithful Teleportation via Noisy Channels},
  author = {Bennett, Charles H. and Brassard, Gilles and Popescu, Sandu and Schumacher, Benjamin and Smolin, John A. and Wootters, William K.},
  journal = {Physical Review Letters},
  volume = {76},
  issue = {5},
  pages = {722--725},
  numpages = {0},
  year = {1996},
  month = {January},
  publisher = {American Physical Society},
  doi = {10.1103/PhysRevLett.76.722},
  url = {https://link.aps.org/doi/10.1103/PhysRevLett.76.722}
}

@article{KMSD19,
  title = {Practical figures of merit and thresholds for entanglement distribution in quantum networks},
  author = {Khatri, Sumeet and Matyas, Corey T. and Siddiqui, Aliza U. and Dowling, Jonathan P.},
  journal = {Physical Review Research},
  volume = {1},
  issue = {2},
  pages = {023032},
  numpages = {16},
  year = {2019},
  month = {September},
  publisher = {American Physical Society},
  doi = {10.1103/PhysRevResearch.1.023032},
  url = {https://link.aps.org/doi/10.1103/PhysRevResearch.1.023032}
}

@article{SSR+11,
  title = {Quantum repeaters based on atomic ensembles and linear optics},
  author = {Sangouard, Nicolas and Simon, Christoph and de Riedmatten, Hugues and Gisin, Nicolas},
  journal = {Reviews of Modern Physics},
  volume = {83},
  issue = {1},
  pages = {33--80},
  numpages = {0},
  year = {2011},
  month = {March},
  publisher = {American Physical Society},
  doi = {10.1103/RevModPhys.83.33},
  url = {https://link.aps.org/doi/10.1103/RevModPhys.83.33}
}

@article{BDC98,
  title = {Quantum Repeaters: The Role of Imperfect Local Operations in Quantum Communication},
  author = {Briegel, H.-J. and D\"ur, W. and Cirac, J. I. and Zoller, P.},
  journal = {Physical Review Letters},
  volume = {81},
  issue = {26},
  pages = {5932--5935},
  numpages = {0},
  year = {1998},
  month = {December},
  publisher = {American Physical Society},
  doi = {10.1103/PhysRevLett.81.5932},
  url = {https://link.aps.org/doi/10.1103/PhysRevLett.81.5932}
}

@article{DBC99,
  title = {Quantum repeaters based on entanglement purification},
  author = {D\"ur, W. and Briegel, H.-J. and Cirac, J. I. and Zoller, P.},
  journal = {Physical Review A},
  volume = {59},
  issue = {1},
  pages = {169--181},
  numpages = {0},
  year = {1999},
  month = {January},
  publisher = {American Physical Society},
  doi = {10.1103/PhysRevA.59.169},
  url = {https://link.aps.org/doi/10.1103/PhysRevA.59.169}
}

@Inbook{Brazil2015,
author={Brazil, Marcus and Zachariasen, Martin},
title={{Steiner Trees in Graphs and Hypergraphs}},
bookTitle={{Optimal Interconnection Trees in the Plane: Theory, Algorithms and Applications}},
year={2015},
publisher={Springer International Publishing},
pages={301--317},
isbn={978-3-319-13915-9},
doi={10.1007/978-3-319-13915-9_5},
url={https://doi.org/10.1007/978-3-319-13915-9_5}
}

@article{DKD18,
  title = {Robust quantum network architectures and topologies for entanglement distribution},
  author = {Das, Siddhartha and Khatri, Sumeet and Dowling, Jonathan P.},
  journal = {Physical Review A},
  volume = {97},
  issue = {1},
  pages = {012335},
  numpages = {12},
  year = {2018},
  month = {January},
  publisher = {American Physical Society},
  doi = {10.1103/PhysRevA.97.012335},
  url = {https://link.aps.org/doi/10.1103/PhysRevA.97.012335}
}

@article{ghosh2012surfacecodedecoherence,
  title = {{Surface code with decoherence: An analysis of three superconducting architectures}},
  author = {Ghosh, Joydip and Fowler, Austin G. and Geller, Michael R.},
  year = 2012,
  month = dec,
  journal = {Physical Review A},
  publisher = {American Physical Society},
  volume = 86,
  pages = {062318},
  doi = {10.1103/PhysRevA.86.062318},
  url = {https://link.aps.org/doi/10.1103/PhysRevA.86.062318},
  issue = 6,
  numpages = 12,
  eprint = {1210.5799},
  eprinttype = {arXiv}
}

@article{knaut2024entanglementnetwork,
	title = {{Entanglement of Nanophotonic Quantum Memory Nodes in a Telecom Network}},
	author = {Can M. Knaut and Aziza Suleymanzade and Yan-Cheng Wei and Daniel R. Assumpcao and Pieter-Jan Stas and Yan Qi Huan and Bartholomeus Machielse and Erik N. Knall and Madison Sutula and Gefen Baranes and Neil Sinclair and Chawina De-Eknamkul and David S. Levonian and Mihir K. Bhaskar and Hongkun Park and Marko Lon\v{c}ar and Mikhail D. Lukin},
	journal = {Nature},
	volume = {629},
	issue = {8012},
	pages = {573--578},
	year = {2024},
	url = {https://doi.org/10.1038/s41586-024-07252-z},
	eprint = {2310.01316},
	eprinttype = {arXiv}
}

@Article{BBC+93,
  author    = {Bennett, Charles H. and Brassard, Gilles and Cr\'epeau, Claude and Jozsa, Richard and Peres, Asher and Wootters, William K.},
  title     = {{Teleporting an unknown quantum state via dual classical and Einstein-Podolsky-Rosen channels}},
  journal   = {Physical Review Letters},
  year      = {1993},
  volume    = {70},
  pages     = {1895--1899},
  month     = {March},
  doi       = {10.1103/PhysRevLett.70.1895},
  issue     = {13},
  numpages  = {0},
  publisher = {American Physical Society},
  url       = {https://link.aps.org/doi/10.1103/PhysRevLett.70.1895},
}

@article{ZZS18,
  title = {Distributed quantum sensing using continuous-variable multipartite entanglement},
  author = {Zhuang, Quntao and Zhang, Zheshen and Shapiro, Jeffrey H.},
  journal = {Physical Review A},
  volume = {97},
  issue = {3},
  pages = {032329},
  numpages = {6},
  year = {2018},
  month = {March},
  publisher = {American Physical Society},
  doi = {10.1103/PhysRevA.97.032329},
  url = {https://link.aps.org/doi/10.1103/PhysRevA.97.032329}
}

@article{XZCZ19,
  title = {Repeater-enhanced distributed quantum sensing based on continuous-variable multipartite entanglement},
  author = {Xia, Yi and Zhuang, Quntao and Clark, William and Zhang, Zheshen},
  journal = {Physical Review A},
  volume = {99},
  issue = {1},
  pages = {012328},
  numpages = {9},
  year = {2019},
  month = {January},
  publisher = {American Physical Society},
  doi = {10.1103/PhysRevA.99.012328},
  url = {https://link.aps.org/doi/10.1103/PhysRevA.99.012328}
}

@article{guo2020distributedsensing,
	title = {Distributed quantum sensing in a continuous-variable entangled network},
	author = {Guo, Xueshi and Breum, Casper R. and Borregaard, Johannes and Izumi, Shuro and Larsen, Mikkel V. and Gehring, Tobias and Christandl, Matthias and Neergaard-Nielsen, Jonas S. and Andersen, Ulrik L.},
	journal = {Nature Physics},
	volume = {16},
	issue = {3},
	pages = {281--284},
	year = {2020},
	url = {https://doi.org/10.1038/s41567-019-0743-x},
	eprint = {1905.09408},
	eprinttype = {arXiv}
}

@article{Kha21b,
  title = {Policies for elementary links in a quantum network},
  author = {Khatri, Sumeet},
  journal = {{Quantum}},
  doi = {10.22331/q-2021-09-07-537},
  url = {https://doi.org/10.22331/q-2021-09-07-537},
  issn = {2521-327X},
  publisher = {{Verein zur F{\"{o}}rderung des Open Access Publizierens in den Quantenwissenschaften}},
  volume = {5},
  pages = {537},
  month = {September},
  year = {2021}
}

@article{Hump+18,
  title={Deterministic delivery of remote entanglement on a quantum network},
  author={Humphreys, Peter C. and Kalb, Norbert and Morits, Jaco P. J. and Schouten, Raymond N. and Vermeulen, Raymond F. L. and Twitchen, Daniel J. and Markham, Matthew and Hanson, Ronald},
  journal={Nature},
  volume={558},
  number={7709},
  pages={268},
  year={2018},
  url={https://doi.org/10.1038/s41586-018-0200-5},
  publisher={Nature Publishing Group}
}

@article{DSG+16,
	title = {Generation of heralded entanglement between distant hole spins},
	author = {Delteil, Aymeric and Sun, Zhe and Gao, Wei-bo and Togan, Emre and Faelt, Stefan and Imamo\v{g}lu, Ata\c{c}},
	journal = {Nature Physics},
	volume = {12},
	issue = {3},
	pages = {218--223},
	year = {2016},
	doi = {10.1038/nphys3605},
	url = {https://doi.org/10.1038/nphys3605}
}

@article{SSH+17,
  title = {Phase-Tuned Entangled State Generation between Distant Spin Qubits},
  author = {Stockill, R. and Stanley, M. J. and Huthmacher, L. and Clarke, E. and Hugues, M. and Miller, A. J. and Matthiesen, C. and Le Gall, C. and Atat\"ure, M.},
  journal = {Physical Review Letters},
  volume = {119},
  issue = {1},
  pages = {010503},
  numpages = {6},
  year = {2017},
  month = {July},
  publisher = {American Physical Society},
  doi = {10.1103/PhysRevLett.119.010503},
  url = {https://link.aps.org/doi/10.1103/PhysRevLett.119.010503}
}

@article{PHB+21,
	title = {Realization of a multi-node quantum network of remote solid-state qubits},
	author = {Matteo Pompili and Sophie L. N. Hermans and Simon Baier and Hans K. C. Beukers and Peter C. Humphreys and Raymond N. Schouten and Raymond F. L. Vermeulen and Marijn J. Tiggelman and Laura dos Santos Martins and Bas Dirkse and Stephanie Wehner and Ronald Hanson},
	journal = {Science},
	year = {2021},
	volume = {372},
	issue = {6539},
	pages = {259--264},
	url = {https://doi.org/10.1126/science.abg1919}
}

@article{hermans2022teleportation,
  title = {{Qubit teleportation between non-neighbouring nodes in a quantum network}},
  author = {Hermans, S. L. N. and Pompili, M. and Beukers, H. K. C. and Baier, S. and Borregaard, J. and Hanson, R.},
  year = 2022,
  journal = {Nature},
  volume = 605,
  pages = {663--668},
  url = {https://doi.org/10.1038/s41586-022-04697-y},
  issue = 7911
}

@article{pompili2022demonstrationstack,
  title = {{Experimental demonstration of entanglement delivery using a quantum network stack}},
  author = {Pompili, M. and Delle Donne, C. and te Raa, I. and van der Vecht, B. and Skrzypczyk, M. and Ferreira, G. and de Kluijver, L. and Stolk, A. J. and Hermans, S. L. N. and Pawe\l{}czak, P. and Kozlowski, W. and Hanson, R. and Wehner, S.},
  year = 2022,
  journal = {npj Quantum Information},
  volume = 8,
  pages = 121,
  url = {https://doi.org/10.1038/s41534-022-00631-2},
  issue = 1
}

@article{RB01,
  title = {{A One-Way Quantum Computer}},
  author = {Raussendorf, Robert and Briegel, Hans J.},
  journal = {Physical Review Letters},
  volume = {86},
  issue = {22},
  pages = {5188--5191},
  numpages = {0},
  year = {2001},
  month = {May},
  publisher = {American Physical Society},
  doi = {10.1103/PhysRevLett.86.5188},
  url = {https://link.aps.org/doi/10.1103/PhysRevLett.86.5188}
}

@Inbook{GHZ89,
author={Greenberger, Daniel M. and Horne, Michael A. and Zeilinger, Anton},
editor={Kafatos, Menas},
title={{Going Beyond Bell's Theorem}},
bookTitle={{Bell's Theorem, Quantum Theory and Conceptions of the Universe}},
year={1989},
publisher={Springer Netherlands},
address={Dordrecht},
pages={69--72},
doi={10.1007/978-94-017-0849-4_10},
url={https://doi.org/10.1007/978-94-017-0849-4_10}
}

@article{ABC21,
author = {Azuma,Koji  and B\"{a}uml,Stefan  and Coopmans,Tim  and Elkouss,David  and Li,Boxi },
title = {Tools for quantum network design},
journal = {AVS Quantum Science},
volume = {3},
number = {1},
pages = {014101},
year = {2021},
doi = {10.1116/5.0024062},
URL = {https://doi.org/10.1116/5.0024062},

}

@article{brito2020statistical,
  title = {{Statistical Properties of the Quantum Internet}},
  author = {Brito, Samura\'{\i} and Canabarro, Askery and Chaves, Rafael and Cavalcanti, Daniel},
  journal = {Physical Review Letters},
  volume = {124},
  issue = {21},
  pages = {210501},
  numpages = {6},
  year = {2020},
  month = {May},
  publisher = {American Physical Society},
  doi = {10.1103/PhysRevLett.124.210501},
  url = {https://link.aps.org/doi/10.1103/PhysRevLett.124.210501}
}

@INPROCEEDINGS{fischer2021graph,
  author={Fischer, Alex and Towsley, Don},
  booktitle={{2021 IEEE International Conference on Quantum Computing and Engineering (QCE)}}, 
  title={{Distributing Graph States Across Quantum Networks}}, 
  year={2021},
  volume={},
  number={},
  pages={324--333},
  doi={10.1109/QCE52317.2021.00049},
  eprint = {2009.10888},
  eprinttype = {arXiv}
  }

@book{Bollobas2001_book,
	author={Bollob\'{a}s, B\'{e}la},
	year={2001},
	location={Cambridge},
	edition={2},
	series={{Cambridge Studies in Advanced Mathematics}},
	title={{Random Graphs}},
	DOI={10.1017/CBO9780511814068},
	publisher={Cambridge University Press}
}

@book{Newman2018_book,
    author = {Newman, Mark},
    title = {{Networks}},
    publisher = {Oxford University Press},
    year = {2018},
    month = {07},
    isbn = {9780198805090},
    url = {https://doi.org/10.1093/oso/9780198805090.001.0001}
}

@article{newman2001randomgraphs,
  title = {Random graphs with arbitrary degree distributions and their applications},
  author = {Newman, M. E. J. and Strogatz, S. H. and Watts, D. J.},
  journal = {Physical Review E},
  volume = {64},
  issue = {2},
  pages = {026118},
  numpages = {17},
  year = {2001},
  month = {July},
  publisher = {American Physical Society},
  doi = {10.1103/PhysRevE.64.026118},
  url = {https://link.aps.org/doi/10.1103/PhysRevE.64.026118}
}

@article{erdosrenyi1959,
  title={{On random graphs. I.}},
  author={Paul L. Erd\H{o}s and Alfr{\'e}d R{\'e}nyi},
  journal={Publicationes Mathematicae Debrecen},
  volume = {6},
  issue = {3},
  pages = {290--297},
  year={1959},
  url = {https://www.renyi.hu/~p_erdos/1959-11.pdf},
  month = {11}
}

@article{erdosrenyi1960,
	author = {Erd\H{o}s, P. L. and R\'{e}nyi, A.},
	title = {On the evolution of random graphs},
	journal = {Magyar Tudom\'{a}nyos Akad\'{e}mia Matematikai Kutat\'{o} Int\'{e}zet\'{e}nek K\H{o}zlem\'{e}nyei},
	volume = {5},
	pages = {17--61},
	year = {1960},
	url = {https://www.renyi.hu/~p_erdos/1960-10.pdf}
}

@article{barabasi1999scalefree,
author = {Albert-L\'{a}szl\'{o} Barab\'{a}si  and R\'{e}ka Albert},
title = {{Emergence of Scaling in Random Networks}},
journal = {Science},
volume = {286},
number = {5439},
pages = {509--512},
year = {1999},
url = {https://doi.org/10.1126/science.286.5439.509},
eprint = {cond-mat/9910332},
eprinttype = {arXiv}
}

@article{bugalho2023multipartite,
  doi = {10.22331/q-2023-02-09-920},
  url = {https://doi.org/10.22331/q-2023-02-09-920},
  title = {{Distributing Multipartite Entanglement over Noisy Quantum Networks}},
  author = {Bugalho, Lu{\'{i}}s and Coutinho, Bruno C. and Monteiro, Francisco A. and Omar, Yasser},
  journal = {{Quantum}},
  issn = {2521-327X},
  publisher = {{Verein zur F{\"{o}}rderung des Open Access Publizierens in den Quantenwissenschaften}},
  volume = {7},
  pages = {920},
  month = feb,
  year = {2023}
}

@article{avis2023multipartitecentral,
  title = {Analysis of multipartite entanglement distribution using a central quantum-network node},
  author = {Avis, Guus and Rozp\k{e}dek, Filip and Wehner, Stephanie},
  journal = {Physical Review A},
  volume = {107},
  issue = {1},
  pages = {012609},
  numpages = {36},
  year = {2023},
  month = {Jan},
  publisher = {American Physical Society},
  doi = {10.1103/PhysRevA.107.012609},
  url = {https://doi.org/10.1103/PhysRevA.107.012609},
  eprint = {2203.05517},
  eprinttype = {arXiv}
}

@article{khatri2022networkMDP,
    author = {Khatri, Sumeet},
    title = {{On the design and analysis of near-term quantum network protocols using Markov decision processes}},
    journal = {AVS Quantum Science},
    volume = {4},
    number = {3},
    pages = {030501},
    year = {2022},
    month = {09},
    issn = {2639-0213},
    doi = {10.1116/5.0084653},
    url = {https://doi.org/10.1116/5.0084653},
    eprint = {2207.03403},
    eprinttype = {arXiv}
}

@article{sen2023multipartitesubgraphcomplementation,
	title = {{Multipartite Entanglement in Quantum Networks using Subgraph Complementations}},
	author = {Aniruddha Sen and Kenneth Goodenough and Don Towsley},
	journal = {arXiv:2308.13700},
	year = {2023},
	url = {https://arxiv.org/abs/2308.13700}
}

@ARTICLE{deBone2020GHZdistillBell,
  author={de Bone, Sebastian and Ouyang, Runsheng and Goodenough, Kenneth and Elkouss, David},
  journal={IEEE Transactions on Quantum Engineering}, 
  title={{Protocols for Creating and Distilling Multipartite GHZ States With Bell Pairs}}, 
  year={2020},
  volume={1},
  pages={1--10},
  doi={10.1109/TQE.2020.3044179},
  url = {https://doi.org/10.1109/TQE.2020.3044179},
  eprint = {2010.12259},
  eprinttype = {arXiv}
}

@article{sadhu2023practical,
	title = {{Practical limitations on robustness and scalability of quantum Internet}},
	author = {Abhishek Sadhu and Meghana Ayyala Somayajula and Karol Horodecki and Siddhartha Das},
	journal = {arXiv:2308.12739},
	year = {2023},
	url = {https://arxiv.org/abs/2308.12739}
}

@article{albert1999diameterWWW,
	title = {{Diameter of the World-Wide Web}},
	author = {Albert, R\'{e}ka and Jeong, Hawoong and Barab\'{a}si, Albert-L\'{a}szl\'{o}},
	journal = {Nature},
	volume = {401},
	issue = {6749},
	pages = {130--131},
	year = {1999},
	url = {https://doi.org/10.1038/43601},
	eprint = {cond-mat/9907038},
	eprinttype = {arXiv}
}

@article{wattsstrogatz1998,
	title = {Collective dynamics of `small-world' networks},
	author = {Watts, Duncan J. and Strogatz, Steven H.},
	journal = {Nature},
	volume = {393},
	issue = {6684},
	pages = {440--442},
	year = {1998},
	url = {https://doi.org/10.1038/30918}
}

@ARTICLE{waxman1988,
  author={Waxman, B.M.},
  journal={IEEE Journal on Selected Areas in Communications}, 
  title={Routing of multipoint connections}, 
  year={1988},
  volume={6},
  number={9},
  pages={1617--1622},
  doi={10.1109/49.12889},
  url = {https://doi.org/10.1109/49.12889}
}

@article{albert2000attacktolerancenetworks,
	title = {Error and attack tolerance of complex networks},
	author = {Albert, R\'{e}ka and Jeong, Hawoong and Barab\'{a}si, Albert-L\'{a}szl\'{o}},
	journal = {Nature},
	volume = {406},
	issue = {6794},
	pages = {378--382},
	year = {2000},
	url = {https://doi.org/10.1038/35019019},
	eprint = {cond-mat/0008064},
	eprinttype = {arXiv}
}

@article{coutinho2022robustness,
	title = {Robustness of noisy quantum networks},
	author = {Coutinho, Bruno Coelho and Munro, William John and Nemoto, Kae and Omar, Yasser},
	journal = {Communications Physics},
	volume = {5},
	issue = {1},
	pages = {105},
	year = {2022},
	url = {https://doi.org/10.1038/s42005-022-00866-7},
	eprint = {2103.03266},
	eprinttype = {arXiv}
}

@article{roga2023dickestatedistribution,
  title = {{Efficient Dicke-state distribution in a network of lossy channels}},
  author = {Roga, Wojciech and Ikuta, Rikizo and Horikiri, Tomoyuki and Takeoka, Masahiro},
  journal = {Physical Review A},
  volume = {108},
  issue = {1},
  pages = {012612},
  numpages = {10},
  year = {2023},
  month = {Jul},
  publisher = {American Physical Society},
  doi = {10.1103/PhysRevA.108.012612},
  url = {https://doi.org/10.1103/PhysRevA.108.012612},
  eprint = {2211.15138},
  eprinttype = {arXiv}
}

@article{shimizu2024GHZdistribution,
	title = {{Simple loss-tolerant protocol for GHZ-state distribution in a quantum network}},
	author = {Hikaru Shimizu and Wojciech Roga and David Elkouss and Masahiro Takeoka},
	journal = {arXiv:2404.19458},
	year = {2024},
	url = {https://arxiv.org/abs/2404.19458}
}

@article{kruszynska2006purificationgraphstates,
  title = {Entanglement purification protocols for all graph states},
  author = {Kruszynska, Caroline and Miyake, Akimasa and Briegel, Hans J. and D\"ur, Wolfgang},
  journal = {Physical Review A},
  volume = {74},
  issue = {5},
  pages = {052316},
  numpages = {9},
  year = {2006},
  month = {Nov},
  publisher = {American Physical Society},
  doi = {10.1103/PhysRevA.74.052316},
  url = {https://doi.org/10.1103/PhysRevA.74.052316},
  eprint = {quant-ph/0606090},
  eprinttype = {arXiv}
}

@article{haldar2023fastreliable,
  title = {{Fast and reliable entanglement distribution with quantum repeaters: Principles for improving protocols using reinforcement learning}},
  author = {Haldar, Stav and Barge, Pratik J. and Khatri, Sumeet and Lee, Hwang},
  year = 2024,
  month = feb,
  journal = {Physical Review Applied},
  publisher = {American Physical Society},
  volume = 21,
  pages = {024041},
  doi = {10.1103/PhysRevApplied.21.024041},
  url = {https://doi.org/10.1103/PhysRevApplied.21.024041},
  issue = 2,
  numpages = 34,
  eprint = {2303.00777},
  eprinttype = {arXiv}
}

@article{azuma2023repeatersRMP,
  title = {{Quantum repeaters: From quantum networks to the quantum internet}},
  author = {Azuma, Koji and Economou, Sophia E. and Elkouss, David and Hilaire, Paul and Jiang, Liang and Lo, Hoi-Kwong and Tzitrin, Ilan},
  year = 2023,
  month = dec,
  journal = {Reviews of Modern Physics},
  publisher = {American Physical Society},
  volume = 95,
  pages = {045006},
  doi = {10.1103/RevModPhys.95.045006},
  url = {https://doi.org/10.1103/RevModPhys.95.045006},
  issue = 4,
  numpages = 66,
  eprint = {2212.10820},
  eprinttype = {arXiv}
}

@article{awschalom2021interconnects,
  title = {{Development of Quantum Interconnects (QuICs) for Next-Generation Information Technologies}},
  author = {Awschalom, David and Berggren, Karl K. and Bernien, Hannes and Bhave, Sunil and Carr, Lincoln D. and Davids, Paul and Economou, Sophia E. and Englund, Dirk and Faraon, Andrei and Fejer, Martin and Guha, Saikat and Gustafsson, Martin V. and Hu, Evelyn and Jiang, Liang and Kim, Jungsang and Korzh, Boris and Kumar, Prem and Kwiat, Paul G. and Lon\ifmmode \check{c}\else \v{c}\fi{}ar, Marko and Lukin, Mikhail D. and Miller, David A.B. and Monroe, Christopher and Nam, Sae Woo and Narang, Prineha and Orcutt, Jason S. and Raymer, Michael G. and Safavi-Naeini, Amir H. and Spiropulu, Maria and Srinivasan, Kartik and Sun, Shuo and Vu\ifmmode \check{c}\else \v{c}\fi{}kovi\ifmmode \acute{c}\else \'{c}\fi{}, Jelena and Waks, Edo and Walsworth, Ronald and Weiner, Andrew M. and Zhang, Zheshen},
  year = 2021,
  month = feb,
  journal = {PRX Quantum},
  publisher = {American Physical Society},
  volume = 2,
  pages = {017002},
  doi = {10.1103/PRXQuantum.2.017002},
  url = {https://doi.org/10.1103/PRXQuantum.2.017002},
  issue = 1,
  numpages = 21
}

@article{gambetta2020ibm,
  title={{IBM’s roadmap for scaling quantum technology}},
  author={Gambetta, Jay},
  journal={IBM Research Blog},
  url = {https://www.ibm.com/quantum/blog/ibm-quantum-roadmap?mhsrc=ibmsearch_a&mhq=condor},
  year={2020},
  month = {Sep}
}

@article{pirandola2018sensingadvances,
	title = {Advances in photonic quantum sensing},
	author = {Pirandola, S. and Bardhan, B. R. and Gehring, T. and Weedbrook, C. and Lloyd, S.},
	journal = {Nature Photonics},
	volume = {12},
	issue = {12},
	pages = {724--733},
	year = {2018},
	url = {https://doi.org/10.1038/s41566-018-0301-6},
	eprint = {1811.01969},
	eprinttype = {arXiv}
}

@ARTICLE{humble2021quantumcomputersHPC,
  author={Humble, Travis S. and McCaskey, Alexander and Lyakh, Dmitry I. and Gowrishankar, Meenambika and Frisch, Albert and Monz, Thomas},
  journal={IEEE Micro}, 
  title={{Quantum Computers for High-Performance Computing}}, 
  year={2021},
  volume={41},
  number={5},
  pages={15--23},
  doi={10.1109/MM.2021.3099140},
  url = {https://doi.org/10.1109/MM.2021.3099140}
}

@article{eisert2000nonlocalgates,
  title = {Optimal local implementation of nonlocal quantum gates},
  author = {Eisert, J. and Jacobs, K. and Papadopoulos, P. and Plenio, M. B.},
  journal = {Physical Review A},
  volume = {62},
  issue = {5},
  pages = {052317},
  numpages = {7},
  year = {2000},
  month = {Oct},
  publisher = {American Physical Society},
  doi = {10.1103/PhysRevA.62.052317},
  url = {https://doi.org/10.1103/PhysRevA.62.052317},
  eprint = {quant-ph/0005101},
  eprinttype = {arXiv}
}

@article{tse2019quantumenhancedLIGO,
  title = {{Quantum-Enhanced Advanced LIGO Detectors in the Era of Gravitational-Wave Astronomy}},
  author = {Tse, M. and others},
  journal = {Physical Review Letters},
  volume = {123},
  issue = {23},
  pages = {231107},
  numpages = {8},
  year = {2019},
  month = {Dec},
  publisher = {American Physical Society},
  doi = {10.1103/PhysRevLett.123.231107},
  url = {https://doi.org/10.1103/PhysRevLett.123.231107}
}

@article{luo2023progressphotonicschips,
	title = {Recent progress in quantum photonic chips for quantum communication and internet},
	author = {Luo, Wei and Cao, Lin and Shi, Yuzhi and Wan, Lingxiao and Zhang, Hui and Li, Shuyi and Chen, Guanyu and Li, Yuan and Li, Sijin and Wang, Yunxiang and Sun, Shihai and Karim, Muhammad Faeyz and Cai, Hong and Kwek, Leong Chuan and Liu, Ai Qun},
	journal = {Light: Science \& Applications},
	volume = {12},
	issue = {1},
	pages = {175},
	year = {2023},
	url = {https://doi.org/10.1038/s41377-023-01173-8}
}

@article{gottesman1999gateteleportation,
	title = {Demonstrating the viability of universal quantum computation using teleportation and single-qubit operations},
	author = {Gottesman, Daniel and Chuang, Isaac L.},
	journal = {Nature},
	volume = {402},
	issue = {6760},
	pages = {390--393},
	year = {1999},
	url = {https://doi.org/10.1038/46503},
	eprint = {quant-ph/9908010},
	eprinttype = {arXiv}
}

@article{nielsen2003MBQC,
author = {Michael A. Nielsen},
title = {Quantum computation by measurement and quantum memory},
journal = {Physics Letters A},
volume = {308},
number = {2},
pages = {96--100},
year = {2003},
issn = {0375-9601},
url = {https://doi.org/10.1016/S0375-9601(02)01803-0},
eprint = {quant-ph/0108020},
eprinttype = {arXiv}
}

@article{leung2004MBQC,
author = {Leung, Debbie W.},
title = {{Quantum Computation by Measurements}},
journal = {International Journal of Quantum Information},
volume = {02},
number = {01},
pages = {33--43},
year = {2004},
doi = {10.1142/S0219749904000055},
URL = {https://doi.org/10.1142/S0219749904000055},
eprint = {quant-ph/0310189},
eprinttype = {arXiv}
}

@article{leung2001MBQC,
	title = {{Two-qubit Projective Measurements are Universal for Quantum Computation}},
	author = {Debbie W. Leung},
	journal = {arXiv:quant-ph/0111122},
	year = {2001},
	url = {https://arxiv.org/abs/quant-ph/0111122}
}

@article{josza2005introMBQC,
	title = {An introduction to measurement based quantum computation},
	author = {Richard Jozsa},
	journal = {arXiv:quant-ph/0508124},
	year = {2005},
	url = {https://arxiv.org/abs/quant-ph/0508124}
}

@ARTICLE{piveteau2024circuitknitting,
  author={Piveteau, Christophe and Sutter, David},
  journal={IEEE Transactions on Information Theory}, 
  title={{Circuit Knitting With Classical Communication}}, 
  year={2024},
  volume={70},
  number={4},
  pages={2734--2745},
  doi={10.1109/TIT.2023.3310797},
  url = {https://doi.org/10.1109/TIT.2023.3310797},
  eprint = {2205.00016},
  eprinttype = {arXiv}
}

@INPROCEEDINGS{broadbent2009blindQC,
  author={Broadbent, Anne and Fitzsimons, Joseph and Kashefi, Elham},
  booktitle={{2009 50th Annual IEEE Symposium on Foundations of Computer Science (FOCS 2009)}}, 
  title={{Universal Blind Quantum Computation}}, 
  year={2009},
  pages={517-526},
  doi={10.1109/FOCS.2009.36},
  url = {https://doi.org/10.1109/FOCS.2009.36},
  eprint = {0807.4154},
  eprinttype = {arXiv}
}

@article{danos2007distributedMBQC,
author = {Vincent Danos and Ellie D'Hondt and Elham Kashefi and Prakash Panangaden},
title = {{Distributed Measurement-based Quantum Computation}},
journal = {Electronic Notes in Theoretical Computer Science},
volume = {170},
pages = {73--94},
year = {2007},
note = {{Proceedings of the 3rd International Workshop on Quantum Programming Languages (QPL 2005)}},
issn = {1571-0661},
url = {https://doi.org/10.1016/j.entcs.2006.12.012},
eprint = {quant-ph/0506070},
eprinttype = {arXiv}
}

@article{toth2012multipartitemetrology,
  title = {Multipartite entanglement and high-precision metrology},
  author = {T\'oth, G\'eza},
  journal = {Physical Review A},
  volume = {85},
  issue = {2},
  pages = {022322},
  numpages = {8},
  year = {2012},
  month = {Feb},
  publisher = {American Physical Society},
  doi = {10.1103/PhysRevA.85.022322},
  url = {https://doi.org/10.1103/PhysRevA.85.022322},
  eprint = {1006.4368},
  eprinttype = {arXiv}
}

@article{hyllus2012multipartitemetrology,
  title = {Fisher information and multiparticle entanglement},
  author = {Hyllus, Philipp and Laskowski, Wies\l{}aw and Krischek, Roland and Schwemmer, Christian and Wieczorek, Witlef and Weinfurter, Harald and Pezz\'e, Luca and Smerzi, Augusto},
  journal = {Physical Review A},
  volume = {85},
  issue = {2},
  pages = {022321},
  numpages = {10},
  year = {2012},
  month = {Feb},
  publisher = {American Physical Society},
  doi = {10.1103/PhysRevA.85.022321},
  url = {https://doi.org/10.1103/PhysRevA.85.022321},
  eprint = {1006.4366},
  eprinttype = {arXiv}
}

@article{murta2020CKAreview,
author = {Murta, Gl\'{a}ucia and Grasselli, Federico and Kampermann, Hermann and Bru{\ss}, Dagmar},
title = {{Quantum Conference Key Agreement: A Review}},
journal = {Advanced Quantum Technologies},
volume = {3},
number = {11},
pages = {2000025},
url = {https://doi.org/10.1002/qute.202000025},
year = {2020},
eprint = {2003.10186},
eprintype = {arXiv}
}

@InProceedings{networkx,
  author =       {Aric A. Hagberg and Daniel A. Schult and Pieter J. Swart},
  title =        {{Exploring Network Structure, Dynamics, and Function using NetworkX}},
  booktitle =   {{Proceedings of the 7th Python in Science Conference}},
  pages =     {11--15},
  address = {Pasadena, CA USA},
  year =      {2008},
  editor =    {Ga\"{e}l Varoquaux and Travis Vaught and Jarrod Millman},
  url = {https://www.osti.gov/biblio/960616},
  note = {\url{https://networkx.org/}}
}

@book{penrose2003_RGG_book,
    author = {Penrose, Mathew},
    title = {{Random geometric graphs}},
    year = {2003},
    publisher = {Oxford University Press},
    location = {Oxford}
}

@book{bfs,
  title        = {{Introduction to Algorithms}},
  author       = {Cormen, Thomas H. and Leiserson, Charles E. and Rivest, Ronald L. and Stein, Clifford},
  year         = {2009},
  publisher    = {MIT Press},
  edition      = {3rd},
  address      = {Cambridge, MA}
}

@book{hwang1992Steinertree_book,
	title = {{The Steiner Tree Problem}},
	author = {Frank K. Hwang and Dana S. Richards and Pawel Winter},
	series = {Annals of Discrete Mathematics},
	volume = {53},
	year = {1992},
	publisher = {North-Holland}
}

@article{garey1977SteinertreeNP,
author = {Garey, M. R. and Graham, R. L. and Johnson, D. S.},
title = {{The Complexity of Computing Steiner Minimal Trees}},
journal = {SIAM Journal on Applied Mathematics},
volume = {32},
number = {4},
pages = {835--859},
year = {1977},
doi = {10.1137/0132072},
URL = {https://doi.org/10.1137/0132072},
}

@article{kou1981steiner,
	author = {Kou, L. and Markowsky, G. and Berman, L.},
	title = {A fast algorithm for Steiner trees},
	journal = {Acta Informatica},
	volume = {15},
	issue = {2},
	pages = {141--145},
	year = {1981},
	url = {https://doi.org/10.1007/BF00288961}
}

@article{mehlhorn1988steiner,
author = {Kurt Mehlhorn},
title = {{A faster approximation algorithm for the Steiner problem in graphs}},
journal = {Information Processing Letters},
volume = {27},
number = {3},
pages = {125--128},
year = {1988},
issn = {0020-0190},
doi = {https://doi.org/10.1016/0020-0190(88)90066-X},
}

@article{CCT+20,
  title = {{Quantum Internet: Networking Challenges in Distributed Quantum Computing}},
  author = {Cacciapuoti, Angela Sara and Caleffi, Marcello and Tafuri, Francesco and Cataliotti, Francesco Saverio and Gherardini, Stefano and Bianchi, Giuseppe},
  year = 2020,
  journal = {IEEE Network},
  volume = 34,
  number = 1,
  pages = {137--143},
  doi = {10.1109/MNET.001.1900092},
  url = {https://doi.org/10.1109/MNET.001.1900092},
  eprint = {1810.08421},
  eprinttype = {arXiv}
}

@article{freund2024quantumdatabus,
  title = {Flexible quantum data bus for quantum networks},
  author = {Freund, Julia and Pirker, Alexander and D\"{u}r, Wolfgang},
  journal = {Physical Review Research},
  volume = {6},
  issue = {3},
  pages = {033267},
  numpages = {11},
  year = {2024},
  month = {Sep},
  publisher = {American Physical Society},
  doi = {10.1103/PhysRevResearch.6.033267},
  url = {https://doi.org/10.1103/PhysRevResearch.6.033267},
  eprint = {2404.06578},
  eprinttype = {arXiv}
}

@article{exp_1,
author = {Arian J. Stolk  and Kian L. van der Enden  and Marie-Christine Slater  and Ingmar te Raa-Derckx  and Pieter Botma  and Joris van Rantwijk  and J. J. Benjamin Biemond  and Ronald A. J. Hagen  and Rodolf W. Herfst  and Wouter D. Koek  and Adrianus J. H. Meskers  and René Vollmer  and Erwin J. van Zwet  and Matthew Markham  and Andrew M. Edmonds  and J. Fabian Geus  and Florian Elsen  and Bernd Jungbluth  and Constantin Haefner  and Christoph Tresp  and Jürgen Stuhler  and Stephan Ritter  and Ronald Hanson },
title = {Metropolitan-scale heralded entanglement of solid-state qubits},
journal = {Science Advances},
volume = {10},
number = {44},
pages = {eadp6442},
year = {2024},
doi = {10.1126/sciadv.adp6442},
URL = {https://www.science.org/doi/abs/10.1126/sciadv.adp6442},
eprint = {2404.03723},
eprinttype = {arXiv}
}

@article{exp_3,
  title = {Demonstration of quantum network protocols over a 14-km urban fiber link},
  volume = {10},
  ISSN = {2056-6387},
  number = {1},
  journal = {npj Quantum Information},
  author = {Kucera,  Stephan and Haen,  Christian and Arensk\"{o}tter,  Elena and Bauer,  Tobias and Meiers,  Jonas and Sch\"{a}fer,  Marlon and Boland,  Ross and Yahyapour,  Milad and Lessing,  Maurice and Holzwarth,  Ronald and Becher,  Christoph and Eschner,  J\"{u}rgen},
  year = {2024},
  url = {https://doi.org/10.1038/s41534-024-00886-x},
  eprint = {2404.04958},
  eprinttype = {arXiv}
}

@article{exp_4,
  title = {{Long-Lived Quantum Memory Enabling Atom-Photon Entanglement over 101 km of Telecom Fiber}},
  author = {Zhou, Yiru and Malik, Pooja and Fertig, Florian and Bock, Matthias and Bauer, Tobias and van Leent, Tim and Zhang, Wei and Becher, Christoph and Weinfurter, Harald},
  journal = {PRX Quantum},
  volume = {5},
  issue = {2},
  pages = {020307},
  numpages = {14},
  year = {2024},
  month = {Apr},
  publisher = {American Physical Society},
  doi = {10.1103/PRXQuantum.5.020307},
  url = {https://doi.org/10.1103/PRXQuantum.5.020307},
  eprint = {2308.08892},
  eprinttype = {arXiv}
}

@article{exp_5,
author = {Lukas Hartung  and Matthias Seubert  and Stephan Welte  and Emanuele Distante  and Gerhard Rempe },
title = {A quantum-network register assembled with optical tweezers in an optical cavity},
journal = {Science},
volume = {385},
number = {6705},
pages = {179--183},
year = {2024},
doi = {10.1126/science.ado6471},
URL = {https://doi.org/10.1126/science.ado6471},
eprint = {2407.09109},
eprinttype = {arXiv}
}

@article{simon_RL,
  title = {Deep reinforcement learning for key distribution based on quantum repeaters},
  author = {Rei\ss{}, Simon D. and van Loock, Peter},
  journal = {Physical Review A},
  volume = {108},
  issue = {1},
  pages = {012406},
  numpages = {25},
  year = {2023},
  month = {Jul},
  url = {https://doi.org/10.1103/PhysRevA.108.012406},
  eprint = {2207.09930},
  eprinttype = {arXiv}
}

@article{evgeny_swap_asap,
  title = {{Optimal Entanglement Swapping in Quantum Repeaters}},
  author = {Shchukin, Evgeny and van Loock, Peter},
  journal = {Physical Review Letters},
  volume = {128},
  issue = {15},
  pages = {150502},
  numpages = {5},
  year = {2022},
  month = {Apr},
  url = {https://doi.org/10.1103/PhysRevLett.128.150502},
  eprint = {2109.00793},
  eprinttype = {arXiv}
}

@article{waiting_time,
  title = {Waiting time in quantum repeaters with probabilistic entanglement swapping},
  author = {Shchukin, E. and Schmidt, F. and van Loock, P.},
  journal = {Physical Review A},
  volume = {100},
  issue = {3},
  pages = {032322},
  numpages = {20},
  year = {2019},
  month = {Sep},
  url = {https://doi.org/10.1103/PhysRevA.100.032322},
  eprint = {1710.06214},
  eprinttype = {arXiv}
}

@article{lars,
  title = {Exact rate analysis for quantum repeaters with imperfect memories and entanglement swapping as soon as possible},
  author = {Kamin, Lars and Shchukin, Evgeny and Schmidt, Frank and van Loock, Peter},
  journal = {Physical Review Research},
  volume = {5},
  issue = {2},
  pages = {023086},
  numpages = {52},
  year = {2023},
  month = {May},
  doi = {10.1103/PhysRevResearch.5.023086},
  url = {https://doi.org/10.1103/PhysRevResearch.5.023086},
  eprint = {2203.10318},
  eprinttype = {arXiv}
}

@article{tim_coopmans_2024,
      title={{On noise in swap ASAP repeater chains: exact analytics, distributions and tight approximations}}, 
      author={Kenneth Goodenough and Tim Coopmans and Don Towsley},
      year={2024},
      journal={arXiv:2404.07146},
      url={https://arxiv.org/abs/2404.07146}
}

@article{haldar2024multiplexing,
  title = {{Reducing classical communication costs in multiplexed quantum repeaters using hardware-aware quasi-local policies}},
  author = {Stav Haldar and Pratik J. Barge and Xiang Cheng and Kai-Chi Chang and Brian T. Kirby and Sumeet Khatri and Chee Wei Wong and Hwang Lee},
  journal = {Communications Physics},
  volume = 8,
  issue = 1,
  pages = 132,
  year = 2025,
  url = {https://doi.org/10.1038/s42005-025-02029-w},
  eprint = {2401.13168},
  eprinttype = {arXiv}
}

@article{robins2005steinertreeapproximation,
author = {Robins, Gabriel and Zelikovsky, Alexander},
title = {{Tighter Bounds for Graph Steiner Tree Approximation}},
journal = {SIAM Journal on Discrete Mathematics},
volume = {19},
number = {1},
pages = {122--134},
year = {2005},
doi = {10.1137/S0895480101393155},
URL = {https://doi.org/10.1137/S0895480101393155},
}

@book{dom_set_np,
  author    = {Michael R. Garey and David S. Johnson},
  title     = {{Computers and Intractability: A Guide to the Theory of NP-Completeness}},
  year      = {1979},
  edition   = {1st},
  series    = {{Series of Books in the Mathematical Sciences}},
  publisher = {W. H. Freeman and Company},
  address   = {New York},
  isbn      = {9780716710455},
}

@book{dom_set_alg_book,
  title     = {{Approximation Algorithms}},
  author    = {Vazirani, Vijay V.},
  year      = {2001},
  publisher = {Springer Science \& Business Media},
  isbn      = {978-3-540-65367-7}
}

@article{fredman1994spanningtrees,
title = {Trans-dichotomous algorithms for minimum spanning trees and shortest paths},
author = {Michael L. Fredman and Dan E. Willard},
journal = {Journal of Computer and System Sciences},
volume = {48},
number = {3},
pages = {533--551},
year = {1994},
issn = {0022-0000},
url = {https://doi.org/10.1016/S0022-0000(05)80064-9}
}

@article{Moses2023,
  title = {{A Race-Track Trapped-Ion Quantum Processor}},
  author = {Moses, S. A. and Baldwin, C. H. and Allman, M. S. and Ancona, R. and Ascarrunz, L. and Barnes, C. and Bartolotta, J. and Bjork, B. and Blanchard, P. and Bohn, M. and Bohnet, J. G. and Brown, N. C. and Burdick, N. Q. and Burton, W. C. and Campbell, S. L. and Campora, J. P. and Carron, C. and Chambers, J. and Chan, J. W. and Chen, Y. H. and Chernoguzov, A. and Chertkov, E. and Colina, J. and Curtis, J. P. and Daniel, R. and DeCross, M. and Deen, D. and Delaney, C. and Dreiling, J. M. and Ertsgaard, C. T. and Esposito, J. and Estey, B. and Fabrikant, M. and Figgatt, C. and Foltz, C. and Foss-Feig, M. and Francois, D. and Gaebler, J. P. and Gatterman, T. M. and Gilbreth, C. N. and Giles, J. and Glynn, E. and Hall, A. and Hankin, A. M. and Hansen, A. and Hayes, D. and Higashi, B. and Hoffman, I. M. and Horning, B. and Hout, J. J. and Jacobs, R. and Johansen, J. and Jones, L. and Karcz, J. and Klein, T. and Lauria, P. and Lee, P. and Liefer, D. and Lu, S. T. and Lucchetti, D. and Lytle, C. and Malm, A. and Matheny, M. and Mathewson, B. and Mayer, K. and Miller, D. B. and Mills, M. and Neyenhuis, B. and Nugent, L. and Olson, S. and Parks, J. and Price, G. N. and Price, Z. and Pugh, M. and Ransford, A. and Reed, A. P. and Roman, C. and Rowe, M. and Ryan-Anderson, C. and Sanders, S. and Sedlacek, J. and Shevchuk, P. and Siegfried, P. and Skripka, T. and Spaun, B. and Sprenkle, R. T. and Stutz, R. P. and Swallows, M. and Tobey, R. I. and Tran, A. and Tran, T. and Vogt, E. and Volin, C. and Walker, J. and Zolot, A. M. and Pino, J. M.},
  journal = {Physical Review X},
  volume = {13},
  issue = {4},
  pages = {041052},
  numpages = {25},
  year = {2023},
  month = {Dec},
  publisher = {American Physical Society},
  doi = {10.1103/PhysRevX.13.041052},
  url = {https://link.aps.org/doi/10.1103/PhysRevX.13.041052}
}

@article{vanLoock2020,
  title = {{Extending Quantum Links: Modules for Fiber‐ and Memory‐Based Quantum Repeaters}},
  volume = {3},
  ISSN = {2511-9044},
  url = {http://dx.doi.org/10.1002/qute.201900141},
  number = {11},
  journal = {Advanced Quantum Technologies},
  publisher = {Wiley},
  author = {van Loock,  Peter and Alt,  Wolfgang and Becher,  Christoph and Benson,  Oliver and Boche,  Holger and Deppe,  Christian and Eschner,  J\"{u}rgen and H\"{o}fling,  Sven and Meschede,  Dieter and Michler,  Peter and Schmidt,  Frank and Weinfurter,  Harald},
  year = {2020},
  month = oct 
}

@article{dom_set_bound_2001,
  author = {Wieland, Ben and Godbole, Anant P.},
  title = {On the Domination Number of a Random Graph},
  journal = {The Electronic Journal of Combinatorics},
  volume = {8},
  number = {1},
  pages = {Research Paper R37},
  year = {2001},
  url = {https://www.combinatorics.org/ojs/index.php/eljc/article/view/v8i1r37}
}

@article{dom_set_bound_2015,
  author = {Glebov, Roman and Liebenau, Anita and Szabó, Tibor},
  title = {On the Concentration of the Domination Number of the Random Graph},
  journal = {SIAM Journal on Discrete Mathematics},
  volume = {29},
  number = {3},
  pages = {1186--1206},
  year = {2015},
  doi = {10.1137/12090054X}
}

\clearpage\newpage

\let\addcontentsline\oldaddcontentsline 
\title{Supplementary Information: A resource- and computationally-efficient protocol for\\multipartite entanglement distribution in Bell-pair networks}

\maketitle

\onecolumngrid


\tableofcontents

\section{Relation to prior work}\label{sec-prior_work}

Prior work on the distribution of multipartite entanglement in genuine network settings, going beyond repeater chains, includes Refs.~\cite{CC12,WZM+16,PWD18,KVS+19,PD19,MMG19,brito2020statistical,deBone2020GHZdistillBell,fischer2021graph,bugalho2023multipartite,sadhu2023practical,avis2023multipartitecentral,roga2023dickestatedistribution,shimizu2024GHZdistribution, sen2023multipartitesubgraphcomplementation}. Below, we highlight and compare our work to some of the most relevant of these prior works. 
    \begin{itemize}
        \item Reference~\cite{CC12} considers the distribution of arbitrary graph states, in which the graph of the corresponding graph state matches the topology of the underlying network of Bell pairs.
        
        \item Reference~\cite{fischer2021graph} involves local preparation of the desired graph state, followed by teleportation of the state to the desired parties using the Bell pair network, which can be arbitrary. The basis of their protocol is the so-called ``Bipartite~A'' protocol of Ref.~\cite{CC12}. Essentially, their protocol is the Bipartite~A protocol supplemented by an analysis of completion time based on particular algorithms for finding paths in the Bell-pair network for the purpose of teleporting the locally-created graph state.
        
        \item Reference~\cite{avis2023multipartitecentral} considers the distribution of GHZ states in a star Bell-pair network topology, in which the central ``factory'' node locally produces the GHZ states and then uses the Bell pairs to teleport it to the outer parties. This is also essentially the Bipartite~A protocol of Ref.~\cite{CC12}, with noisy Bell pairs and an analysis of the rate at which the GHZ states can be distributed.
        
        \item Reference~\cite{MMG19} uses the Steiner tree, along with a sub-routine called ``star expansion'', to distribution a GHZ state between a given set of nodes in the network.
        
        \item Reference~\cite{bugalho2023multipartite} also makes use of the Steiner tree problem, and includes an analysis with noisy Bell pairs.
         
        \item Similarly to Ref.~\cite{avis2023multipartitecentral}, Ref.~\cite{sen2023multipartitesubgraphcomplementation} requires a star Bell-pair topology. Here, an arbitrary graph state can be produced using a ``subgraph complementation system'' associated with the graph of the desired graph state. Reference~\cite{shimizu2024GHZdistribution} also considers a star topology only.
    \end{itemize}

\begin{figure*}
    \centering
    \includegraphics[width=0.8\textwidth]{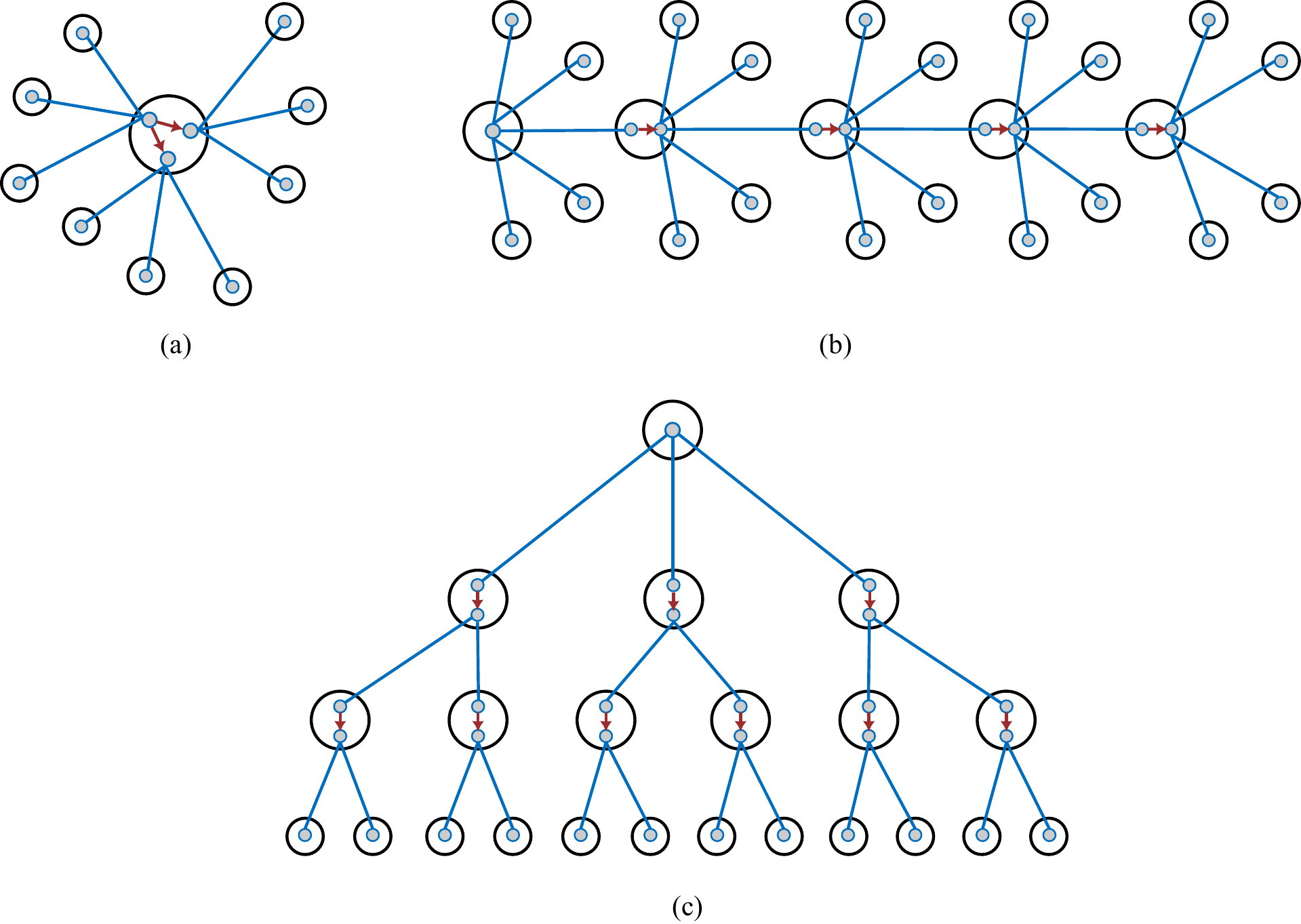}
    \caption{Fusion of GHZ states in various configurations. The arrows indicate CNOT gates, pointing from the control qubit to the target qubit. The target qubit is then measured in the Pauli-$Z$ basis. (a) Three GHZ states in a star configuration (Protocol~\ref{alg:multiple_GHZ_merge_protocol}). (b) Five GHZ states in a linear chain (Protocol~\ref{alg:multiple_GHZ_merge_linear_protocol}). (c) 10 GHZ states in a tree configuration (Protocol~\ref{alg:multiple_GHZ_merge_tree_protocol_main}).}
    \label{fig:GHZ_merging}
\end{figure*}

\section{Fusion of GHZ states in a linear topology (Protocol~\ref{alg:multiple_GHZ_merge_linear_protocol})}

Fusing GHZ states in a linear topology is presented in Protocol~\ref{alg:multiple_GHZ_merge_linear_protocol}.

\begin{algorithm}[H]
\caption{Fusion of GHZ states in a linear topology~\cite{kruszynska2006purificationgraphstates,PWD18}}\label{alg:multiple_GHZ_merge_linear_protocol}
\begin{algorithmic}[1]
\Require $N\in\{2,3,\dotsc\}$ GHZ states, $\ket{\text{GHZ}_{n_1+1}}_{A_{1:n_1}^1R_1},\, \ket{\text{GHZ}_{n_2+2}}_{R_2A_{1:n_2}^2R_3},\dotsc,\ket{\text{GHZ}_{n_{N-1}+2}}_{R_{2N-4}A_{1:n_{N-1}}^{N-1}R_{2N-3}},\,\ket{\text{GHZ}_{n_N+1}}_{R_{2N-2}A_{1:n_N}^N}$, such that the qubits $R_i$ and $R_{i+1}$, $i\in\{1,3,5,\dotsc,2N-1\}$, are located at the same node; see Fig.~\ref{fig:GHZ_merging}(b).
\Ensure GHZ state $\ket{\text{GHZ}_{N_{\text{tot}}}}$, with $N_{\text{tot}}=n_1+n_2+\dotsb+n_N+N-1$, shared by $A_{1:n_1}^1,R_1,A_{1:n_2}^2,R_3,\dotsc,A_{1:n_{N-1}}^{N-1},R_{2N-3},A_{1:n_N}^N$.
\State Apply the CNOT gates $\text{CNOT}_{R_iR_{i+1}}$, $i\in\{1,3,5,\dotsc,2N-3\}$.
\State Measure the target qubits $R_2,R_4,R_6,\dotsc,R_{2N-2}$ in the $Z$-basis $\{\ket{0},\ket{1}\}$. The outcomes $z_1,z_2,\dotsc,z_{N-1}\in\{0,1\}$ are communicated to all nodes.
\State The qubits associated with the second GHZ state apply the Pauli-$X$ correction $X^{z_1}$; the qubits associated with the third GHZ state apply the Pauli-$X$ correction $X^{z_1\oplus z_2}$; etc. In general, the qubits associated with the $k$-th node, $k\in\{2,3,\dotsc,N\}$, apply the Pauli-$X$ correction $X^{z_1\oplus z_2\oplus\dotsb\oplus z_{k-1}}$.
\end{algorithmic}
\end{algorithm}

Let us verify that Protocol~\ref{alg:multiple_GHZ_merge_linear_protocol} performs as expected. The initial state is
\begin{multline}
    \ket{\text{GHZ}_{n_1+1}}_{A_{1:n_1}^1R_1}\otimes\ket{\text{GHZ}_{n_2+2}}_{R_2A_{1:n_2}^2R_3}\otimes\ket{\text{GHZ}_{n_3+2}}_{R_4A_{1:n_3}^3R_5}\dotsb\otimes\ket{\text{GHZ}_{n_{N-1}+2}}_{R_{2N-4}A_{1:n_{N-1}}^{N-1}R_{2N-3}}\otimes\ket{\text{GHZ}_{n_N+1}}_{R_{2N-2}A_{1:n_N}^N}\\=\frac{1}{\sqrt{2^N}}\sum_{x_1,x_2,x_3,\dotsc,x_{N-1},x_N\in\{0,1\}}\ket{x_1}_{A_{1:n_1}^1}^{\otimes n_1}\ket{x_1}_{R_1}\otimes\ket{x_2}_{R_2}\ket{x_2}_{A_{1:n_2}^2}^{\otimes n_2}\ket{x_2}_{R_3}\otimes\ket{x_3}_{R_4}\ket{x_3}_{A_{1:n_3}}^{\otimes n_3}\ket{x_3}_{R_5}\otimes\dotsb\\\dotsb\otimes\ket{x_{N-1}}_{R_{2N-4}}\ket{x_{N-1}}_{A_{1:n_{N-1}}^{N-1}}^{\otimes n_{N-1}}\ket{x_{N-1}}_{R_{2N-3}}\otimes\ket{x_N}_{R_{2N-2}}\ket{x_N}_{A_{1:n_N}^N}^{\otimes n_N}.
\end{multline}
After the CNOT gates, we obtain
\begin{multline}
    \frac{1}{\sqrt{2^N}}\sum_{x_1,x_2,x_3,\dotsc,x_{N-1},x_N\in\{0,1\}}\ket{x_1}_{A_{1:n_1}^1}^{\otimes n_1}\ket{x_1}_{R_1}\otimes\ket{x_1\oplus x_2}_{R_2}\ket{x_2}_{A_{1:n_2}^2}^{\otimes n_2}\ket{x_2}_{R_3}\otimes\ket{x_2\oplus x_3}_{R_4}\ket{x_3}_{A_{1:n_3}^3}^{\otimes n_3}\ket{x_3}_{R_5}\otimes\dotsb\\
    \dotsb\otimes\ket{x_{N-2}\oplus x_{N-1}}_{R_{2N-4}}\ket{x_{N-1}}_{A_{1:n_{N-1}}^{N-1}}^{\otimes n_{N-1}}\ket{x_{N-1}}_{R_{2N-3}}\otimes\ket{x_{N-1}\oplus x_N}_{R_{2N-2}}\ket{x_N}_{A_{1:n_N}^N}^{\otimes n_N}.
\end{multline}
Then, upon measuring the target qubits $R_2,R_4,\dotsc$ in the $\{\ket{0},\ket{1}\}$ basis, to every outcome string $(z_1,z_2,\dotsc,z_{N-1})\in\{0,1\}^{N-1}$ we have the following associated (unnormalized) post-measurement state:
\begin{multline}
    \frac{1}{\sqrt{2^N}}\sum_{x_1,x_2,x_3,\dotsc,x_{N-1},x_N\in\{0,1\}}\ket{x_1}_{A_{1:n_1}^1}^{\otimes n_1}\ket{x_1}_{R_1}\otimes\braket{z_1}{x_1\oplus x_2}_{R_2}\ket{x_2}_{A_{1:n_2}^2}^{\otimes n_2}\ket{x_2}_{R_3}\otimes\braket{z_2}{x_2\oplus x_3}_{R_4}\ket{x_3}_{A_{1:n_3}^3}^{\otimes n_3}\ket{x_3}_{R_5}\otimes\dotsb\\
    \dotsb\otimes\braket{z_{N-2}}{x_{N-2}\oplus x_{N-1}}_{R_{2N-4}}\ket{x_{N-1}}_{A_{1:n_{N-1}}^{N-1}}^{\otimes n_{N-1}}\ket{x_{N-1}}_{R_{2N-3}}\otimes\braket{z_{N-1}}{x_{N-1}\oplus x_N}_{R_{2N-2}}\ket{x_N}_{A_{1:n_N}^N}^{\otimes n_N},
\end{multline}
which we can simplify to
\begin{multline}
    \frac{1}{\sqrt{2^N}}\sum_{x_1\in\{0,1\}}\ket{x_1}_{A_{1:n_1}^1}^{\otimes n_1}\ket{x_1}_{R_1}\otimes\ket{z_1\oplus x_1}_{A_{1:n_2}^2}^{\otimes n_2}\ket{z_1\oplus x_1}_{R_3}\otimes\ket{z_1\oplus z_2\oplus x_1}_{A_{1:n_3}^3}^{\otimes n_3}\ket{z_1\oplus z_2\oplus x_1}_{R_5}\otimes\dotsb\\
    \dotsb\otimes\ket{z_1\oplus z_2\oplus\dotsb\oplus z_{N-2}\oplus x_{1}}_{A_{1:n_{N-1}}^{N-1}}^{\otimes n_{N-1}}\ket{z_1\oplus z_2\oplus\dotsb\oplus z_{N-2}\oplus x_{1}}_{R_{2N-3}}\otimes\ket{z_1\oplus z_2\oplus\dotsb\oplus z_{N-1}\oplus x_1}_{A_{1:n_N}^N}^{\otimes n_N}\\=\left((X^{z_1})_{A_{1:n_2}^2R_3}^{\otimes(n_2+1)}\otimes (X^{z_1\oplus z_2})_{A_{1:n_3}^3R_5}^{\otimes (n_3+1)}\otimes\dotsb\otimes (X^{z_1\oplus z_2\oplus\dotsb\oplus z_{N-2}})_{A_{1:n_{N-1}}^{N-1}R_{2N-3}}^{\otimes(n_{N-1}+1)}\otimes (X^{z_1\oplus z_2\oplus\dotsb\oplus z_{N-1}})_{A_{1:n_N}^N}^{\otimes n_N}\right)\ket{\text{GHZ}_{N_{\text{tot}}}}.
\end{multline}
We thus see that, as expected, the resulting state is the required GHZ state up to the appropriate Pauli-$X$ corrections.

As an LOCC channel, we can write the action of Protocol~\ref{alg:multiple_GHZ_merge_linear_protocol} as follows:
\begin{multline}\label{eq-multpiple_GHZ_merge_linear_protocol_LOCC_chan}
    \mathcal{L}\left(\rho_{A_{1:n_1}^1R_1R_2A_{1:n_2}^2R_3\dotsb R_{2N-4}A_{1:n_{N-1}}^{N-1}R_{2N-3}R_{2N-2}A_{1:n_N}^N}\right)\\=\sum_{\vec{z}\in\{0,1\}^{N-1}}\left((X^{y_1})_{A_{1:n_2}^2R_3}^{\otimes(n_2+1)}\otimes (X^{y_2})_{A_{1:n_3}^3R_5}^{\otimes (n_3+1)}\otimes\dotsb\otimes (X^{y_{N-2}})_{A_{1:n_{N-1}}^{N-1}R_{2N-3}}^{\otimes(n_{N-1}+1)}\otimes (X^{y_{N-1}})_{A_{1:n_N}^N}^{\otimes n_N}\right)\\\times\left(\bra{\vec{z}}_{R_2R_4R_6\dotsb R_{2N-2}}\text{CNOT}_N^{\prime}\right)\rho_{A_{1:n_1}^1R_1R_2A_{1:n_2}^2R_3\dotsb R_{2N-4}A_{1:n_{N-1}}^{N-1}R_{2N-3}R_{2N-2}A_{1:n_N}^N}\left(\text{CNOT}_N^{\prime}\ket{\vec{z}}_{R_2R_4R_6\dotsb R_{2N-2}}\right)\\\times\left((X^{y_1})_{A_{1:n_2}^2R_3}^{\otimes(n_2+1)}\otimes (X^{y_2})_{A_{1:n_3}^3R_5}^{\otimes (n_3+1)}\otimes\dotsb\otimes (X^{y_{N-2}})_{A_{1:n_{N-1}}^{N-1}R_{2N-3}}^{\otimes(n_{N-1}+1)}\otimes (X^{y_{N-1}})_{A_{1:n_N}^N}^{\otimes n_N}\right),
\end{multline}
where
\begin{align}
    y_k&\coloneqq z_1\oplus z_2\oplus\dotsb\oplus z_k,\quad k\in\{1,2,\dotsc,N-1\},\label{eq-multiple_GHZ_merge_linear_pf1}\\
    \text{CNOT}_N^{\prime}&\coloneqq \text{CNOT}_{R_1R_2}\text{CNOT}_{R_3R_4}\dotsb \text{CNOT}_{R_{2N-3}R_{2N-2}}.\label{eq-multiple_GHZ_merge_linear_pf2}
\end{align}

\begin{proposition}[Fidelity after Protocol~\ref{alg:multiple_GHZ_merge_linear_protocol}]\label{prop:multiple_GHZ_merge_linear_protocol_post_fid}
    Consider states $\rho_{A_{1:n_1}^1R_1}^1,\rho_{R_2A_{1:n_2}^2R_3}^2,\dotsc,\rho_{R_{2N-2}A_{1:n_N}^N}^N$, representing noisy input states to Protocol~\ref{alg:multiple_GHZ_merge_linear_protocol}. With these input states, the fidelity of the state produced by Protocol~\ref{alg:multiple_GHZ_merge_linear_protocol} with respect to the final, ideal GHZ state is
    \begin{multline}\label{eq-multiple_GHZ_merge_linear_protocol_fidelity}
        \bra{\text{GHZ}_{N_{\text{tot}}}}\mathcal{L}\!\left(\rho_{A_{1:n_1}^1R_1}^1\otimes\rho_{R_2A_{1:n_2}^2R_3}^2\otimes\dotsb\otimes\rho_{R_{2N-2}A_{1:n_N}^N}^N\right)\\
        =\sum_{y_1,y_2,\dotsc,y_{N-1}\in\{0,1\}}\bra{\text{GHZ}_{n_1+1}^{(y_1,\vec{0})}}\rho_{A_{1:n_1}^1R_1}^1\ket{\text{GHZ}_{n_1+1}^{(y_1,\vec{0})}}\bra{\text{GHZ}_{n_2+2}^{(y_2,\vec{0})}}\rho_{R_2A_{1:n_2}^2R_3}^2\ket{\text{GHZ}_{n_2+2}^{(y_2,\vec{0})}}\dotsb\\\dotsb\bra{\text{GHZ}_{n_N+1}^{(y',\vec{0})}}\rho_{R_{2N-2}A_{1:n_N}}^N\ket{\text{GHZ}_{n_N+1}^{(y',\vec{0})}},
    \end{multline}
    where $y'=y_1\oplus y_2\oplus\dotsb\oplus y_{N-1}$.
\end{proposition}

\begin{proof}
    We start by defining the vector
    \begin{multline}
        \ket{u(\vec{z})}\coloneqq\left(\text{CNOT}_N^{\prime}\ket{\vec{z}}_{R_2R_4R_6\dotsb R_{2N-2}}\right)\left((X^{y_1})_{A_{1:n_2}^2R_3}^{\otimes(n_2+1)}\otimes (X^{y_2})_{A_{1:n_3}^3R_5}^{\otimes (n_3+1)}\otimes\dotsb\right.\\\left.\dotsb\otimes (X^{y_{N-2}})_{A_{1:n_{N-1}}^{N-1}R_{2N-3}}^{\otimes(n_{N-1}+1)}\otimes (X^{y_{N-1}})_{A_{1:n_N}^N}^{\otimes n_N}\right)\ket{\text{GHZ}_{N_{\text{tot}}}}_{A_{1:n_1}^1R_1A_{1:n_2}^2R_3\dotsb A_{1:n_{N-1}}^{N-1}R_{2N-3}A_{1:n_N}^N},
    \end{multline}
    which can be simplified as follows:
    \begin{align}
        \ket{u(\vec{z})}&=\frac{1}{\sqrt{2}}\sum_{\vec{x}\in\{0,1\}}\ket{x}_{A_{1:n_1}^1}^{\otimes n_1}\text{CNOT}_{R_1R_2}\ket{x,z_1}_{R_1R_2}\otimes(X^{y_1}\ket{x})_{A_{1:n_2}^2}^{\otimes n_2}\text{CNOT}_{R_3R_4}X_{R_3}^{y_1}\ket{x,z_2}_{R_3R_4}\otimes (X^{y_2}\ket{x})_{A_{1:n_3}^3}^{\otimes n_3}\text{CNOT}_{R_5R_6}X_{R_5}^{y_2}\ket{x,z_3}_{R_5R_6}\otimes\dotsb\nonumber\\
        &\qquad\qquad\dotsb\otimes(X^{y_{N-2}}\ket{x})_{A_{1:n_{N-1}}^{N-1}}^{\otimes n_{N-1}}\text{CNOT}_{R_{2N-3}R_{2N-2}}X_{R_{2N-3}}^{y_{N-2}}\ket{x,z_{N-1}}_{R_{2N-3}R_{2N-2}}(X^{y_{N-1}}\ket{x})_{A_{1:n_N}}^{\otimes n_{N}}\\
        &=\frac{1}{\sqrt{2}}\sum_{x\in\{0,1\}}\ket{x}_{A_{1:n_1}}^{\otimes n_1}\ket{x}_{R_1}\otimes\ket{x\oplus z_1}_{R_2}\ket{x\oplus y_1}_{A_{1:n_2}^2}^{\otimes n_2}\ket{x\oplus y_1}_{R_3}\otimes\ket{x\oplus y_1\oplus z_2}_{R_4}\ket{x\oplus y_2}_{A_{1:n_3}^3}^{\otimes n_3}\ket{x\oplus y_2}_{R_5}\ket{x\oplus y_2\oplus z_3}_{R_6}\otimes\dotsb\nonumber\\
        &\qquad\qquad\dotsb\otimes\ket{x\oplus y_{N-2}}_{A_{1:n_{N-1}}^{N-1}}^{\otimes n_{N-1}}\ket{x\oplus y_{N-2}}_{R_{2N-3}}\otimes\ket{x\oplus y_{N-2}\oplus z_{N-1}}_{R_{2N-2}}\ket{x\oplus y_{N-1}}_{A_{1:n_N}^N}^{\otimes n_{N}}\\
        &=\frac{1}{\sqrt{2}}\sum_{x\in\{0,1\}}\ket{x}_{A_{1:n_1}}^{\otimes n_1}\ket{x}_{R_1}\otimes\ket{x\oplus z_1}_{R_2}\ket{x\oplus z_1}_{A_{1:n_2}^2}^{\otimes n_2}\ket{x\oplus z_1}_{R_3}\otimes\ket{x\oplus z_1\oplus z_2}_{R_4}\ket{x\oplus z_1\oplus z_2}_{A_{1:n_3}^3}^{\otimes n_3}\ket{x\oplus z_1\oplus z_2}_{R_5}\nonumber\\
        &\qquad\qquad\otimes\ket{x\oplus z_1\oplus z_2\oplus z_3}_{R_6}\otimes\dotsb\otimes\ket{x\oplus z_1\oplus\dotsb\oplus z_{N-2}}_{A_{1:n_{N-1}}^{N-1}}^{\otimes n_{N-1}}\ket{x\oplus z_1\oplus\dotsb\oplus z_{N-2}}_{R_{2N-3}}\nonumber\\
        &\qquad\qquad\otimes\ket{x\oplus z_1\oplus\dotsb\oplus z_{N-2}\oplus z_{N-1}}_{R_{2N-2}}\ket{x\oplus z_1\oplus\dotsb\oplus z_{N-1}}_{A_{1:n_N}^N}^{\otimes n_{N}},
    \end{align}
    where we have used \eqref{eq-multiple_GHZ_merge_linear_pf1} and \eqref{eq-multiple_GHZ_merge_linear_pf2}. Let us now use $\ket{x}^{\otimes n}=\frac{1}{\sqrt{2}}\sum_{z\in\{0,1\}}(-1)^{zx}\ket{\text{GHZ}_n^{(z,\vec{0})}}$ to make further simplifications, as follows:
    \begin{align}
        \ket{u(\vec{z})}&=\frac{1}{\sqrt{2}}\frac{1}{\sqrt{2^N}}\sum_{x\in\{0,1\}}\sum_{y_1,y_2,\dotsc,y_N\in\{0,1\}}(-1)^{xy_1}\ket{\text{GHZ}_{n_1+1}^{(y_1,\vec{0})}}_{A_{1:n_1}^1R_1}\otimes (-1)^{(x\oplus z_1)y_2}\ket{\text{GHZ}_{n_2+2}^{(y_2,\vec{0})}}_{R_2A_{1:n_2}^2R_3}\nonumber\\
        &\qquad\qquad\otimes (-1)^{(x\oplus z_1\oplus z_2)y_3}\ket{\text{GHZ}_{n_3+2}^{(y_3,\vec{0})}}_{R_4A_{1:n_3}^3R_5}\otimes\dotsb\otimes (-1)^{y_{N-1}(x\oplus z_1\oplus\dotsb\oplus z_{N-2})}\ket{\text{GHZ}_{n_{N-1}+2}^{(y_{N-1},\vec{0})}}_{R_{2N-4}A_{1:n_{N-1}}^NR_{2N-3}}\\
        &\qquad\qquad\otimes(-1)^{y_N(x\oplus z_1\oplus\dotsb\oplus z_{N-1})}\ket{\text{GHZ}_{n_N+1}^{(y_N,\vec{0})}}_{R_{2N-2}A_{1:n_N}}\\
        &=\frac{1}{\sqrt{2}}\frac{1}{\sqrt{2^N}}\sum_{y_1,y_2,\dotsc,y_N\in\{0,1\}}\left(\sum_{x\in\{0,1\}}(-1)^{x(y_1+y_2+\dotsb+y_N)}\right)(-1)^{z_1(y_2+\dotsb +y_N)}(-1)^{z_2(y_3+\dotsb+y_N)}\dotsb(-1)^{z_{N-2}(y_{N-1}+y_N)}(-1)^{z_{N-1}y_{N}}\nonumber\\
        &\qquad\qquad\times \ket{\text{GHZ}_{n_1+1}^{(y_1,\vec{0})}}_{A_{1:n_1}^1R_1}\otimes\ket{\text{GHZ}_{n_2+2}^{(y_2,\vec{0})}}_{R_2A_{1:n_2}^2R_3}\otimes\ket{\text{GHZ}_{n_3+2}^{(y_3,\vec{0})}}_{R_4A_{1:n_3}^3R_5}\otimes\dotsb\nonumber\\
        &\qquad\qquad\dotsb\otimes\ket{\text{GHZ}_{n_{N-1}+2}^{(y_{N-1},\vec{0})}}_{R_{2N-4}A_{1:n_{N-1}}^NR_{2N-3}}\otimes \ket{\text{GHZ}_{n_N+1}^{(y_N,\vec{0})}}_{R_{2N-2}A_{1:n_N}}\\
        &=\frac{1}{\sqrt{2^{N-1}}}\sum_{y_1,y_2,\dotsc,y_{N-1}\in\{0,1\}}(-1)^{z_1y_1}(-1)^{z_2(y_1+y_2)}\dotsb(-1)^{z_{N-2}(y_1+y_2+\dotsb+y_{N-2})}(-1)^{z_{N-1}(y_1+y_2+\dotsb+y_{N-1})}\nonumber\\
        &\qquad\qquad\times \ket{\text{GHZ}_{n_1+1}^{(y_1,\vec{0})}}_{A_{1:n_1}^1R_1}\otimes\ket{\text{GHZ}_{n_2+2}^{(y_2,\vec{0})}}_{R_2A_{1:n_2}^2R_3}\otimes\ket{\text{GHZ}_{n_3+2}^{(y_3,\vec{0})}}_{R_4A_{1:n_3}^3R_5}\otimes\dotsb\nonumber\\
        &\qquad\qquad\dotsb\otimes\ket{\text{GHZ}_{n_{N-1}+2}^{(y_{N-1},\vec{0})}}_{R_{2N-4}A_{1:n_{N-1}}^NR_{2N-3}}\otimes \ket{\text{GHZ}_{n_N+1}^{(y',\vec{0})}}_{R_{2N-2}A_{1:n_N}},
    \end{align}
    where $y'=y_1+y_2+\dotsb+y_{N-1}$. Therefore, the fidelity of the state in \eqref{eq-multpiple_GHZ_merge_linear_protocol_LOCC_chan} is
    \begin{align}
        &\sum_{z_1,z_2,\dotsc,z_{N-1}\in\{0,1\}}\bra{u(\vec{z})}\rho_{A_{1:n_1}^1R_1R_2A_{1:n_2}^2R_3\dotsb R_{2N-4}A_{1:n_{N-1}}^{N-1}R_{2N-3}R_{2N-2}A_{1:n_N}^N}\ket{u(\vec{z})}\\
        &\quad=\frac{1}{2^{N-1}}\sum_{\substack{z_1,z_2,\dotsc z_{N-1}\in\{0,1\}\\y_1,y_2,\dotsc,y_{N-1}\in\{0,1\}\\x_1,x_2,\dotsc,x_{N-1}\in\{0,1\}}}(-1)^{z_1(x_1+y_1)}(-1)^{z_2(y_1+y_2+x_1+x_2)}\dotsb(-1)^{z_{N-2}(y_1+\dotsb+y_{N-2}+x_1+\dotsb+x_{N-2})}(-1)^{z_{N-1}(y_1+\dotsb+y_{N-1}+x_1+\dotsb+x_{N-1})}\nonumber\\
        &\qquad\qquad\times \left(\bra{\text{GHZ}_{n_1+1}^{(y_1,\vec{0})}}_{A_{1:n_1}^1R_1}\otimes\bra{\text{GHZ}_{n_2+2}^{(y_2,\vec{0})}}_{R_2A_{1:n_2}^2R_3}\otimes\bra{\text{GHZ}_{n_3+2}^{(y_3,\vec{0})}}_{R_4A_{1:n_3}^3R_5}\otimes\dotsb\right.\nonumber\\
        &\qquad\qquad\dotsb\otimes\left.\bra{\text{GHZ}_{n_{N-1}+2}^{(y_{N-1},\vec{0})}}_{R_{2N-4}A_{1:n_{N-1}}^NR_{2N-3}}\otimes \bra{\text{GHZ}_{n_N+1}^{(y',\vec{0})}}_{R_{2N-2}A_{1:n_N}}\right)\rho_{A_{1:n_1}^1R_1R_2A_{1:n_2}^2R_3\dotsb R_{2N-4}A_{1:n_{N-1}}^{N-1}R_{2N-3}R_{2N-2}A_{1:n_N}^N}\nonumber\\
        &\qquad\qquad\times \left(\ket{\text{GHZ}_{n_1+1}^{(x_1,\vec{0})}}_{A_{1:n_1}^1R_1}\otimes\ket{\text{GHZ}_{n_2+2}^{(x_2,\vec{0})}}_{R_2A_{1:n_2}^2R_3}\otimes\ket{\text{GHZ}_{n_3+2}^{(x_3,\vec{0})}}_{R_4A_{1:n_3}^3R_5}\otimes\dotsb\right.\nonumber\\
        &\qquad\qquad\dotsb\otimes\left.\ket{\text{GHZ}_{n_{N-1}+2}^{(x_{N-1},\vec{0})}}_{R_{2N-4}A_{1:n_{N-1}}^NR_{2N-3}}\otimes \ket{\text{GHZ}_{n_N+1}^{(x',\vec{0})}}_{R_{2N-2}A_{1:n_N}}\right)\\
        &\quad=\sum_{y_1,y_2,\dotsc,y_{N-1}\in\{0,1\}}\left(\bra{\text{GHZ}_{n_1+1}^{(y_1,\vec{0})}}_{A_{1:n_1}^1R_1}\otimes\bra{\text{GHZ}_{n_2+2}^{(y_2,\vec{0})}}_{R_2A_{1:n_2}^2R_3}\otimes\bra{\text{GHZ}_{n_3+2}^{(y_3,\vec{0})}}_{R_4A_{1:n_3}^3R_5}\otimes\dotsb\right.\nonumber\\
        &\qquad\qquad\dotsb\otimes\left.\bra{\text{GHZ}_{n_{N-1}+2}^{(y_{N-1},\vec{0})}}_{R_{2N-4}A_{1:n_{N-1}}^NR_{2N-3}}\otimes \bra{\text{GHZ}_{n_N+1}^{(y',\vec{0})}}_{R_{2N-2}A_{1:n_N}}\right)\rho_{A_{1:n_1}^1R_1R_2A_{1:n_2}^2R_3\dotsb R_{2N-4}A_{1:n_{N-1}}^{N-1}R_{2N-3}R_{2N-2}A_{1:n_N}^N}\nonumber\\
        &\qquad\qquad\times \left(\ket{\text{GHZ}_{n_1+1}^{(y_1,\vec{0})}}_{A_{1:n_1}^1R_1}\otimes\ket{\text{GHZ}_{n_2+2}^{(y_2,\vec{0})}}_{R_2A_{1:n_2}^2R_3}\otimes\ket{\text{GHZ}_{n_3+2}^{(y_3,\vec{0})}}_{R_4A_{1:n_3}^3R_5}\otimes\dotsb\right.\nonumber\\
        &\qquad\qquad\dotsb\otimes\left.\ket{\text{GHZ}_{n_{N-1}+2}^{(y_{N-1},\vec{0})}}_{R_{2N-4}A_{1:n_{N-1}}^NR_{2N-3}}\otimes \ket{\text{GHZ}_{n_N+1}^{(y',\vec{0})}}_{R_{2N-2}A_{1:n_N}}\right),
    \end{align}
    as required.
\end{proof}

\subsection{Fusion of GHZ states in a tree topology (Protocol~\ref{alg:multiple_GHZ_merge_tree_protocol_main})}\label{sec:GHZ_merge_tree}

Protocol~\ref{alg:multiple_GHZ_merge_linear_protocol} can be used almost identically in the case that the GHZ states are in a tree topology, as shown in Fig.~\ref{fig:GHZ_merging}. (Observe that, in fact, the previous examples of a star topology and a linear topology are both special cases of tree graphs.)  In this case (see Protocol~\ref{alg:multiple_GHZ_merge_tree_protocol_main}), we start by taking the graph and arranging it explicitly as a tree, as in Fig.~\ref{fig:GHZ_merging}. (Note that this can be done by taking an arbitrary node as the root.) Then, we perform CNOT gates ``downward'' through the tree, followed by measuring the target qubits in the Pauli-$Z$ basis. Finally, the measurement outcomes are announced globally. All of the nodes apply Pauli-$X$ corrections based on their level in the tree, with the root node being level 0.

When the input GHZ states in Protocol~\ref{alg:multiple_GHZ_merge_tree_protocol_main} are noisy, the output state after executing Protocol~\ref{alg:multiple_GHZ_merge_tree_protocol_main} has fidelity to the ideal GHZ state given by an expression analogous to the one in \eqref{eq-multiple_GHZ_merge_linear_protocol_fidelity}.

\end{document}